\documentclass[a4paper,12pt]{article}
\usepackage{latexsym,amsmath,amssymb,amsfonts,amsthm,dsfont,enumitem,float,xcolor,natbib,bbm,threeparttable,graphicx}
\usepackage[normalem]{ulem}
\usepackage[ruled]{algorithm2e}
\usepackage[font={small}]{caption}
\usepackage[colorlinks,linkcolor=red,citecolor=blue,urlcolor=blue,breaklinks]{hyperref}

\newtheorem{theorem}{Theorem}
\newtheorem{proposition}[theorem]{Proposition}
\newtheorem{lemma}[theorem]{Lemma}
\newtheorem{cor}[theorem]{Corollary}

\DeclareMathOperator*{\argmax}{argmax}

\def\R#1{(\ref{#1})}

\DeclareMathOperator*{\sargmax}{sargmax}

\title{Sparse principal component analysis via axis-aligned random projections} 

\author{Milana Gataric, Tengyao Wang and Richard J.~Samworth \\ \medskip {\small Statistical Laboratory, University of Cambridge} \\ \medskip{\small $\{$m.gataric,t.wang,r.samworth$\}$@statslab.cam.ac.uk}} 

\begin{document}

\maketitle

\begin{abstract}
We introduce a new method for sparse principal component analysis, based on the aggregation of eigenvector information from carefully-selected axis-aligned random projections of the sample covariance matrix.  Unlike most alternative approaches, our algorithm is non-iterative, so is not vulnerable to a bad choice of initialisation. We provide theoretical guarantees under which our principal subspace estimator can attain the minimax optimal rate of convergence in polynomial time. In addition, our theory provides a more refined understanding of the statistical and computational trade-off in the problem of sparse principal component estimation, revealing a subtle interplay between the effective sample size and the number of random projections that are required to achieve the minimax optimal rate.  Numerical studies provide further insight into the procedure and confirm its highly competitive finite-sample performance. 
\end{abstract}

\section{Introduction}

Principal component analysis (PCA) is one of the most widely-used techniques for dimensionality reduction in Statistics, Image Processing and many other fields.  The aim is to project the data along directions that explain the greatest proportion of the variance in the population.  In the simplest setting where we seek a single, univariate projection of our data, we may estimate this optimal direction by computing the leading eigenvector of the sample covariance matrix.   

Despite its successes and enormous popularity, it has been well-known for a decade or more that PCA breaks down as soon as the dimensionality $p$ of the data is of the same order as the sample size $n$.  More precisely, suppose that 
$X_1,\ldots,X_n \stackrel{\mathrm{iid}}{\sim} N_p(0,\Sigma)$, with $p \geq 2$, are observations from a Gaussian distribution with a spiked covariance matrix $\Sigma = I_p + v_1v_1^\top$ whose leading eigenvector is $v_1 \in \mathcal{S}^{p-1} := \{v \in \mathbb{R}^p: \|v\|_2 = 1\}$, and let $\hat{v}_1$ denote the leading unit-length eigenvector of the sample covariance matrix $\hat\Sigma := n^{-1}\sum_{i=1}^n X_iX_i^{\top}$.  Then \citet{Johnstone2009} and \citet{Paul2007} showed that $\hat{v}_1$ is a consistent estimator of $v_1$, i.e.\ $|\hat{v}_1^\top v_1| \stackrel{p}{\rightarrow} 1$, if and only if $p = p_n$ satisfies $p/n \rightarrow 0$ as $n \rightarrow \infty$.  It is also worth noting that the principal component $v_1$ may be a linear combination of all elements of the canonical basis in $\mathbb{R}^p$, which can often make it difficult to interpret the estimated projected directions \citep{Jolliffe2003}.

To remedy this situation, and to provide additional interpretability to the principal components in high-dimensional settings, \citet{Jolliffe2003} and \citet{Zou2006} proposed Sparse Principal Component Analysis (SPCA).  Here it is assumed that the leading population eigenvectors belong to the $k$-sparse unit ball
\[
\mathcal{B}^{p-1}_0(k) := \biggl\{ v = (v^{(1)},\ldots,v^{(p)})^\top \in \mathcal{S}^{p-1} : \sum_{j=1}^p \mathds{1}_{\{v^{(j)} \neq 0\}}  \leq k\biggr\}
\]
for some $k \in \{1,\ldots,p\}$.  In addition to the easier interpretability, a great deal of research effort has shown that such an assumption facilitates improved estimation performance \citep[e.g.][]{Johnstone2009,Paul2012,Vu2013,Cai2013,Ma2013,Wang2016}.  To give a flavour of these results, let $\mathcal{V}_n$ denote the set of all estimators of $v_1$, i.e.\ the class of Borel measurable functions from $\mathbb{R}^{n \times p}$ to $\mathcal{S}^{p-1}$.  \citet{Vu2013} introduced a class~$\mathcal{Q}$ of sub-Gaussian distributions whose first principal component $v_1$ belongs to $\mathcal{B}^{p-1}_0(k)$ and showed that
\begin{equation}
\label{Eq:Minimax}
\inf_{\tilde{v}_1 \in \mathcal{V}_n} \sup_{Q \in \mathcal{Q}} \mathbb{E}_Q \{1 - (\tilde{v}_1^\top v_1)^2\} \asymp \frac{k\log p}{n},
\end{equation}
where $a_n \asymp b_n$ means $0 < \liminf_{n \rightarrow \infty} |a_n/b_n| \leq \limsup_{n \rightarrow \infty} |a_n/b_n| < \infty$. Thus, consistent estimation is possible in this framework provided only that $k = k_n$ and $p = p_n$ satisfy $(k \log p)/n \rightarrow 0$.  \citet{Vu2013} showed further that this estimation rate is achieved by the natural estimator
\begin{equation}
\label{opt_est}
\hat{v}_1 \in \argmax_{v\in\mathcal{B}^{p-1}_0(k)} v^{\top}  \hat \Sigma v.
\end{equation}

However, results such as~\eqref{Eq:Minimax} do not complete the story of SPCA.  Indeed, computing the estimator defined in \R{opt_est} turns out to be an NP-hard problem \citep[e.g.][]{TillmanPfetsch2014}: the naive approach would require searching through all $\binom{p}{k}$ of the $k\times k$ symmetric submatrices of $\hat \Sigma$, which takes exponential time in $k$.  Therefore, in parallel to the theoretical developments described above, numerous alternative algorithms for SPCA have been proposed in recent years.  For instance, several papers have introduced techniques based on solving the non-convex optimisation problem in~\eqref{opt_est} by invoking an $\ell_1$-penalty \citep[e.g.][]{Jolliffe2003,Zou2006,Shen2008,Witten2009}. Typically, these methods are fast, but lack theoretical performance guarantees. On the other hand, \cite{d'Aspremont2007} propose to compute~\eqref{opt_est} via semidefinite relaxation. This approach and its variants were analysed by \cite{Amini2009}, \citet{Vu2013a}, \citet{Wang2014} and \cite{Wang2016}, and have been proved to achieve the minimax rate of convergence under certain assumptions on the underlying distribution and asymptotic regime, but the algorithm is slow compared to other approaches.  In a separate, recent development, it is now understood that, conditional on a Planted Clique Hypothesis from theoretical computer science, there is an asymptotic regime in which no randomised polynomial time algorithm can attain the minimax optimal rate \citep{Wang2016}.  Various fast, iterative algorithms were introduced by \cite{Johnstone2009}, \cite{Paul2012}, and \cite{Ma2013}; the last of these was shown to attain the minimax rate under a Gaussian spiked covariance model.  We also mention the computationally-efficient combinatorial approaches proposed by \cite{Moghaddam2006} and \cite{d'Aspremont2008} that aim to find solutions to the optimisation problem in \R{opt_est} using greedy methods. 

A common feature to all of the computationally efficient algorithms mentioned above is that they are iterative, in the sense that, starting from an initial guess $\hat{v}^{[0]} \in \mathbb{R}^p$, they refine their guess by producing a finite sequence of iterates $\hat{v}^{[1]},\ldots,\hat{v}^{[T]} \in \mathbb{R}^p$, with the estimator defined to be the final iterate.  A major drawback of such iterative methods is that a bad initialisation may yield a disastrous final estimate.  To illustrate this point, we ran a simple simulation in which the underlying distribution is $N_{400}(0,\Sigma)$, with
\begin{equation}\label{eq:sigma_intro}
 \Sigma = \begin{pmatrix}
           10J_{10} & \\
                  &  8.9 J_{390} + I_{390}
          \end{pmatrix} + 0.01I_{400},
\end{equation}
where $J_q := \mathbf{1}_q\mathbf{1}_q^\top/q \in \mathbb{R}^{q\times q}$ denotes the matrix with each entry equal to $1/q$.  In this example, $v_1 = (\mathbf{1}_{10}^\top,\mathbf{0}_{390}^\top)^\top/\sqrt{10}$, so $k=10$.  Figure~\ref{fig:iterative_algs} shows, for several different SPCA algorithms, different sample sizes and different initialisation methods, the average values of the loss function
\begin{equation}
\label{eq:loss}
L(u,v) := \sin \measuredangle (u,v) = \{1 - (u^\top v)^2\}^{1/2},
\end{equation}
over 100 repetitions of the experiment.  In the top panels of Figure~\ref{fig:iterative_algs}, the initialisation methods used were the default recommendations of the respective authors, namely diagonal thresholding \citep{d'Aspremont2008,Ma2013}, and vanilla PCA \citep{Zou2006,Shen2008,Witten2009}.  We note that the consistency of diagonal thresholding relies on a spiked covariance structure, which is violated in this example.  In the middle panels of Figure~\ref{fig:iterative_algs}, we ran the same algorithms with $10$ independent initialising vectors chosen uniformly at random on $\mathcal{S}^{p-1}$, and selected the solution $\hat{v}$ from these $10$ that maximises $v \mapsto v^\top\hat{\Sigma}v$.  The main observation is that each of the previously proposed algorithms mentioned above produces very poor estimates, with some almost orthogonal to the true principal component!  The reason for this is that all of the default initialisation procedures are unsuccessful in finding a good starting point. For some methods, this problem may be fixed by increasing the number of random initialisations, but it may take an enormous number of such random restarts (and consequently a very long time) to achieve this. We demonstrate this in the bottom panels of Figure~\ref{fig:iterative_algs}, where for $n=350$ (left) and $n=2000$ (right), we plot the logarithm of the average loss as time increases through the number of random restarts.  As an alternative method, in the top and middle panels of Figure~\ref{fig:iterative_algs}, we also present the corresponding results for \citet{Wang2016}'s variant of the semi-definite programming (SDP) algorithm introduced by \citet{d'Aspremont2007}.  This method is guaranteed to converge from any initialisation, so does not suffer the same poor performance as mentioned above.  However, SDP took even longer to reach algorithmic convergence than any of the alternative approaches, so that in the setting of the bottom panels of Figure~\ref{fig:iterative_algs}, it finally reached a logarithmic average loss of around $-4$ (left panel) and $-5.9$ (right panel) after an average time of $e^8 \approx 3000$ seconds (left panel) and $e^{9.25} \approx 10000$ seconds (right panel); this slow running time means it does not appear in the bottom panels of the figure.  We refer to Section~\ref{ss:comparison} for further comparisons using different examples.  

\begin{figure}[htbp]
    \centering
  \begin{tabular}{c}
  \includegraphics[scale=0.5]{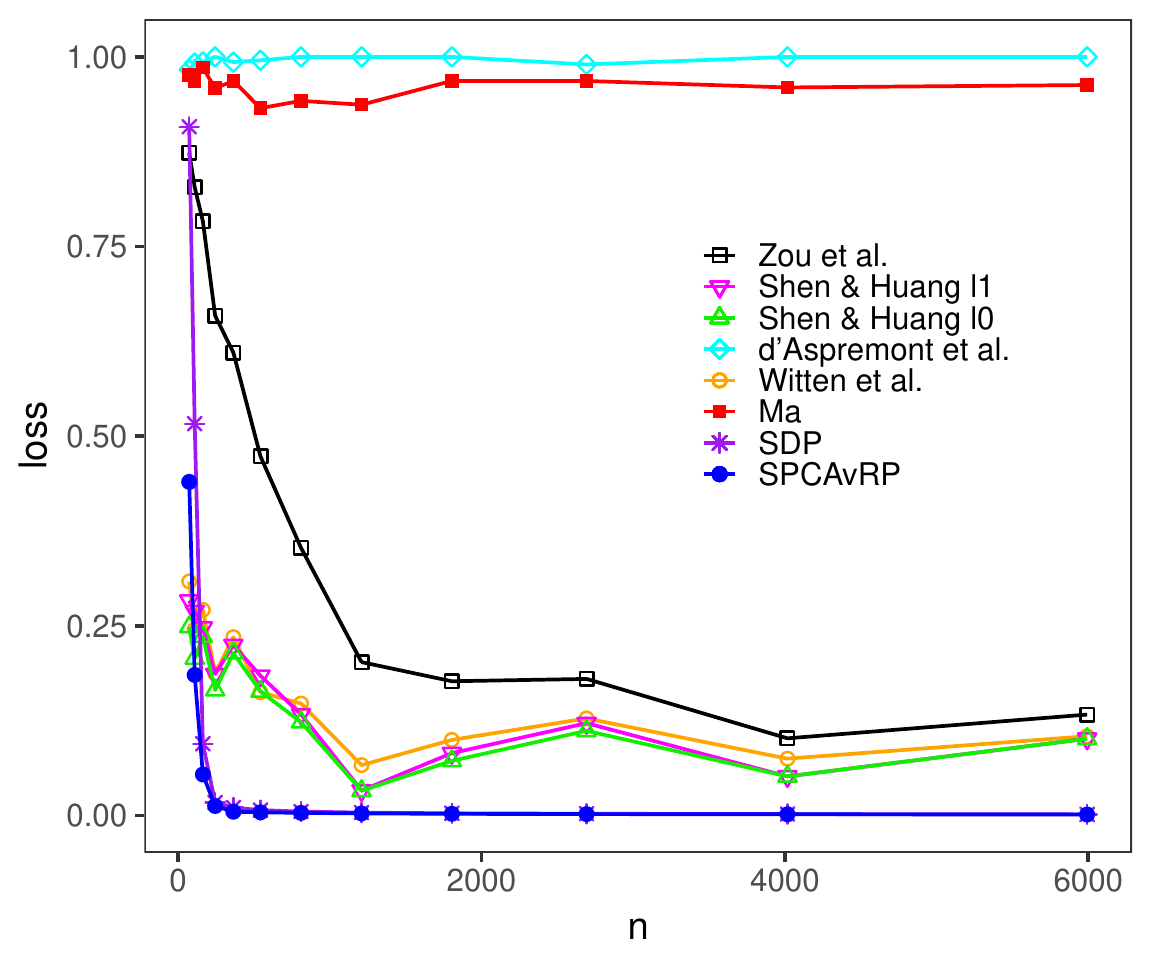} \includegraphics[scale=0.5]{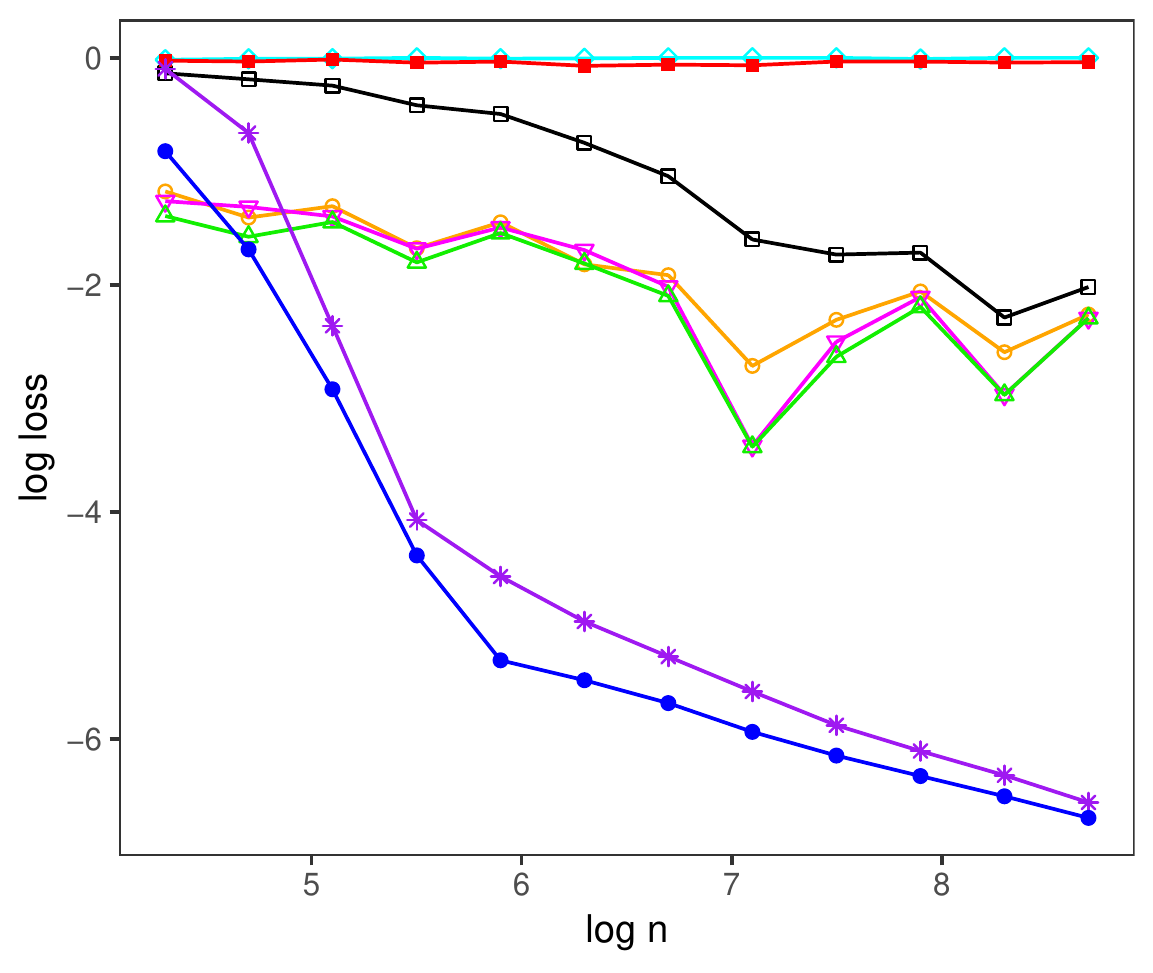} \\
  \includegraphics[scale=0.5]{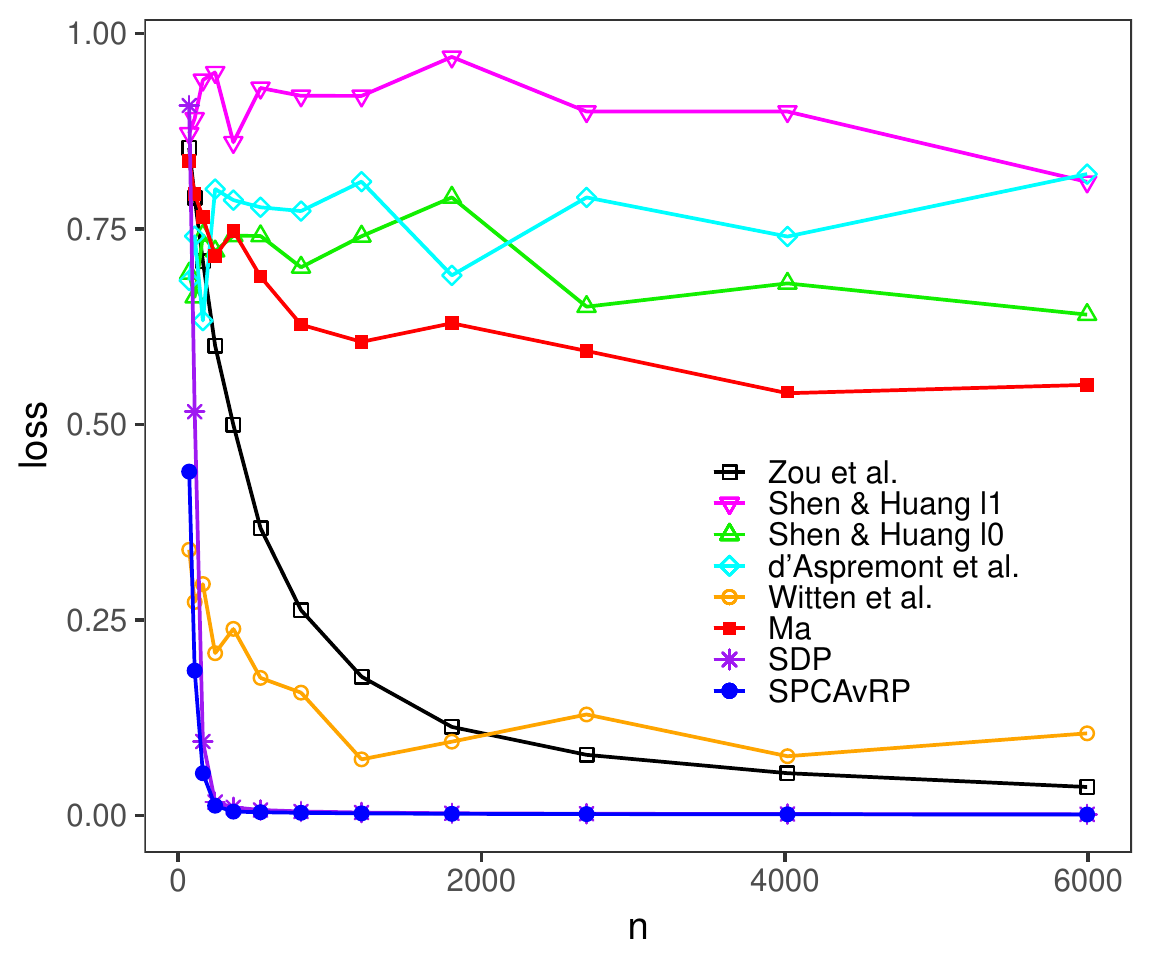} \includegraphics[scale=0.5]{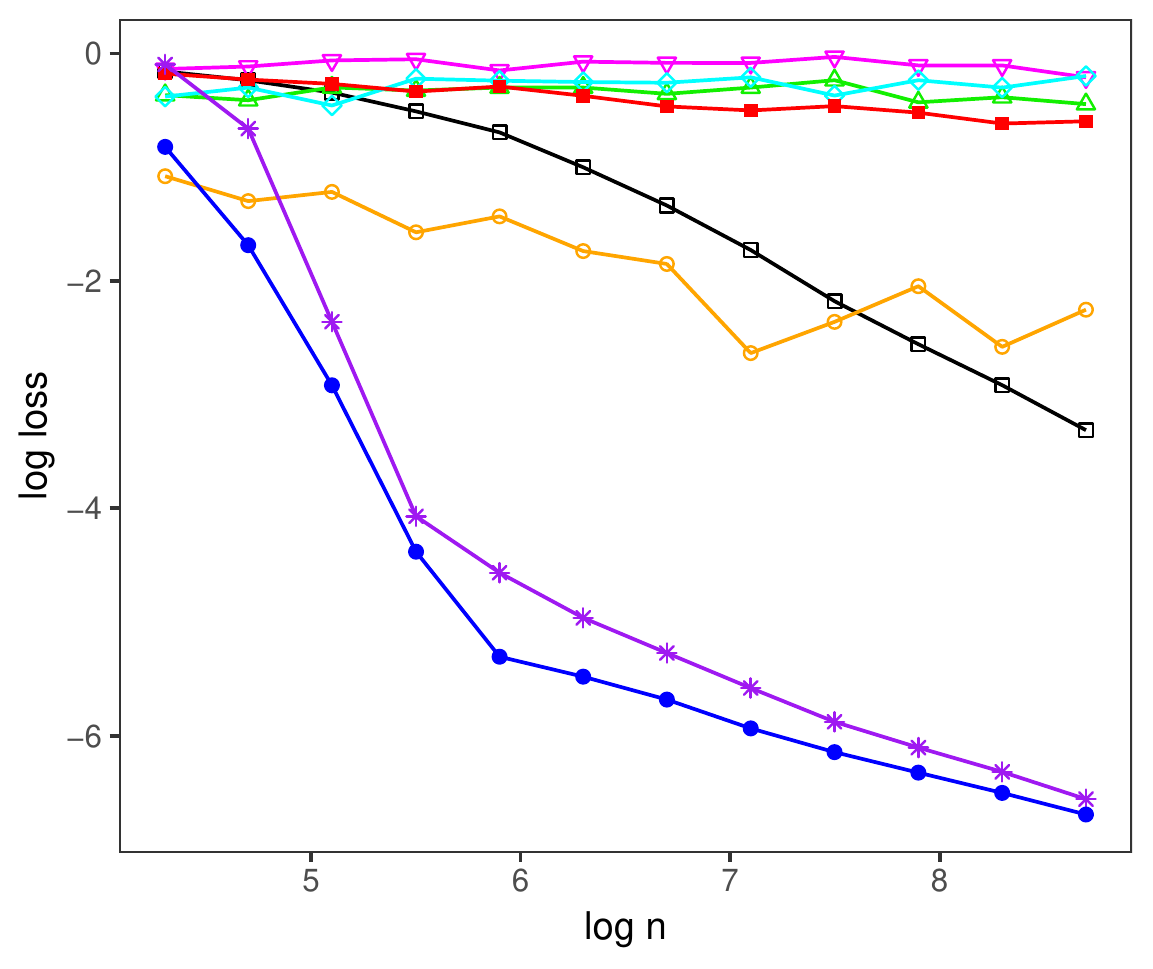} \\
  \includegraphics[scale=0.5]{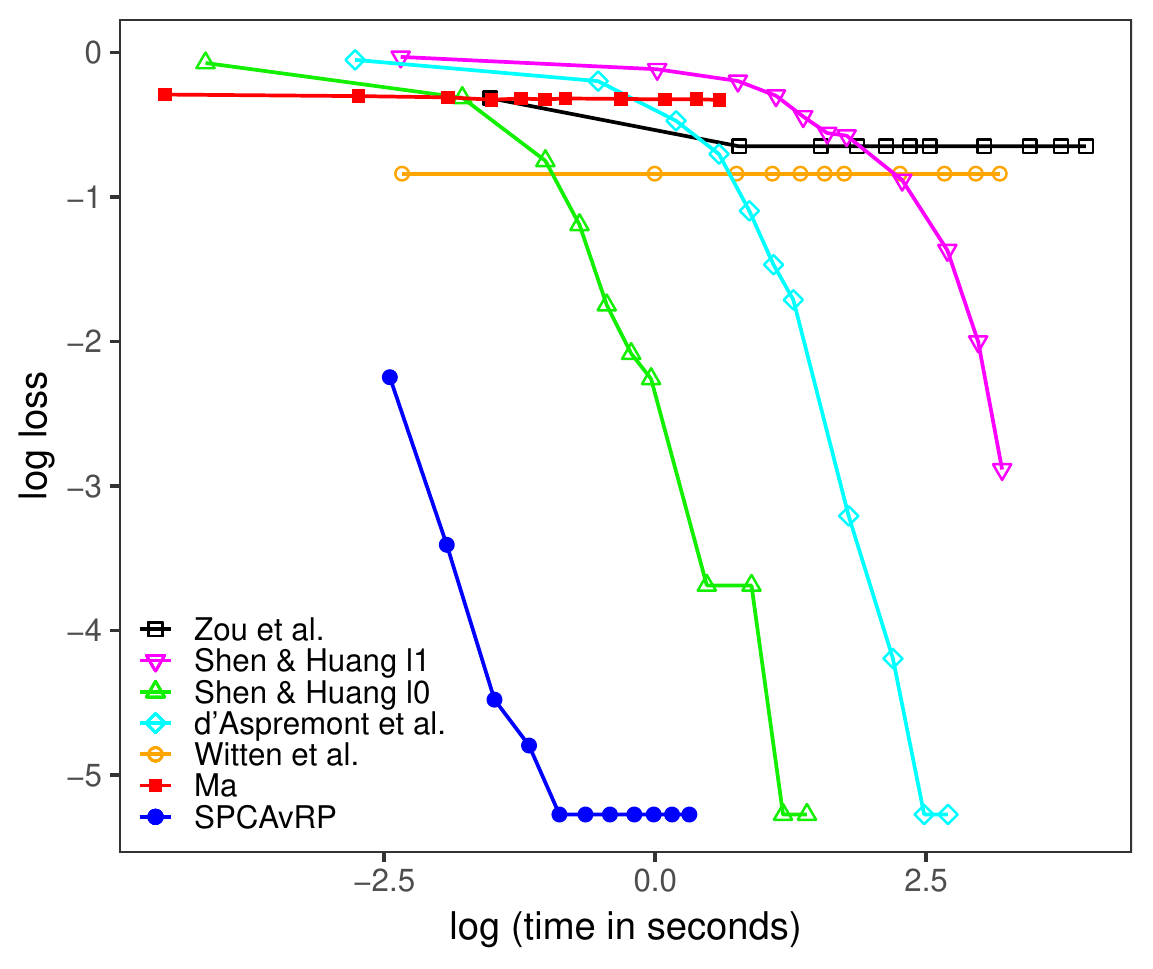}  \includegraphics[scale=0.5]{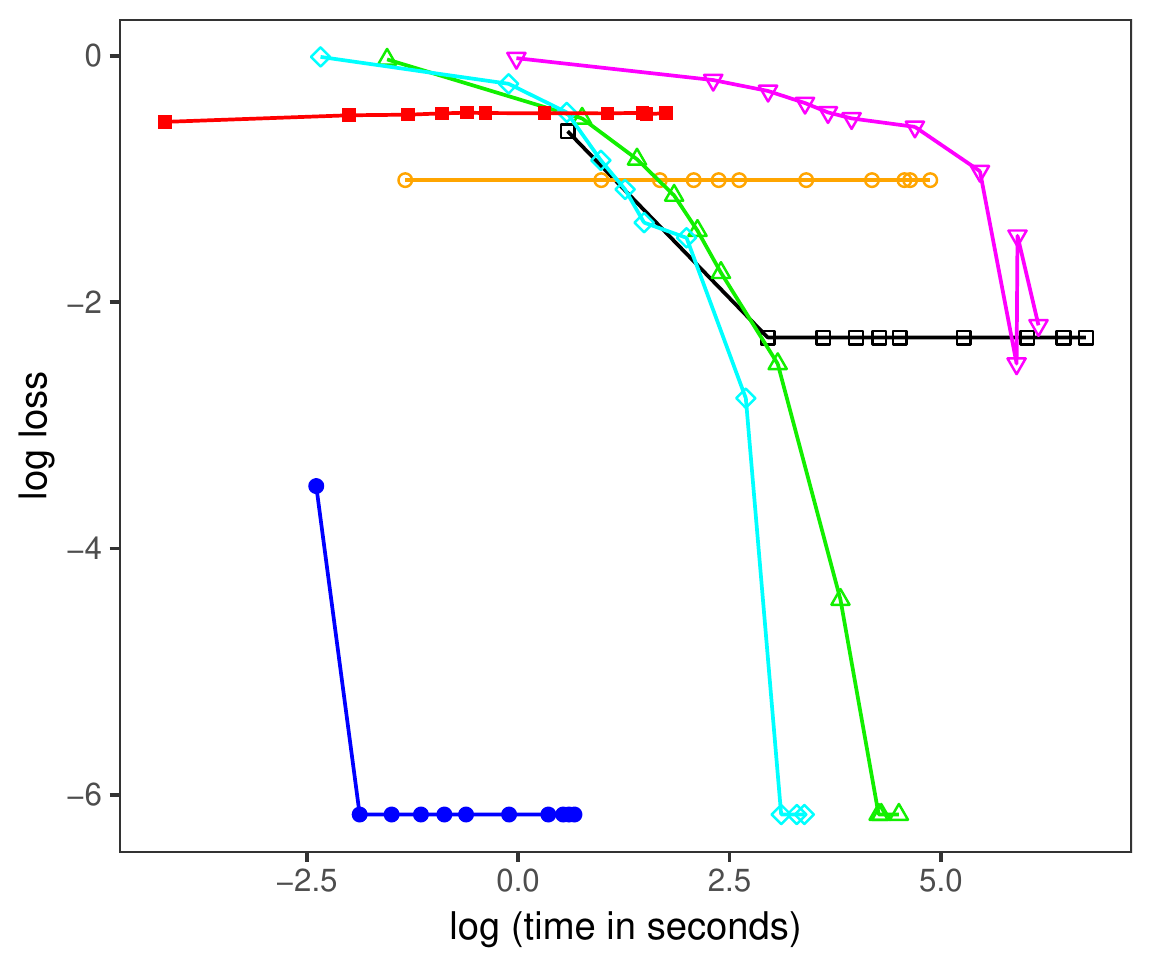}
  \end{tabular}
  \caption{\textbf{Comparison of different approaches using covariance model~\eqref{eq:sigma_intro}.} Top and middle panels: Average loss~\eqref{eq:loss} for different sample sizes $n$, on normal scale (left) and log-log scale (right); top: default initialisation; middle: best of $10$ random initialisations. Bottom panels: Average loss against time in seconds on the log-log scale when $n=350$ (left) and $n=2000$ (right); we vary the number of random projections ($A\in(50,200)$ and $B=\lceil A/2 \rceil$) for SPCAvRP and the number of random initialisations (from $1$ to $250$) for other iterative competing methods.  Blue: the SPCAvRP algorithm proposed in this paper; purple: SDP; red: \cite{Ma2013}; orange:  \cite{Witten2009}; cyan: \cite{d'Aspremont2008}; magenta and green: \cite{Shen2008} with $\ell_1$ and $\ell_0$-thresholding respectively; black: \cite{Zou2006}.}
\label{fig:iterative_algs}
\end{figure}

In Section~\ref{s:SPCAvRP} of this paper, we propose a novel algorithm for SPCA that aggregates estimates over carefully-chosen axis-aligned random projections of the data into a lower-dimensional space.  In contrast to the other algorithms mentioned above, it is non-iterative and does not depend on a choice of  initialisation, so it has no difficulty with the simulation example above.  Indeed, from the blue curve in Figure~\ref{fig:iterative_algs}, we see that it 
outperforms even the SDP algorithm, compared to which it was over 7000 times faster in the $n=2000$ case.

Our algorithm, which we refer to as SPCAvRP, turns out to be attractive for both theoretical and computational reasons.  
From a theoretical point of view, our algorithm provides a new perspective on the statistical and computational trade-off involved in the SPCA problem.  As we show in Section~\ref{s:guarantees}, when the effective sample size is large, the SPCAvRP procedure can attain the minimax optimal rate with a number of projections that grows only polynomially in the problem parameters.  On the other hand, if one were to use a number of random projections exponentially large in $k$, SPCAvRP could even achieve this minimax rate in a much smaller effective sample size regime. Although this exponentially large number of projections may seem discouraging, we emphasise it is in fact not a drawback of the SPCAvRP algorithm, but simply a reflection of the fundamental difficulty of the problem in this effective sample size regime. Indeed, \citet{Wang2016} established a computational lower bound, which reveals that no randomised polynomial time algorithm can attain the minimax rate of convergence for these effective sample sizes.  The elucidation of the transition from polynomial to exponentially large number of projections is an illustration of the fascinating fundamental statistical and computational trade-off in this problem.
The computational attractions of the proposed algorithm include the fact that it is highly scalable due to easy parallelisation, and does not even require computation of $\hat \Sigma \in \mathbb{R}^{p\times p}$, since it suffices to extract principal submatrices of $\hat{\Sigma}$, which can be done by computing the sample covariance matrices of the projected data.  This may result in a significant computational saving if $p$ is very large.  Several numerical aspects of the algorithm, including a finite-sample simulation comparison with alternative methods on both simulated and real data, are considered in Section~\ref{s:numerics}.  These reveal that our SPCAvRP algorithm has very competitive performance, and furthermore, it enjoys robustness properties that iterative algorithms do not share.  The proofs of all of our results are given in Section~\ref{ap:proofs}.

Algorithms based on random projections have recently been shown to be highly effective for several different problems in high-dimensional statistical inference.  For instance, in the context of high-dimensional classification, \cite{Cannings2017} showed that their random projection ensemble classifier that aggregates over projections that yield small estimates of the test error can result in excellent performance.  \cite{Marzetta2011} employ an ensemble of random projections to construct an estimator of the population covariance matrix and its inverse in the setting where $n<p$.  \cite{Fowler2009} introduced a so-called compressive-projection PCA that reconstructs the sample principal components from many low-dimensional projections of the data.  Finally, to decrease the computational burden of classical PCA, \cite{Qi2012} and \cite{Anaraki2014} propose estimating $v_1(\Sigma)$ by the leading eigenvector of $n^{-1}\sum_{i=1}^{n} P_iX_iX_i^\top P_i$, where $P_1,\ldots,P_n$ are random projections of a particular form.

\paragraph{Notation.} We conclude this introduction with some notation used throughout the paper. For $r \in \mathbb{N}$, let $[r] := \{1,\ldots,r\}$.   For a vector $u \in \mathbb{R}^p$, we write $u^{(j)}$ for its $j$th component and $\|u\|_2 :=\bigl\{\sum_{j=1}^{p} (u^{(j)})^2\bigr\}^{1/2}$ for its  Euclidean norm.  
For a real symmetric matrix $U\in\mathbb{R}^{p\times p}$, let $\lambda_1(U) \geq \lambda_2(U) \geq \ldots \geq \lambda_p(U)$ denote its eigenvalues, arranged in decreasing order, and let $v_1(U),\ldots,v_p(U)$ denote the corresponding eigenvectors. 
In addition, for $m \in [p]$, we write $V_m(U):=(v_1(U),\ldots,v_m(U))$ for the $p \times m$ matrix whose columns are the $m$ leading eigenvectors of $U$.  
In the special case where $U = \Sigma$, we drop the argument, and write 
$\lambda_r = \lambda_r(\Sigma)$, $v_r = v_r(\Sigma)$ and $V_m = V_m(\Sigma)$.  
For a general $U \in \mathbb{R}^{p \times m}$, we define $U^{(j,j')}$ to be the $(j,j')$th entry of $U$, and $U^{(j,\cdot)}$ the $j$th row of $U$, regarded as a column vector.  Given $S \subseteq [p]$ and $S' \subseteq [m]$, we write $U^{(S,S')}$ for the $|S| \times |S'|$ matrix obtained by extracting the rows of $U$ indexed by $S$ and columns indexed by $S'$; we also write $U^{(S,\cdot)} := U^{(S,[m])}$.  We write $\|U\|_{\mathrm{op}} := \sup_{x \in \mathcal{S}^{m-1}} \|Ux\|_2$ and $\|U\|_{\mathrm{F}}:=\bigl\{\sum_{j=1}^{p}\sum_{j'=1}^{m}|U^{(j,j')}|^2\bigr\}^{1/2}$ for the operator and Frobenius norms of $U$ respectively.  We denote the set of real orthogonal $p\times p$ matrices by $\mathbb{O}_p$ and the set of real $p\times m$ matrices with orthonormal columns by $\mathbb{O}_{p,m}$. For matrices $U,V\in\mathbb{O}_{p,m}$, we define the loss function
$$
L(U,V):=\|\sin\Theta(U,V)\|_{\mathrm{F}},
$$
where $\Theta(U,V)$ is the $m\times m$ diagonal matrix whose $j$th diagonal entry is the $j$th principal angle between $U$ and $V$, i.e.~$\arccos \sigma_j$, where $\sigma_j$ is the $j$th singular value of $U^{\top}V$. Observe that this loss function reduces to~\eqref{eq:loss} when $m=1$. 

For any index set $J\subseteq[p]$  we write $P_{J}$ to denote the projection onto the span of $\{e_j: j\in J\}$, where $e_1,\ldots,e_p$ are the standard Euclidean basis vectors in $\mathbb{R}^p$, so that $P_J$ is a $p \times p$ diagonal matrix whose $j$th diagonal entry is $\mathds{1}_{\{j \in J\}}$.  Finally, for $a,b \in \mathbb{R}$, we write $a \lesssim b$ to mean that there exists a universal constant $C > 0$ such that $a \leq Cb$.  

\section{SPCA via random projections}\label{s:SPCAvRP}

\subsection{Single principal component estimation}

In this section, we describe our algorithm for estimating a single principal component $v_1$ in detail; more general estimation of multiple principal components $v_1,\ldots,v_m$ is treated in Section~\ref{Sec:Multiple} below.  Let $x_1,\ldots,x_n$ be data points in $\mathbb{R}^p$ and let $\hat{\Sigma} := n^{-1}\sum_{i=1}^n x_ix_i^\top$.  We think of $x_1,\ldots,x_n$ as independent realisations of a mean zero random vector $X$, so a practitioner may choose to centre each variable so that $\sum_{i=1}^n x_i^{(j)} = 0$ for each $j \in [p]$.  For $d \in [p]$, let $\mathcal{P}_d := \{P_S:S \subseteq [p], \ |S| = d\}$ denote the set of $d$-dimensional, axis-aligned projections. 
For fixed $A, B \in \mathbb{N}$, consider projections $\{P_{a,b}: a \in [A], b \in [B]\}$ independently and uniformly distributed on $\mathcal{P}_d$.  We think of these projections as consisting of $A$ groups, each of cardinality $B$.  For each $a \in [A]$, let  
$$
b^*(a) := \sargmax_{b\in[B]}  \lambda_1( P_{a,b} \hat{\Sigma} P_{a,b})
$$
denote the index of the selected projection within the $a$th group, where $\sargmax$ denotes the smallest element of the $\argmax$ in the lexicographic ordering.  The idea is that the non-zero entries of $P_{a,b^*(a)} \hat{\Sigma} P_{a,b^*(a)}$ form a principal submatrix of $\hat{\Sigma}$ that should have a large leading eigenvalue, so the non-zero entries of the corresponding leading eigenvector $\hat{v}_{a,b^*(a);1}$ of $P_{a,b^*(a)} \hat{\Sigma} P_{a,b^*(a)}$ should have some overlap with those of~$v_1$.  Observe that, if $d=k$ and $\{P_{a,b}: b \in [B]\}$ were to contain all $\binom{p}{k}$ projections, then the leading eigenvector of $P_{a,b^*(a)} \hat{\Sigma} P_{a,b^*(a)}$ would yield the minimax optimal estimator in \eqref{opt_est}.  Of course, it would typically be too computationally expensive to compute all such projections, so instead we only consider $B$ randomly chosen ones.  

The remaining challenge is to aggregate over the selected projections.  To this end, for each coordinate $j \in [p]$, we compute an importance score $\hat{w}^{(j)}$, defined as an average over $a \in [A]$ of the squared $j$th components of the selected eigenvectors $\hat{v}_{a,b^*(a);1}$, weighted by the eigengap $\lambda_1\bigl( P_{a,b^*(a)} \hat{\Sigma} P_{a,b^*(a)}\bigr) - \lambda_2\bigl( P_{a,b^*(a)} \hat{\Sigma} P_{a,b^*(a)} \bigr)$.  This means that we take account not just of the frequency with which each coordinate is chosen, but also their corresponding magnitudes in the selected eigenvector, as well as an estimate of the signal strength. Finally, we select the $\ell$ indices $\hat{S}$ corresponding to the largest values of $\hat{w}^{(1)},\ldots,\hat{w}^{(p)}$ and output our estimate $\hat{v}_1$ as the leading eigenvector of $P_{\hat S} \hat{\Sigma} P_{\hat S}$.  Pseudo-code for our SPCAvRP algorithm is given in Algorithm~\ref{Algo:SPCAvRP}.


\begin{algorithm}[htbp]
\SetAlgoLined
\IncMargin{1em}
\DontPrintSemicolon
\KwIn{
$x_1,\ldots,x_n \in \mathbb{R}^p$, $A, B \in \mathbb{N}$, $d, \ell \in [p]$.
}
Generate $\{P_{a,b}: a \in [A], b \in [B]\}$ independently and uniformly from $\mathcal{P}_d$.\;
Compute $\{P_{a,b}\hat{\Sigma}P_{a,b}:a \in [A], b \in [B]\}$, where $\hat{\Sigma} := n^{-1}\sum_{i=1}^n x_ix_i^\top$.\;
\For{$a=1,\ldots,A$}{
      \For{$b=1,\ldots,B$}{
      Compute $\hat\lambda_{a,b;1} := \lambda_1( P_{a,b} \hat{\Sigma} P_{a,b} )$, $\hat\lambda_{a,b;2} := \lambda_2( P_{a,b} \hat{\Sigma} P_{a,b} )$ and $\hat{v}_{a,b;1} \in v_1( P_{a,b} \hat{\Sigma} P_{a,b} )$.
      }   
      Compute
      \[
       b^*(a) := \sargmax_{b\in[B]}  \hat\lambda_{a,b;1}.
      \]
} 
Compute $\hat w=(\hat w^{(1)},\ldots,\hat w^{(p)})^\top$, where
\begin{equation}\label{alg:aggregation}
\hat w^{(j)} :=  \frac{1}{A} \sum_{a=1}^{A} (\hat\lambda_{a,b;1} - \hat\lambda_{a,b;2})\bigl(\hat v_{a,b^*(a);1}^{(j)}\bigr)^2,
\end{equation}
and let $\hat S\subseteq[p]$ be the index set of the $\ell$ largest components of $\hat w$.  \;
\KwOut{$\hat v_1 := \sargmax_{v\in\mathcal{S}^{p-1}} v^{\top}  P_{\hat S} \hat{\Sigma} P_{\hat S} v$.}
\caption{Pseudo-code for the SPCAvRP algorithm for a single principal component}
\label{Algo:SPCAvRP}
\end{algorithm}

Besides the intuitive selection of the most important coordinates, the use of axis-aligned projections facilitates faster computation as opposed to the use of general orthogonal projections.  Indeed, the multiplication of $\hat{\Sigma} \in \mathbb{R}^{p \times p}$ by an axis-aligned projection $P \in \mathcal{P}_d$ from the left (or right) can be recast as the selection of $d$ rows (or columns) of $\hat{\Sigma}$ corresponding to the indices of the non-zero diagonal entries of $P$.  Thus, instead of the typical $\mathcal{O}(p^2d)$ matrix multiplication complexity, only $\mathcal{O}(pd)$ operations are required.  We also remark that, instead of storing $P$, it suffices to store its non-zero indices. 

More generally, the computational complexity of Algorithm~\ref{Algo:SPCAvRP} can be analysed as follows. Generating $AB$ initial random projections takes $\mathcal{O}({A}{B}d)$ operations.  Next, we need to compute $P_{{a},{b}}\hat \Sigma P_{{a},{b}}$ for all ${a}$ and ${b}$, which can be done in two different ways.  One option is to compute~$\hat \Sigma$, and then for each projection $P_{{a},{b}}$ select the corresponding $d\times d$ principal submatrix of~$\hat\Sigma$, which requires $\mathcal{O}(np^2+{A}{B}d^2)$ operations.  Alternatively, we can avoid computing $\hat{\Sigma}$ by computing the sample covariance matrix of the projected data $\{P_{{a},{b}} x_1,\ldots,P_{{a},{b}} x_n:a \in [A],b \in [B]\}$, which has $\mathcal{O}({A}{B}nd^2)$ complexity.  If $p^2 \gg {A}{B}d^2$, then the second option is preferable. 

The rest of Algorithm~\ref{Algo:SPCAvRP} entails computing an eigendecomposition of each $d\times d$ matrix, and computing $\{b^*({a}) : a \in [A]\}$, $\hat w$, $\hat S$, and $\hat v_1$, which altogether amounts to $\mathcal{O}({A}{B} d^3 + Ap + \ell^3)$ operations. Thus, assuming that $n\geq d$, the overall computational complexity of the SPCAvRP algorithm  is 
$$
\mathcal{O}(\min\{np^2 + {A}{B} d^3 + Ap + \ell^3, {A}{B}nd^2 + Ap+ \ell^3\} ).
$$
We also note that, due to the use of random projections, the algorithm is highly parallelisable. In particular, both for-loops of Algorithm~\ref{Algo:SPCAvRP} can be parallelised, and the selection of good projections can easily be carried out using different (up to ${A}$) machines.

Finally, we note that the numbers $A$ and $B$ of projections, the dimension $d$ of those projections and the sparsity $\ell$ of the final estimator, need to be provided as inputs to Algorithm~\ref{Algo:SPCAvRP}.  The effect of these parameter choices on the theoretical guarantees of our SPCAvRP algorithm is elucidated in our theory in Section~\ref{s:guarantees}, while their practical selection is discussed in Section~\ref{ss:inputparameters}.

\subsection{Multiple principal component estimation}
\label{Sec:Multiple}


The estimation of higher-order principal components is typically achieved via a deflation scheme.  Having computed estimates $\hat{v}_1,\ldots,\hat{v}_{r-1}$ of the top $r-1$ principal components, the aim of such a procedure is to estimate the $r$th principal component based on modified observations, which have had their correlation with these previously-estimated components removed \citep[e.g.][]{Mackey2009}.  For any matrix $V\in\mathbb{R}^{p\times r}$ of full column rank, we define the projection onto the orthogonal complement of the column space of $V$ by $\text{Proj}^{\perp}(V) := I_p - V(V^\top V)^{-1}V^\top$ if $V\neq 0$ and $I_p$ otherwise. Then writing $\hat{V}_{r-1}:=(\hat{v}_1,\ldots,\hat{v}_{r-1})$, one possibility to implement a deflation scheme is to set $\tilde{x}_i := \text{Proj}^{\perp}(\hat{V}_{r-1}) x_i$ for $i\in [n]$.  Note that in sparse PCA, by contrast with classical PCA, the estimated principal components from such a deflation scheme are typically not orthogonal.  In Algorithm~\ref{Algo:SPCAvRPdefl}, we therefore propose a modified deflation scheme, which in combination with Algorithm~\ref{Algo:SPCAvRP} can be used to compute an arbitrary $m \in [p]$ principal components that are orthogonal (as well as sparse), as verified in Lemma~\ref{Lemma:SparseOrthogonal} below.
\begin{lemma}
  \label{Lemma:SparseOrthogonal}
For any $m \in [p]$, the outputs $\hat{v}_1,\ldots,\hat{v}_m$ of Algorithm~\ref{Algo:SPCAvRPdefl} are mutually orthogonal.
\end{lemma}
We remark that, in fact, our proposed deflation method can be used in conjunction with any SPCA algorithm.

\begin{algorithm}[htbp]
\SetAlgoLined
\IncMargin{1em}
\DontPrintSemicolon
\KwIn{
$x_1,\ldots, x_n \in \mathbb{R}^{p}$, $A, B\in \mathbb{N}$, $m, d, \ell_1, \ldots, \ell_m \in [p]$.
}
Let $\hat v_1$ be output of Algorithm~\ref{Algo:SPCAvRP} with inputs $x_1,\ldots, x_n$, $A$, $B$, $d$ and $\ell_1$.\;
\For{$r=2,\ldots,m$}{
  Let $H_{r} := \text{Proj}^{\perp}(\hat{V}_{r-1})$, where $\hat{V}_{r-1} := (\hat{v}_1, \ldots, \hat{v}_{r-1})$.\;
  Let $\tilde{v}_r$ be output of Algorithm~\ref{Algo:SPCAvRP} with inputs $H_{r}x_1,\ldots,H_{r}x_n$, $A$, $B$, $d$ and $\ell_r$.\;

  Let  $\tilde{S}_r := \{j \in [p]: \tilde{v}_r^{(j)} \neq 0\}$ and $H_{\tilde{S}_r}:= \text{Proj}^{\perp}(P_{\tilde{S}_r} \hat{V}_{r-1})$.\;
  Compute
    $$\hat v_r := v_1 \bigl( H_{\tilde S_r}  P_{\tilde S_r} \hat{\Sigma}  P_{\tilde S_r} H_{\tilde S_r} \bigr).$$
}
\KwOut{$\hat v_1, \ldots, \hat v_m$.}
\caption{Pseudo-code of the modified deflation scheme}
\label{Algo:SPCAvRPdefl}
\end{algorithm}

Although Algorithm \ref{Algo:SPCAvRPdefl} can conveniently be used to compute sparse principal components up to order~$m$, it requires Algorithm~\ref{Algo:SPCAvRP} to be executed $m$ times. Instead, we can modify  Algorithm~\ref{Algo:SPCAvRP} to estimate directly the leading eigenspace of dimension $m$ --- the subspace spanned by the columns of matrix $V_m=(v_1,\ldots,v_m)$ --- at a computational cost not much higher than that of executing Algorithm~\ref{Algo:SPCAvRP} only once. To this end, we propose a generalisation of the SPCAvRP algorithm for eigenspace estimation in Algorithm~\ref{Algo:SPCAvRPsub}.  In this generalisation, $A$ projections are selected from a total of $A\times B$ random projections, by computing
$$
b^*(a) := \sargmax_{b\in[B]}  \sum_{r=1}^m \lambda_r( P_{a,b} \hat{\Sigma} P_{a,b})
$$
for each $a\in[A]$.  
We can regard $\sum_{r=1}^m \bigl(\hat\lambda_{a,b^*(a);r}-\hat\lambda_{a,b^*(a);m+1}\bigr)\bigl(\hat v_{a,b^*(a);r}^{(j)}\bigr)^2$ as the contribution of the $a$th selected projection to the importance score of the $j$th coordinate, and, analogously to the single component estimation case, we average these contributions over $a \in [A]$ to obtain a vector of final importance scores.  Again, similar to the case $m=1$, we then threshold the top $\ell$ importance scores to obtain a final projection and our $m$ estimated principal components.  A notable difference, then, between Algorithm~\ref{Algo:SPCAvRPsub} and the deflation scheme (Algorithm~\ref{Algo:SPCAvRPdefl}) is that now we estimate the union of the supports of the leading $m$ eigenvectors of $\Sigma$ simultaneously rather than one at a time.  A consequence is that Algorithm~\ref{Algo:SPCAvRPsub} is particularly well suited to a sparsity setting known in the literature as `row sparsity' \citep{Vu2013}, where leading eigenvectors of interest may share common support, because it borrows strength regarding the estimation of this support from the simultaneous nature of the multiple component estimation.  On the other hand, Algorithm~\ref{Algo:SPCAvRPdefl} may have a slight advantage in cases where the leading eigenvectors have disjoint supports; see Section~\ref{sss:comparison-multi} for further discussion.


Observe that for $m=1$, both Algorithm~\ref{Algo:SPCAvRPdefl} and Algorithm~\ref{Algo:SPCAvRPsub} reduce to Algorithm~\ref{Algo:SPCAvRP}. Furthermore, for any $m$, up to the step where $\hat{w}$ is computed, Algorithm~\ref{Algo:SPCAvRPsub} has the same complexity as Algorithm~\ref{Algo:SPCAvRP}, with the total complexity of Algorithm~\ref{Algo:SPCAvRPsub} amounting to $\mathcal{O}(\min\{np^2+ABd^3 + Amp + \ell^3, ABnd^2+ Amp + \ell^3\})$, provided that $n\geq d$.

\begin{algorithm}[htbp]
\SetAlgoLined
\IncMargin{1em}
\DontPrintSemicolon
\KwIn{
$x_1,\ldots,x_n \in \mathbb{R}^p$, $A, B \in \mathbb{N}$, $d,\ell \in [p]$, $m \in [d]$.
}
Generate $\{P_{a,b}: a \in [A], b \in [B]\}$ independently and uniformly from $\mathcal{P}_d$.\;
Compute $\{P_{a,b}\hat{\Sigma}P_{a,b}:a \in [A], b \in [B]\}$, where $\hat{\Sigma} := n^{-1}\sum_{i=1}^n x_ix_i^\top$.\;
\For{$a=1,\ldots,A$}{
      \For{$b=1,\ldots,B$}{
		  For $r\in[m+1]$, compute $\hat\lambda_{a,b;r}:= \lambda_r( P_{a,b} \hat{\Sigma} P_{a,b})$ and the corresponding eigenvector $\hat{v}_{a,b;r}$, with the convention that $\hat\lambda_{a,b;d+1}:=0$.
      }
      Compute $b^*(a) := \sargmax_{b\in[B]}  \sum_{r=1}^m \hat\lambda_{a,b;r}$.
} 
Compute $\hat w = (\hat w^{(1)}, \ldots, \hat w^{(p)})^\top$ with
\[
\hat w^{(j)} :=  \frac{1}{A} \sum_{a=1}^{A} \sum_{r=1}^m \bigl(\hat\lambda_{a,b^*(a);r}-\hat\lambda_{a,b^*(a);m+1}\bigr)\bigl(\hat v_{a,b^*(a);r}^{(j)}\bigr)^2.
\]
Let $\hat S\subseteq[p]$ be the index set of the $\ell$ largest components of $\hat w$.  \;
\KwOut{$\hat V_m=(\hat v_1, \ldots, \hat v_m)$, where $\hat v_1, \ldots, \hat v_m$ are the principal eigenvectors of $P_{\hat S} \hat{\Sigma} P_{\hat S}$.}
\caption{Pseudo-code of the SPCAvRP algorithm for eigenspace estimation}
\label{Algo:SPCAvRPsub}
\end{algorithm}

\section{Theoretical guarantees}\label{s:guarantees}

In this section, we focus on the general Algorithm~\ref{Algo:SPCAvRPsub}.
We assume that $X_1,\ldots,X_n$ are independently sampled from a distribution $Q$ satisfying a Restricted Covariance Concentration (RCC) condition introduced in \citet{Wang2016}. Recall that, for $K > 0$, we say that a mean zero distribution $Q$ on $\mathbb{R}^p$ satisfies an RCC condition with parameter $K$, and write $Q\in\mathrm{RCC}_p(K)$, if for all $\delta>0$, $n \in \mathbb{N}$ and $r\in [p]$, we have 
\begin{equation}
\label{Eq:RCC}
 \mathbb{P}\biggl\{\sup_{u\in \mathcal{B}^{p-1}_0(r)} \bigl|u^\top (\hat\Sigma-\Sigma) u\bigr| \geq K \max\biggl(\sqrt\frac{r\log (p/\delta)}{n}\, ,\, \frac{r\log (p/\delta)}{n}\biggr)\biggr\}\leq \delta.
\end{equation}
In particular, if $Q = N_p(0,\Sigma)$, then $Q\in\mathrm{RCC}_p\bigl(8\lambda_1(1+9/\log{p})\bigr)$; and if $Q$ is sub-Gaussian with parameter $\sigma^2$, in the sense that $\int_{\mathbb{R}^p} e^{u^\top x}\,dQ(x)\leq e^{\sigma^2\|u\|_2^2/2}$ for all $u\in\mathbb{R}^p$, then $Q\in\mathrm{RCC}_p\bigl(16\sigma^2(1+9/\log{p})\bigr)$ \citep[Proposition~1]{Wang2016}.

As mentioned in Section~\ref{Sec:Multiple}, our theoretical justification of Algorithm~\ref{Algo:SPCAvRPsub} does not require that the leading eigenvectors enjoy disjoint supports.  Instead, we ask for $V_m$ to have not too many non-zero rows, and for these non-zero rows to have comparable Euclidean norms (i.e.~to satisfy an incoherence condition).  More precisely, writing $\mathrm{nnzr}(V)$ for the number of non-zero rows of a matrix $V$, for $\mu \geq 1$, we consider the setting where $V_m$ belongs to the set 
\[
\mathbb{O}_{p,m,k}(\mu) := \biggl\{V\in\mathbb{O}_{p,m}, \ \text{$\mathrm{nnzr}(V) \leq k$, } \frac{\max_{j:\|V^{(j,\cdot)}\|_2 \neq 0} \|V^{(j,\cdot)}\|_2}{\min_{j:\|V^{(j,\cdot)}\|_2 \neq 0} \|V^{(j,\cdot)}\|_2} \leq \mu\biggr\}.
\]
Writing $S_0:=\{j\in[p]: V_m^{(j,\cdot)}\neq 0\}$ for the set of indices of the non-zero rows of $V_m$, since $\sum_{j\in S_0} \|V_m^{(j,\cdot)}\|_2^2 = \|V_m\|_{\text{F}}^2 = m$, a consequence of our incoherence parameter definition is that for $V_m \in \mathbb{O}_{p,m,k}(\mu)$, we have
\begin{equation}
\label{Eq:VMinVMax2}
\frac{m^{1/2}}{k^{1/2}\mu}\leq \|V_m^{(j,\cdot)}\|_2 \leq \frac{m^{1/2}\mu}{k^{1/2}},\qquad \forall\ j\in S_0.
\end{equation}

The following is our main result on the performance of our SPCAvRP algorithm.
\begin{theorem}
  \label{Thm:Main}
Suppose $Q\in \textup{RCC}_p(K)$ has an associated covariance matrix $\Sigma = I_p+V_m\Theta V_m^\top$, where $V_m\in\mathbb{O}_{p,m,k}(\mu)$ and $\Theta = \textup{diag}(\theta_1, \ldots, \theta_m)$, with $\theta_1\geq\cdots\geq \theta_m > 0$.  Let $X_1,\ldots,X_n\stackrel{\mathrm{iid}}{\sim} Q$ and let $\hat V_m$ be the output of Algorithm~\ref{Algo:SPCAvRPsub} with input $X_1,\ldots,X_n$, $A$, $B$, $m$, $d$ and $\ell$. Suppose $d\geq \max\{m+1,k\}$, $\ell\geq k$, and 
\begin{equation}
\label{Eq:Cond2}
32K\sqrt\frac{d\log p}{n}\leq \frac{\theta_m}{k\mu^2}.
\end{equation}
Then with probability at least $1-2p^{-3}-pe^{-A\theta_m^2/(50p^2\mu^8\theta_1^2)}$, we have 
\[
L(\hat V_m, V_m)\leq 4K\sqrt\frac{m\ell\log p}{n\theta_m^2}.
\]
\end{theorem}
An immediate consequence of Theorem~\ref{Thm:Main} is that, provided that $A \gtrsim p^2\mu^8\theta_1^2\theta_m^{-2} \log p$ and $p^{-3} \leq K\sqrt\frac{m\ell\log p}{n\theta_m^2}$, our SPCAvRP algorithm achieves the bound
\begin{equation}
\label{Eq:ExpectationBound}
  \mathbb{E}L(\hat V_m, V_m)\lesssim K\sqrt\frac{m\ell\log p}{n\theta_m^2}
\end{equation}
under the conditions of the theorem. The salient observation here is that this choice of $A$, together with the algorithmic complexity analysis given in Section~\ref{Sec:Multiple}, ensures that Algorithm~\ref{Algo:SPCAvRPsub} achieves the rate in~\eqref{Eq:ExpectationBound} in polynomial time (provided we consider $\mu$, $\theta_1$ and $\theta_m$ as constants).  The minimax lower bound given in Proposition~\ref{Prop:LowerBound} below complements Theorem~\ref{Thm:Main} by showing that this rate is minimax optimal, up to logarithmic factors, over \emph{all} possible estimation procedures, provided that $\ell \lesssim k$, that $m \lesssim \log (p/k) \asymp \log p$  and that we regard $K$ and $\mu$ as constants (as well as other regularity conditions).  It is important to note that this does not contradict the fundamental statistical and computational trade-off for this problem established in \citet{Wang2016}, because Condition~\eqref{Eq:Cond2} ensures that we are in the high effective sample size regime defined in that work.  Assuming the Planted Clique Hypothesis from theoretical computer science, this is the only setting in which any (randomised) polynomial time algorithm can be consistent.

The following proposition establishes a minimax lower bound for principal subspace estimation. It is similar to existing minimax lower bounds in the literature for SPCA under row sparsity, e.g.~\citet[][Theorem~3.1]{Vu2013}.  The main difference is that we show that imposing an incoherence condition on the eigenspace does not make the problem any easier from this minimax perspective. For any $V\in\mathbb{O}_{p,m}$ and $\theta>0$, we write $P_{V,\theta} := N_p(0,I_p+\theta VV^\top)$. 
\begin{proposition}
  \label{Prop:LowerBound}
Assume that $p\geq 5k$, $k\geq 4m$, $k\log\bigl((p-m)/k\bigr)\geq 17$ and that $nm^2\theta^2\geq k^2\max\bigl\{m, \log(p/k)\bigr\}$. Then 
\[
\inf_{\tilde V}\sup_{V\in\mathbb{O}_{p,m,k}(3)}\mathbb{E}_{P_{V,\theta}} L(\tilde V, V) \gtrsim \sqrt\frac{k\{m+\log(p/k)\}}{n\theta^2},
\]
where the infimum is taken over all estimators $\tilde V = \tilde V(X_1,\ldots,X_n)$ and the expectation is with respect to $X_1,\ldots,X_n\stackrel{\mathrm{iid}}{\sim}P_{V,\theta}$. 
\end{proposition}


An interesting aspect of Theorem~\ref{Thm:Main} is that the same conclusion holds for every $B \in \mathbb{N}$.  On the one hand, it is attractive that we do not need to make any restrictions here; however, one would also expect the statistical performance of the algorithm to improve as $B$ increases.  Indeed, this is what we observe empirically; see Figure~\ref{fig:ABselection} in Section~\ref{s:numerics}.  It turns out that we are able to demonstrate the effect of increasing $B$ theoretically in the special setting where all signal coordinates have homogeneous signal strength, i.e., $V_m\in\mathbb{O}_{p,m,k}(1)$. As illustrated by the following corollary (to Theorem~\ref{Thm:Main}) and its proof, as $B$ increases, signal coordinates are selected with increasing probability by the best projection within each group of $B$ projections, and this significantly reduces the number of groups $A$ required for rate optimal estimation.

Recall that the hypergeometric distribution $\mathrm{HyperGeom}(d,k,p)$ models the number of white balls obtained when drawing $d$ balls uniformly and without replacement from an urn containing $p$ balls, $k$ of which are white. We write $F_{\mathrm{HG}}(\cdot;d,k,p)$ for its distribution function. 
\begin{cor}
\label{Cor:Homogeneous}
In addition to the conditions of Theorem~\ref{Thm:Main}, assume that $\mu = 1$, $\theta_1 = \cdots = \theta_m$ and that $B = \bigl\lceil 2^{-1}\bigl(1-F_{\mathrm{HG}}(t-1; d,k,p)\bigr)^{-1}\bigr\rceil$ for some $t\in [k]$. Then
\[
\mathbb{P}\biggl(L(\hat V_m, V_m) > 4K\sqrt\frac{m\ell\log p}{n\theta_m^2}\biggr) \leq 2p^{-3} + pe^{-At^2/(800k^2)}.
\]
\end{cor}
Since in this corollary, we use Lemma~\ref{Lem:UniqueMax} instead of \eqref{Eq:SignalCoordinateProbability2} to control the inclusion probability of signal coordinates, the condition $d\geq k$ from Theorem~\ref{Thm:Main} is in fact no longer needed. We note that for any fixed $t$, the function $F_{\mathrm{HG}}(t-1; d, k, p)$ is decreasing with respect to $d \in [p]$. Thus, Corollary~\ref{Cor:Homogeneous} also illustrates a computational trade-off between the choice of $d$ and $B$. This trade-off is also demonstrated numerically in Figure~\ref{fig:d_and_AB}.

Finally, we remark that our algorithm allows us to understand the statistical and computational trade-off in SPCA in a more refined way. Recall that in the limiting case when $B=\infty$,  the estimator produced by Algorithm~\ref{Algo:SPCAvRPsub} (with $d =\ell = k$ and, for the simplicity of discussion, $m=1$) is equal to the estimator $\hat{v}_1$ given in~\eqref{opt_est}, i.e.~the leading $k$-sparse eigenvector of $\hat{\Sigma}$. In fact, this is already true with high probability for $B\gtrsim \binom{p}{k}$. Hence, for $B$ exponentially large, the SPCAvRP estimator is minimax rate optimal as long as $n\gtrsim mk\theta_m^{-2}\log p$, which corresponds to the intermediate effective sample size regime defined in \citet{Wang2016}. For such a choice of $B$, however, Algorithm~\ref{Algo:SPCAvRPsub} will not run in polynomial time, which is in agreement with the conclusion of \citet{Wang2016} that there is no randomised polynomial time algorithm that can attain the minimax rate of convergence in this intermediate effective sample size regime. On the other hand, as mentioned above, SPCAvRP is minimax rate optimal, using only a polynomial number of projections, in the high effective sample size regime as discussed after Theorem~\ref{Thm:Main}. Therefore, the flexibility in varying the number of projections in our algorithm allows us to analyse its performance in a continuum of scenarios ranging from where consistent estimation is barely possible, through to high effective sample size regimes where the estimation problem is much easier. 

\section{Numerical experiments}\label{s:numerics}

In this section we demonstrate the  performance of our proposed method in different examples and discuss the practical choice of its input parameters. We also compare our method with several existing sparse principal component estimation algorithms on both simulated and experimental data. All experiments were carried out using the \textsf{R} package `SPCAvRP' \citep{GWS2017}.

\subsection{Choice of input parameters}\label{ss:inputparameters}

\subsubsection{Choice of \texorpdfstring{$A$}{A} and \texorpdfstring{$B$}{B}}

In Figure~\ref{fig:ABselection}, we show that choosing $B>1$, which ensures that we make a non-trivial selection within each group of projections, considerably improves the statistical performance of the SPCAvRP algorithm. Specifically, we see that using the same total number of random projections, our two-stage procedure has superior performance over the naive aggregation over all projections, which corresponds to setting $B = 1$ in the SPCAvRP algorithm.   Interestingly, Figure~\ref{fig:ABselection} shows that simply increasing the number of projections, without performing a selection step, does not noticeably improve the performance of the basic aggregation. We note that even for the relatively small choices $A = 50$ and $B=25$, the SPCAvRP algorithm does significantly better than the naive aggregation over $180000$ projections.

\begin{figure}[htbp]
 \centering
  \hspace{2em} \scriptsize{$v_1=k^{-1/2} (\boldsymbol{1}_{k}^\top, \boldsymbol{0}_{p-k}^\top)^{\top}$} \hspace{11em} \scriptsize{$v_1\propto (k,k-1,\ldots,1,\boldsymbol{0}_{p-k}^\top)^{\top}$}\\
 \includegraphics[scale=0.55]{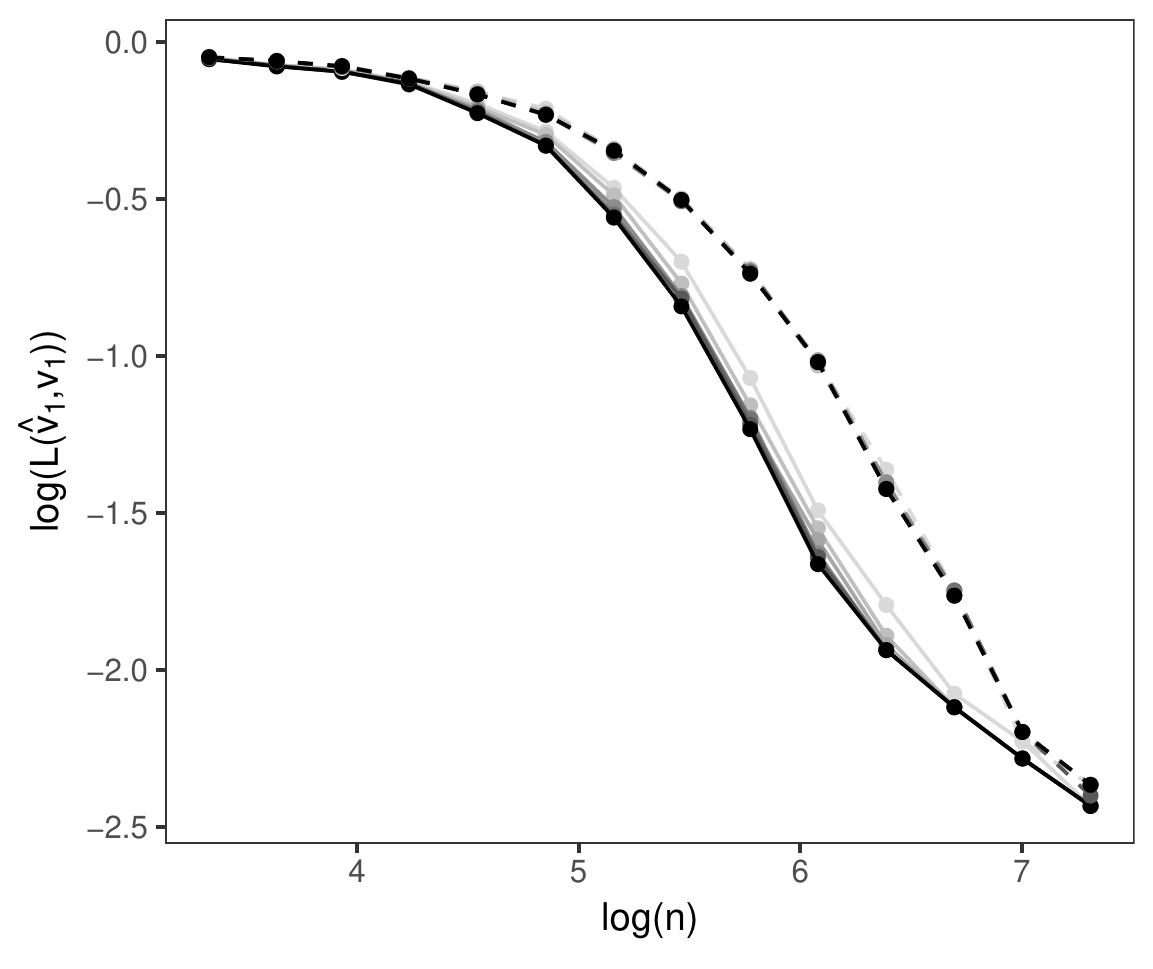}
 \includegraphics[scale=0.55]{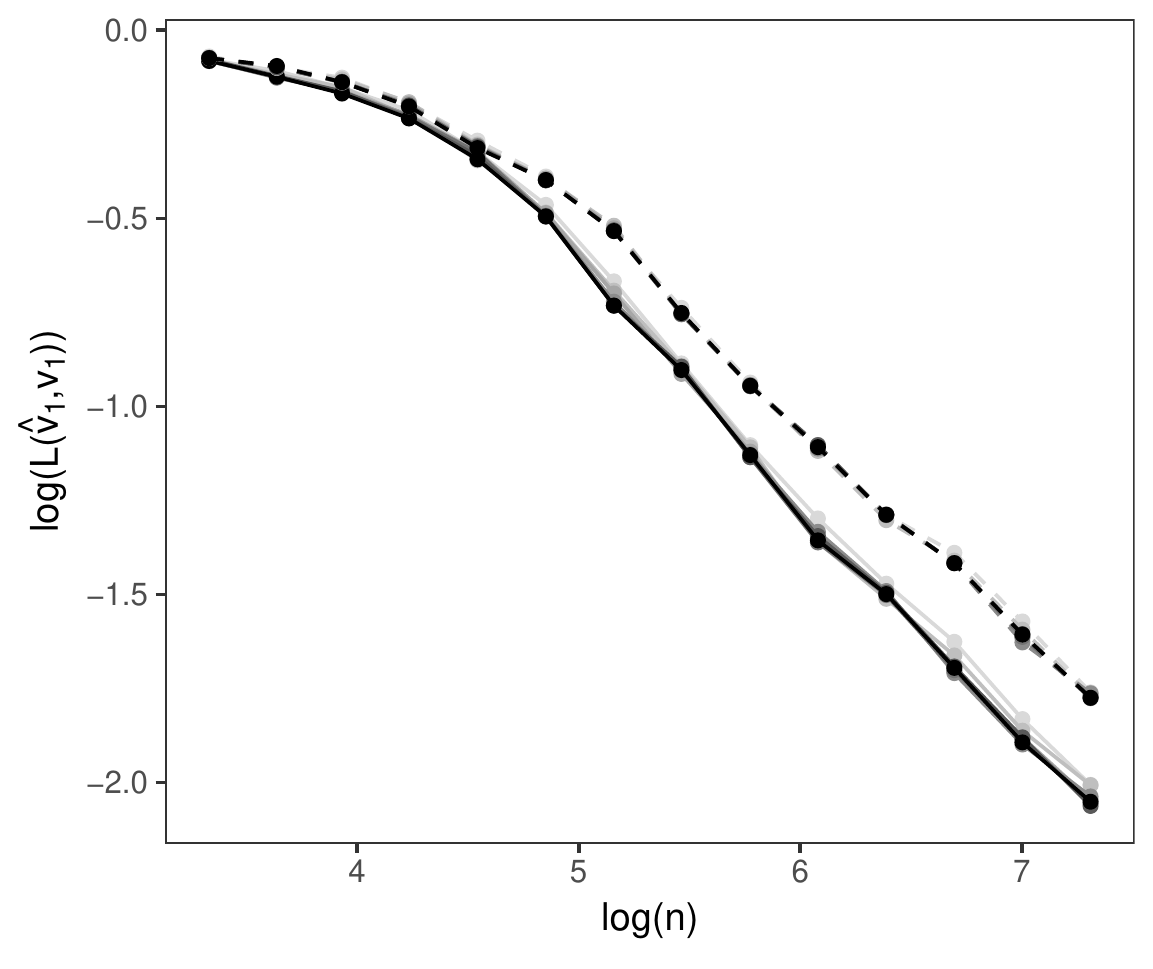}
  \caption{\textbf{Average loss $L(\hat v_1,v_1)$ against the sample size $n$, on the log-log scale, 
when $B=1$ (dashed lines) and $B>1$ (solid lines).} In each case, $n$ observations are generated from $N_p(0,I_p + v_1v_1^\top)$, with $p=50$, $k=7$ and $v_1$ written above each panel, and the loss $L(\hat v_1,v_1)$ is computed for $\hat v_1$ as in Algorithm \ref{Algo:SPCAvRP}, with $d=\ell=k$ and $A$ and $B$ selected as described next, which is then averaged over 100 repetitions.  Light to dark grey solid lines: $A\in\{50,100,200,300,400,500,600\}$ and $B=A/2$. 
Light to dark grey dashed lines: ${A}\in\{50\times25,100\times50, 200\times100, 300\times150, 400\times200, 500\times250, 600\times300\}$ and $B=1$.}
  \label{fig:ABselection}
\end{figure}

Figure~\ref{fig:ABselection_fixed} demonstrates the effect of increasing either $A$ or $B$ while keeping the other fixed. We can see from the left panel of Figure~\ref{fig:ABselection_fixed} that increasing $A$ steadily improves the estimation quality, especially in the medium effective sample size regime and when $A$ is relatively small.  This agrees with the result in Theorem~\ref{Thm:Main}, where the bound on the probability of attaining the minimax optimal rate improves as $A$ increases. Thus, in practice, we should choose $A$ to be as large as possible subject to our computational budget.   
The choice of~$B$, however, is a little more delicate.  In some settings, such as the single-spiked, homogeneous model in the right panel of Figure~\ref{fig:ABselection_fixed}, the performance appears to improve as $B$ increases, though the effect is only really noticeable in the intermediate effective sample size regime.  
On the other hand, we can also construct examples where as $B$ increases, some signal coordinates will have increasingly high probability of inclusion compared with other signal coordinates, making the latter less easily distinguishable from the noise coordinates. Hence the performance does not necessarily improve as $B$ increases; see Figure~\ref{Fig:Btradeoff}.

In general, we find that $A$ and $B$ should increase with $p$. Based on our numerical experiments, we suggest using $B = \lceil A/3\rceil$ with $A =300$ when $p\approx100$, and $A = 800$ when $p\approx1000$.

\begin{figure}[htbp]
 \centering
 \hspace{2em} \scriptsize{$B=50$} \hspace{18em} \scriptsize{$A=200$}\\
 \includegraphics[scale=0.55]{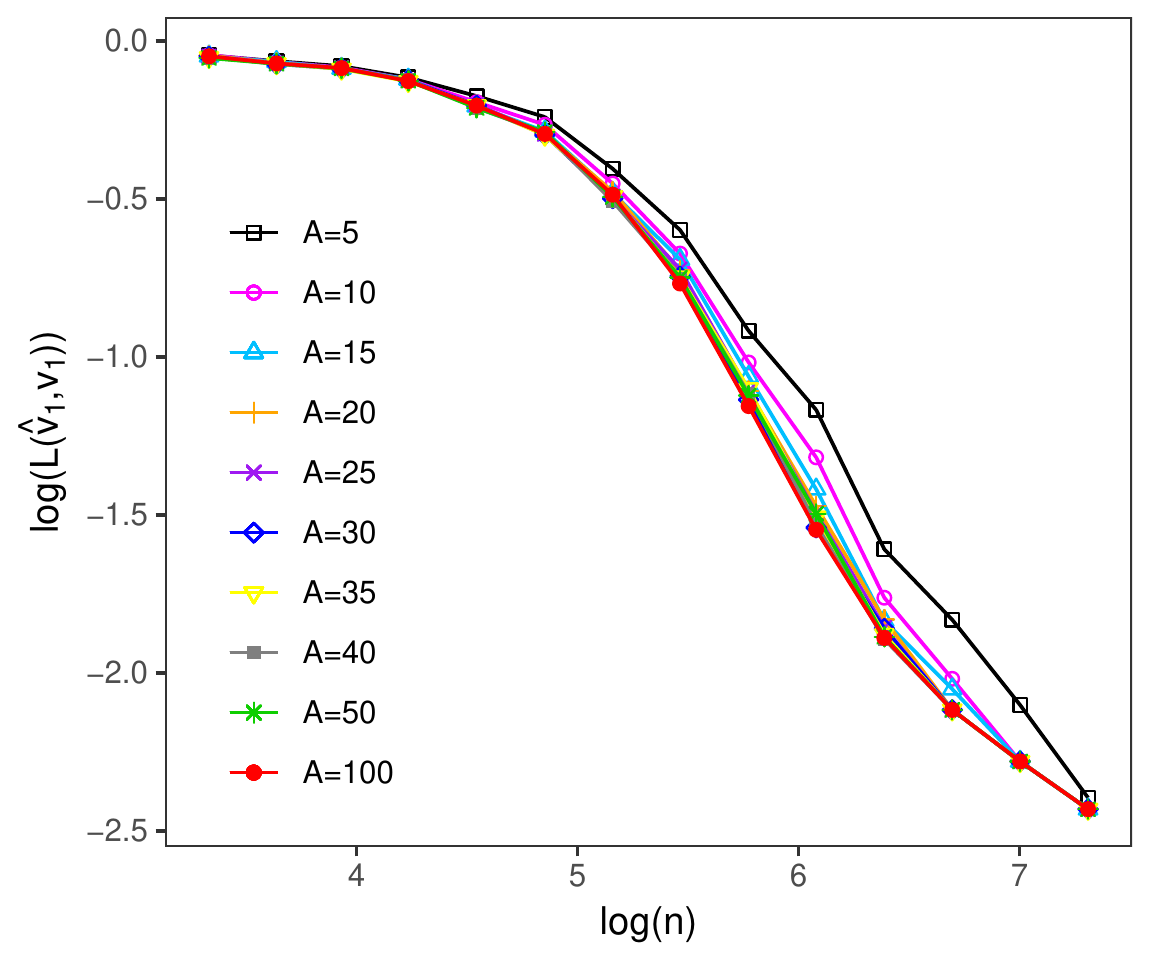} \includegraphics[scale=0.55]{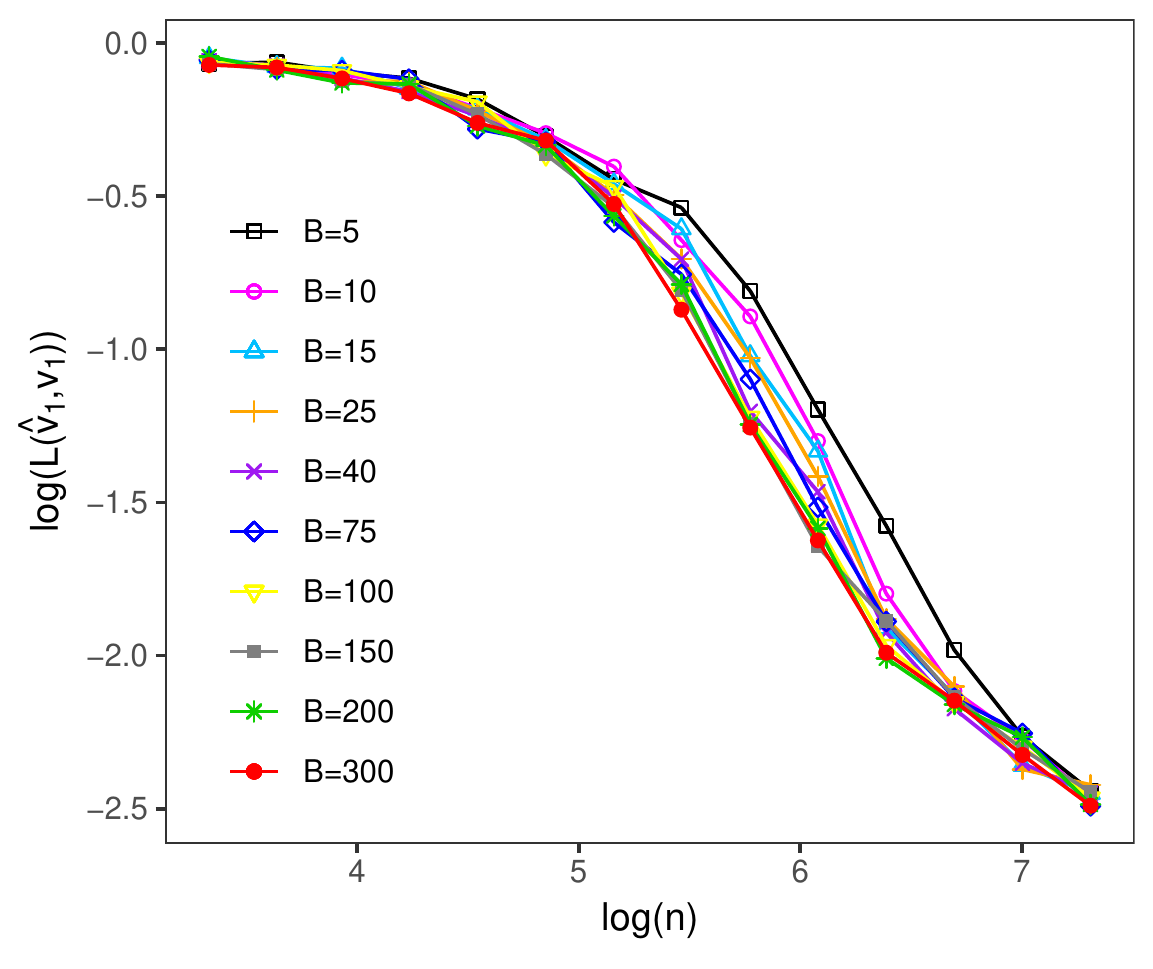}
  \caption{\textbf{Average loss $L(\hat v_1,v_1)$ as the sample size $n$ increases for different choices of $A$ or $B$.} In the left panel, $B=100$ and $A$ is varied; on the right, $A = 200$ and $B$ is varied. In both panels, the distribution is $N_p(0,I_p + v_1v_1^\top)$ with $v_1 = k^{-1/2} (\boldsymbol{1}_{k}^\top, \boldsymbol{0}_{p-k}^\top)^{\top}$, $p=50$, $k=7$, and other algorithmic parameters are $d=l=7$.}
  \label{fig:ABselection_fixed}
\end{figure}

\begin{figure}[htbp]
 \centering
 \includegraphics[scale=0.55]{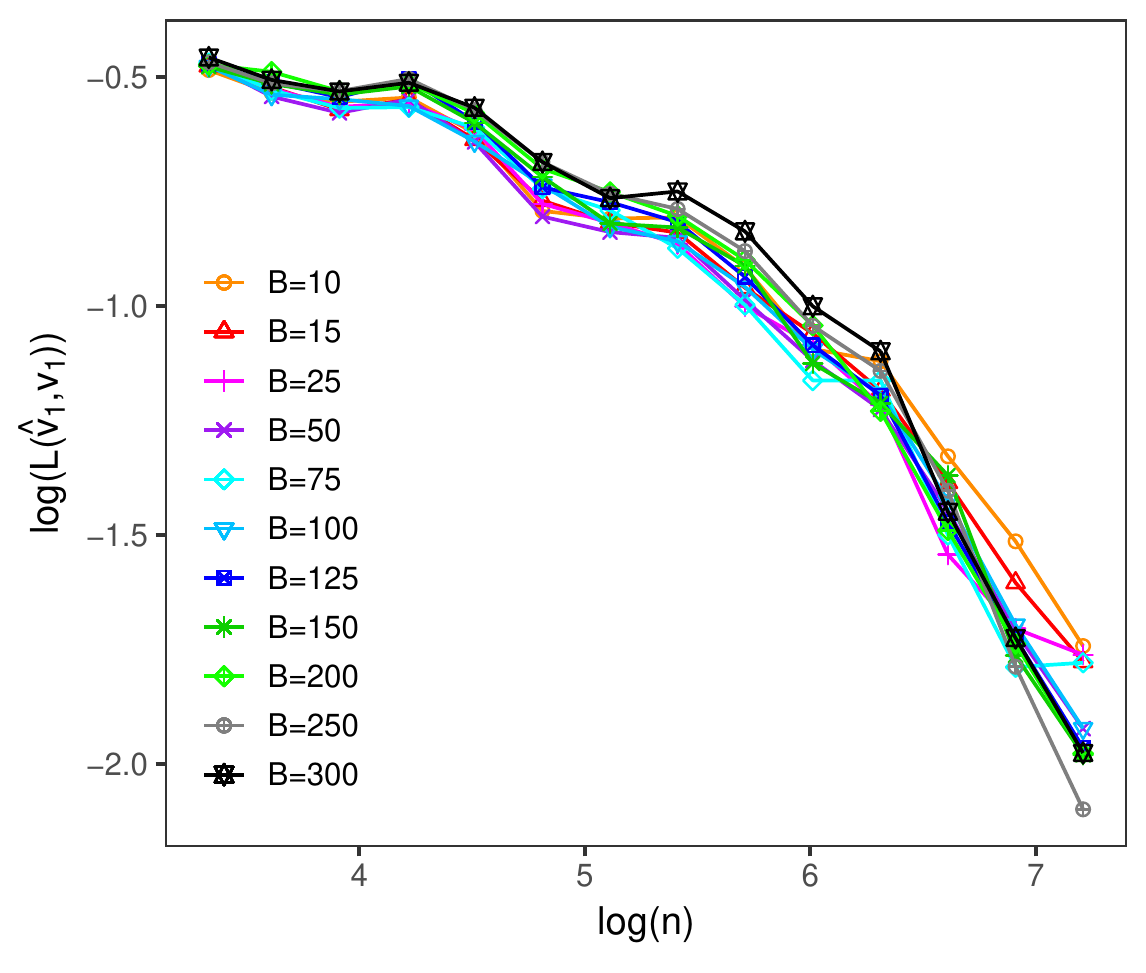} \includegraphics[scale=0.55]{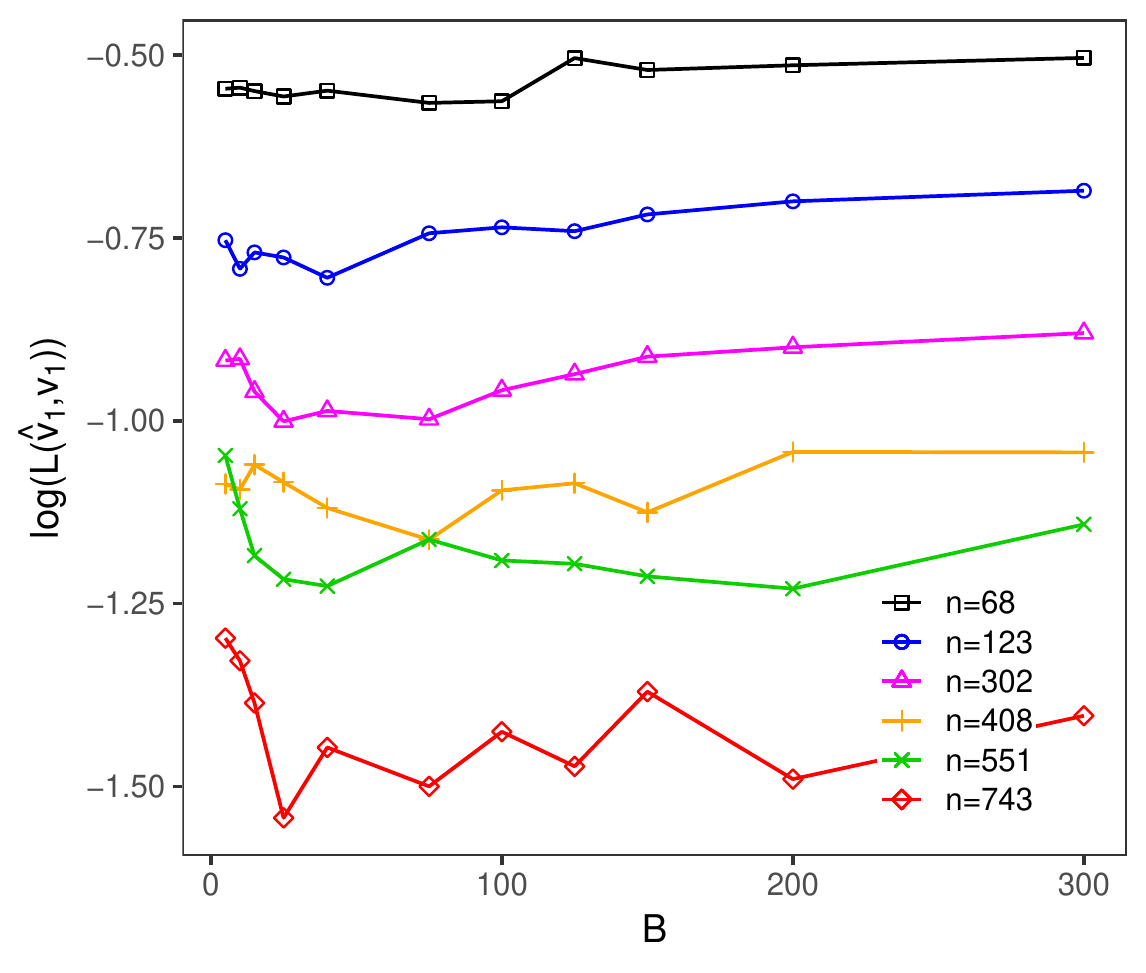}
  \caption{\textbf{Trade-off in the choice of $B$.} Left panel: the average loss as a function of  $n$, on the log-log scale, where $B$ is varied. Right panel: the logarithm of average loss as a function of $B$, where $n$ is varied. In both panels, the distribution is $N_p(0,I_p + 10 v_1v_1^\top + 9 v_2v_2^\top)$ with $v_1 = k^{-1/2} (\boldsymbol{1}_{k}^\top, \boldsymbol{0}_{p-k}^\top)^{\top}$, $v_2 = k^{-1/2} (\boldsymbol{0}_{3}^\top, -1, 1,-1, 1,-1, 1, 1, \boldsymbol{0}_{p-k-3}^\top)^{\top}$, $p=50$, $k=7$, and algorithmic parameters $A=200$, $d=l=7$.}
  \label{Fig:Btradeoff}
\end{figure}

\subsubsection{Choice of \texorpdfstring{$d$}{d}}\label{ss:ld_choice}
So far in our simulations we have assumed that the true sparsity level $k$ is known and we took the dimension $d$ of the random projections to be equal to $k$, but in practice $k$ may not be known in advance.  In Figure~\ref{fig:selecting_d}, however, we see that for a wide range of values of $d$, the loss curves are relatively close to each other, indicating the robustness of the SPCAvRP algorithm to the choice of $d$. For the homogeneous signal case, the loss curves for different choices of $d$ merge in the high effective sample size regime, whereas in the intermediate effective sample size regime, we may in fact see improved performance when $d$ exceeds $k$. In the inhomogeneous case, the loss curves improve as $d$ increases up to $k$ and then exhibit little dependence on $d$ when $d\geq k$.

\begin{figure}[htbp]
 \centering
 \hspace{2em} \scriptsize{$v_1=k^{-1/2} (\boldsymbol{1}_{k}^\top, \boldsymbol{0}_{p-k}^\top)^{\top}$} \hspace{11em} \scriptsize{$v_1\propto (k,k-1,\ldots,1,\boldsymbol{0}_{p-k}^\top)^{\top}$}\\
 \includegraphics[scale=0.55]{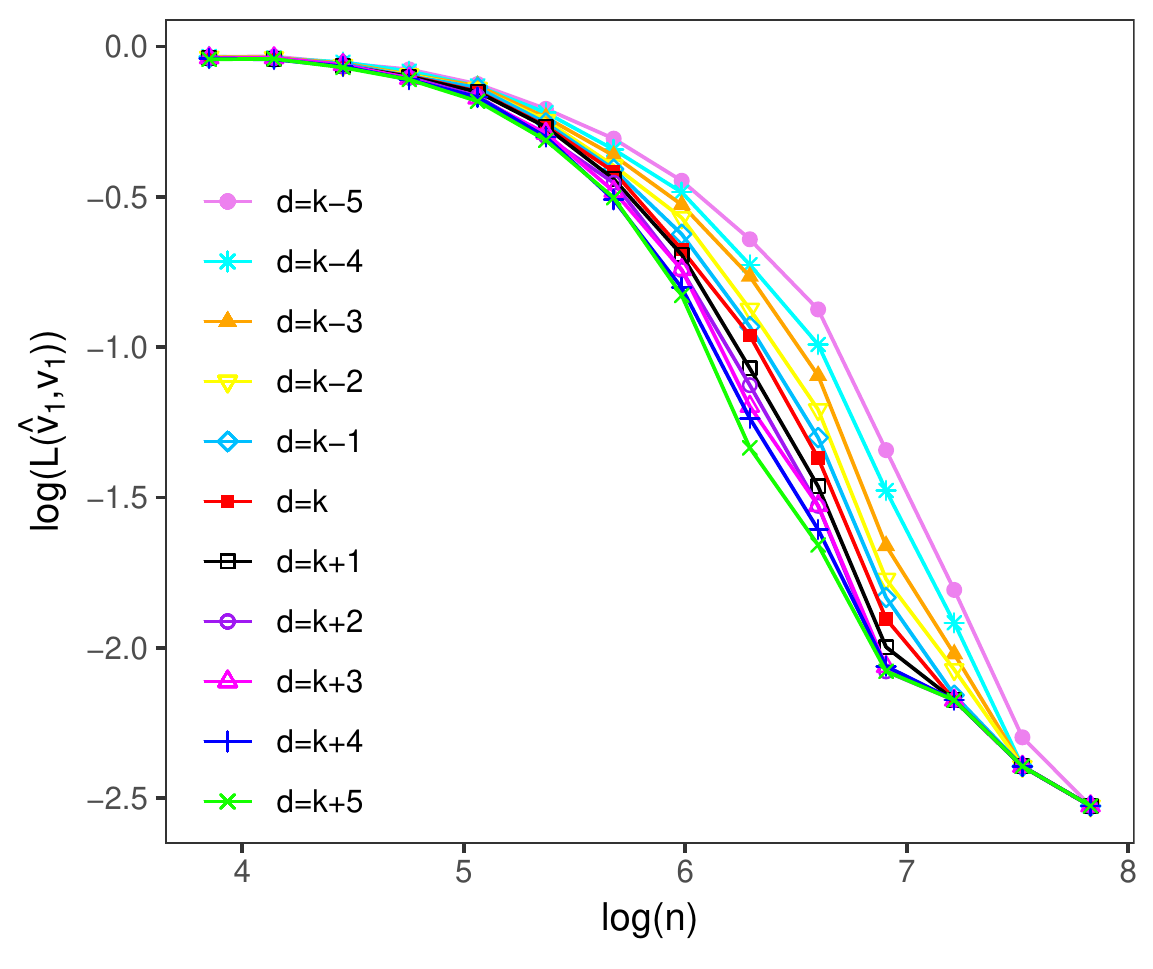} \includegraphics[scale=0.55]{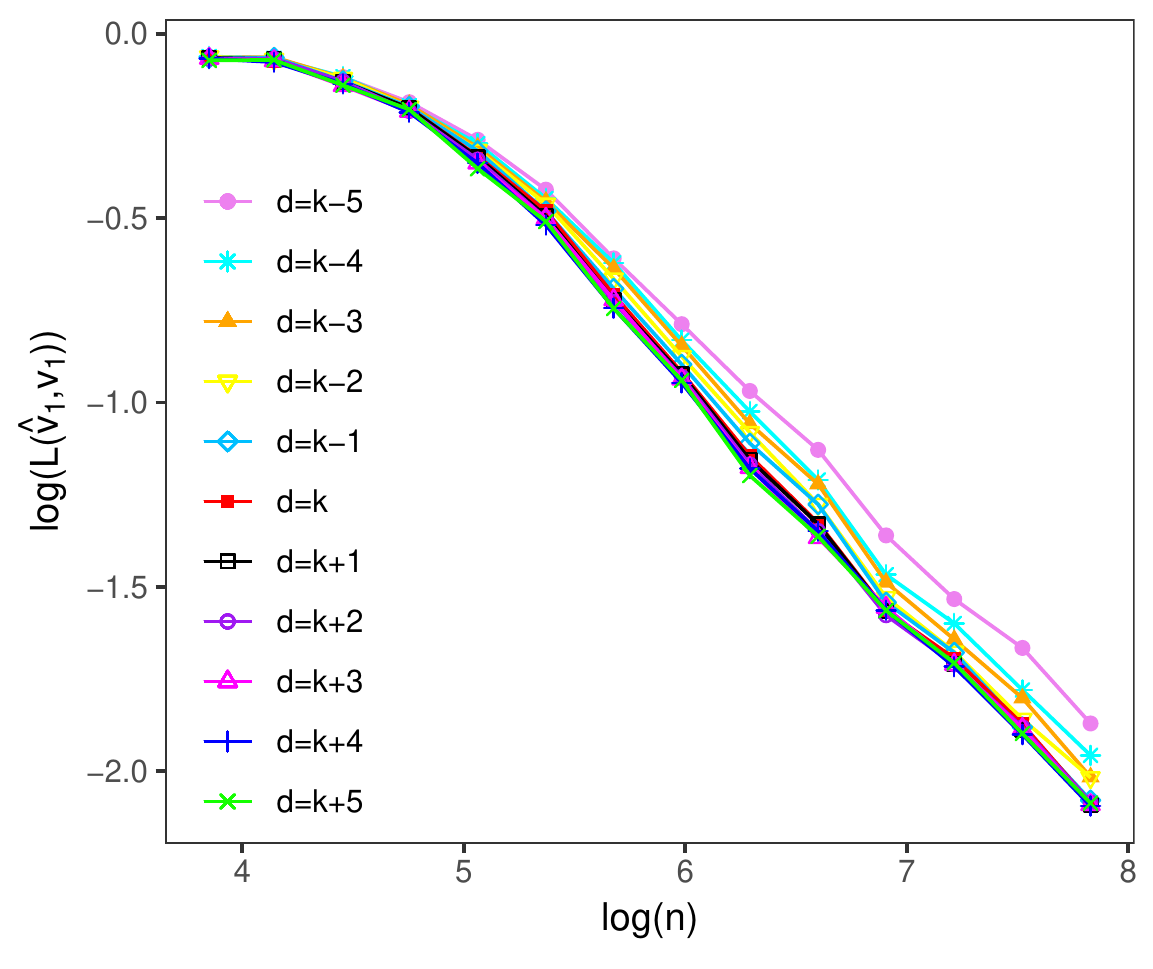}
  \caption{\textbf{Average loss $L(\hat v_1,v_1)$ as  $n$ increases for different choices of $d$.} The distribution is $N_p(0,I_p + v_1v_1^\top)$ with $p=100$, $k=10$ and $v_1$ written above each panel. Other algorithmic parameters: ${A}=150$, ${B}=50$, $\ell=k$.}
  \label{fig:selecting_d}
\end{figure}

Although decreasing $d$ reduces computational time, for a smaller choice of $d$ it is then less likely that each signal coordinate will be selected in a given random projection. This means that a smaller $d$ will require a larger number of projections $A$ and $B$ to achieve desired accuracy, thereby increasing computational time.  To illustrate this computational trade-off, in Figure~\ref{fig:d_and_AB}, for a single spiked homogeneous model, we plot the trajectories of the average loss as a function of time (characterised by the choices of $A$ and $B$), for different choices of $d$.  Broadly speaking, the figures reveal that choosing $d < k$ needs to be compensated by a very large choice of $A$ and $B$ to achieve similar statistical performance to that which can be obtained with $d$ equal to, or even somewhat larger than, $k$.  

\begin{figure}[htbp]
 \centering
 \hspace{1.5em} \scriptsize{$A\in\{30,\ldots,300\}$, $B=50$} \hspace{5.7em} \scriptsize{$A=100$, $B\in\{20,\ldots,300\}$} \hspace{5.5em} \scriptsize{$A\in\{30,\ldots,300\}$, $B=A/2$}\\
  \includegraphics[scale=0.46]{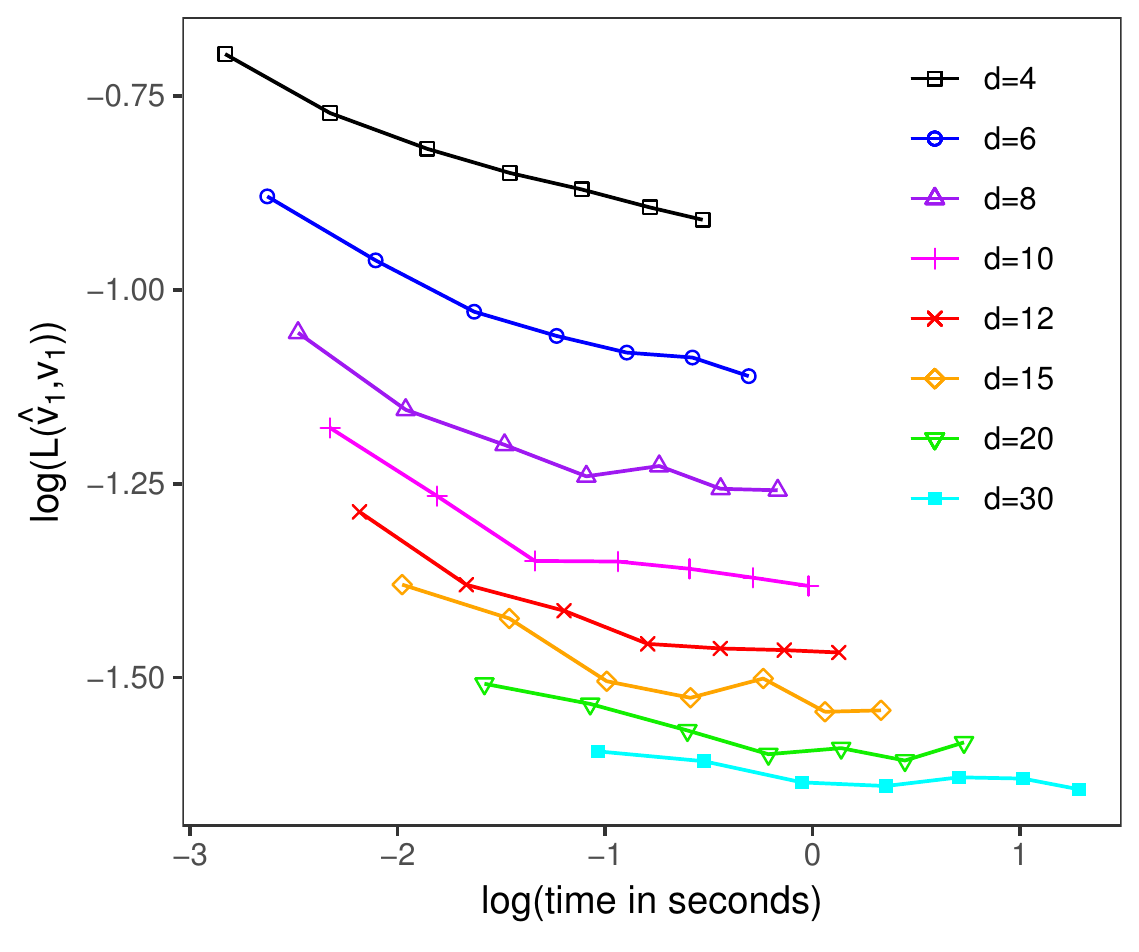}  \includegraphics[scale=0.46]{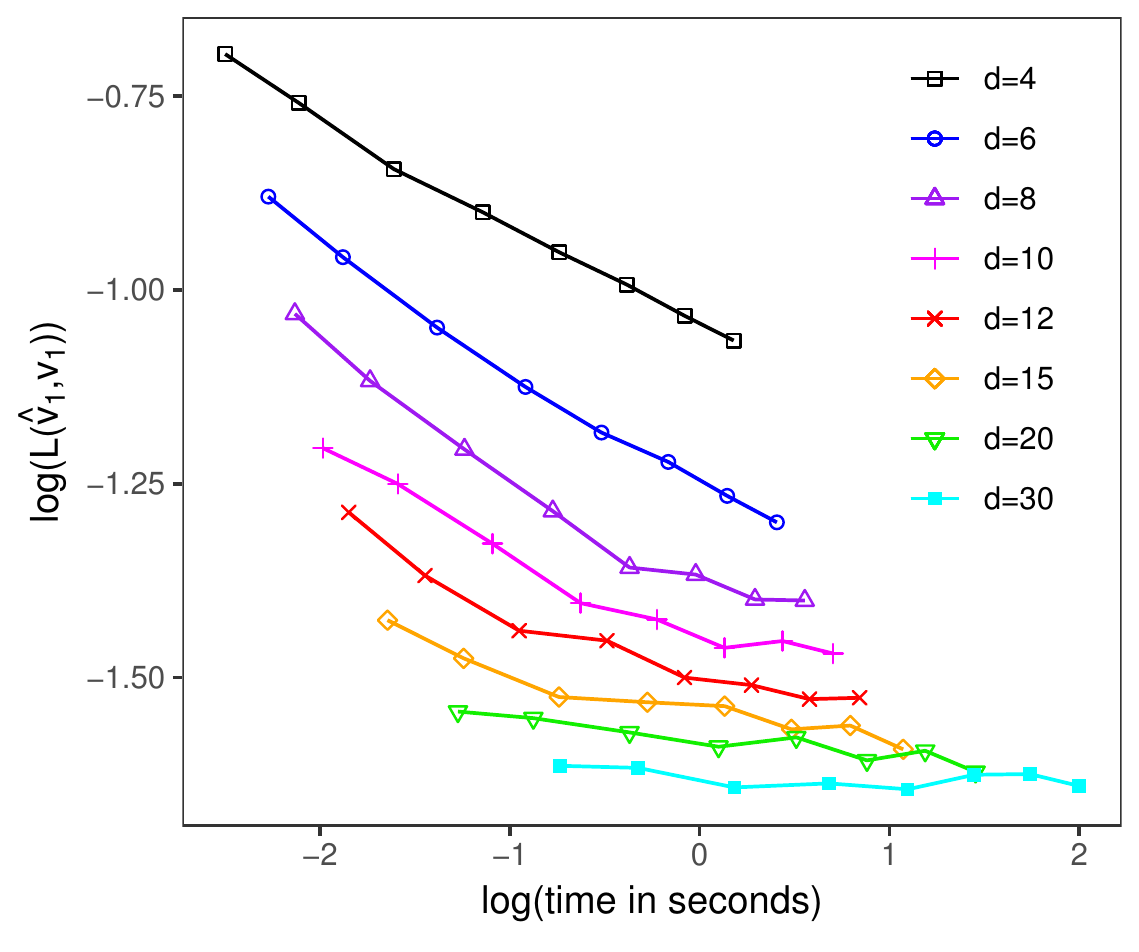} \includegraphics[scale=0.46]{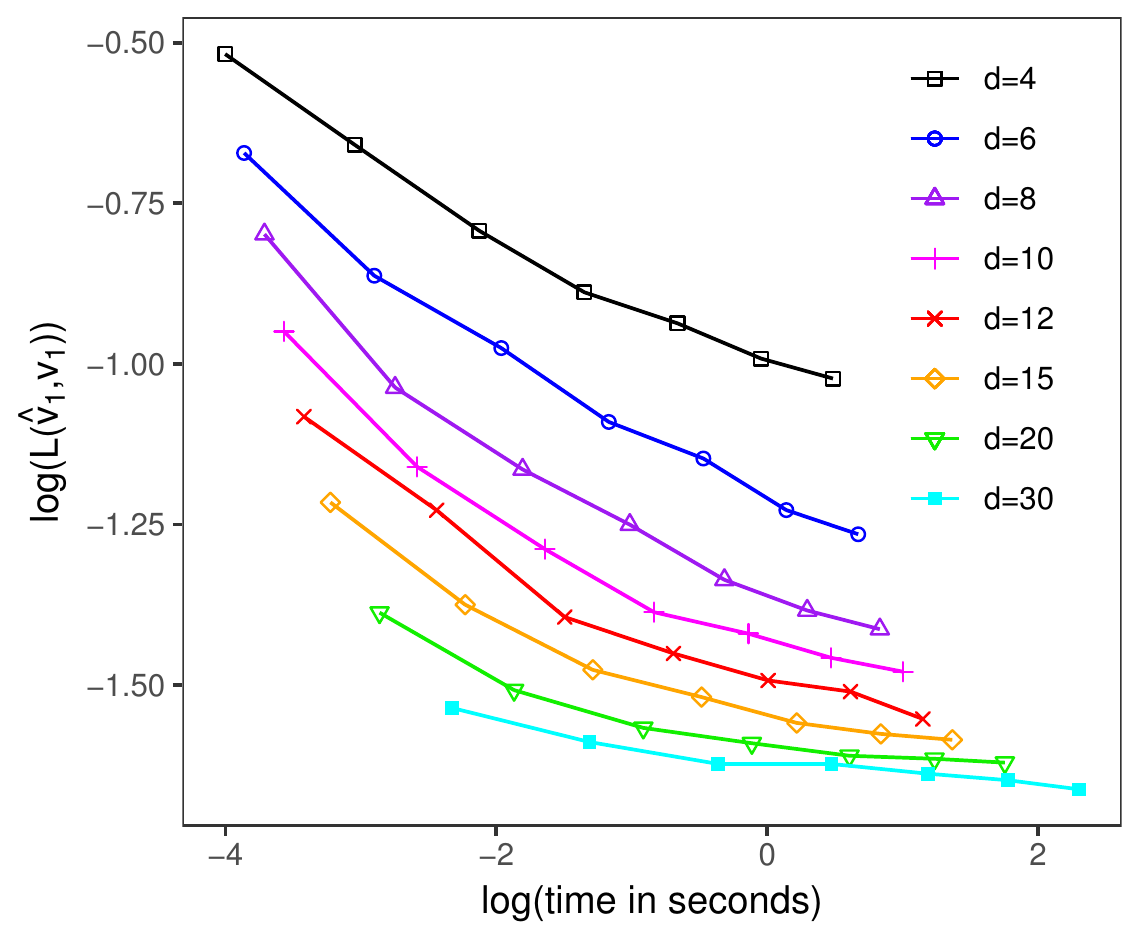}
 \caption{\textbf{Computational trade-off in the choice of $d$ and $A$ and $B$.} We generated $n=600$ observations from distribution  $N_p(0,I_p + v_1v_1^\top)$, where $p=100$, $k=10$, $v_1=k^{-1/2} (\boldsymbol{1}_{k}^\top, \boldsymbol{0}_{p-k}^\top)^{\top}$. For a fixed $d\in\{4,\ldots,30\}$ we plot the trajectory realised when increasing $A$ (left); $B$ (middle); both $A$ and $B$ (right).}
 \label{fig:d_and_AB}
\end{figure}

In practice, we suggest using $d=k$  where $k$ is known, but when $k$ is not given in advance, we would advocate erring on the side of projecting into a subspace of dimension slightly larger than the sparsity level of the true eigenvectors, as this allows a significantly smaller choice of $A$ and $B$, which results in an overall time saving. 

\subsubsection{Choice of  \texorpdfstring{$\ell$}{l}}\label{ss:selection_l}

The parameter $\ell$ corresponds to the sparsity of the computed estimator; 
large values of $\ell$ increase the chance that signal coordinates are discovered but also increase the probability of including noise coordinates. This statistical trade-off 
is typical for any algorithm that aims to estimate the support of a sparse eigenvector.  It is worth noting that many of the SPCA algorithms proposed in the literature have a tuning parameter corresponding to the sparsity level, and thus cross-validation techniques have been proposed in earlier works \citep[e.g.][]{Witten2009}.

A particularly popular approach in the SPCA literature \citep[e.g.][]{Shen2008} is to choose $\ell$  by inspecting the total variance.  More precisely, for each $\ell$ on a grid of plausible values, we can compute an estimate $\hat{v}_{1,\ell} \in \mathcal{B}_0^{p-1}(\ell)$ 
and its explained variance $\mathrm{Var}_{\ell} := \hat{v}_{1,\ell}^\top \hat\Sigma \hat{v}_{1,\ell}$, and then plot $\mathrm{Var}_\ell$ against $\ell$. As can be seen from Figure~\ref{fig:selecting_k}, $\mathrm{Var}_{\ell}$ 
increases with $\ell$, but plateaus off for $\ell\geq k$. An attractive feature of our algorithm  is that this procedure does not significantly increase the total computational time, since there is no need to re-run  the entire algorithm for each value of~$\ell$. Recall that $\hat w$ in \R{alg:aggregation} of Algorithm~\ref{Algo:SPCAvRP} ranks the coordinates by their importance. Therefore, we only need to compute $\hat w$ once and then calculate $\mathrm{Var}_{\ell}$ by selecting the top $\ell$ coordinates in $\hat{w}$ for each value of~$\ell$. 


In cases where higher-order principal components need to be computed, namely when $m>1$, we can choose $\ell=\mathrm{nnzr}(V_{m})$ in Algorithm~\ref{Algo:SPCAvRPsub}, and $\ell_r=\|v_{r}\|_0$, $r\in[m]$, in Algorithm~\ref{Algo:SPCAvRPdefl}, when these quantities are known.  If this is not the case, we can choose $\ell$ in Algorithm~\ref{Algo:SPCAvRPsub} in a similar fashion as described above, by replacing $\hat{v}_{1,\ell}$ with $\hat{V}_{m,\ell}$ where $\mathrm{nnzr}(\hat V_{m,\ell})\leq \ell$, or we can choose $\ell_r$ by inspecting the total variance at each iteration $r$ of Algorithm~\ref{Algo:SPCAvRPdefl}.

\begin{figure}[htbp]
 \centering
 \hspace{2em} \scriptsize{$v_1=k^{-1/2} (\boldsymbol{1}_{k}^\top, \boldsymbol{0}_{p-k}^\top)^{\top}$} \hspace{11em} \scriptsize{$v_1\propto (k,k-1,\ldots,1,\boldsymbol{0}_{p-k}^\top)^{\top}$}\\
 \includegraphics[scale=0.55]{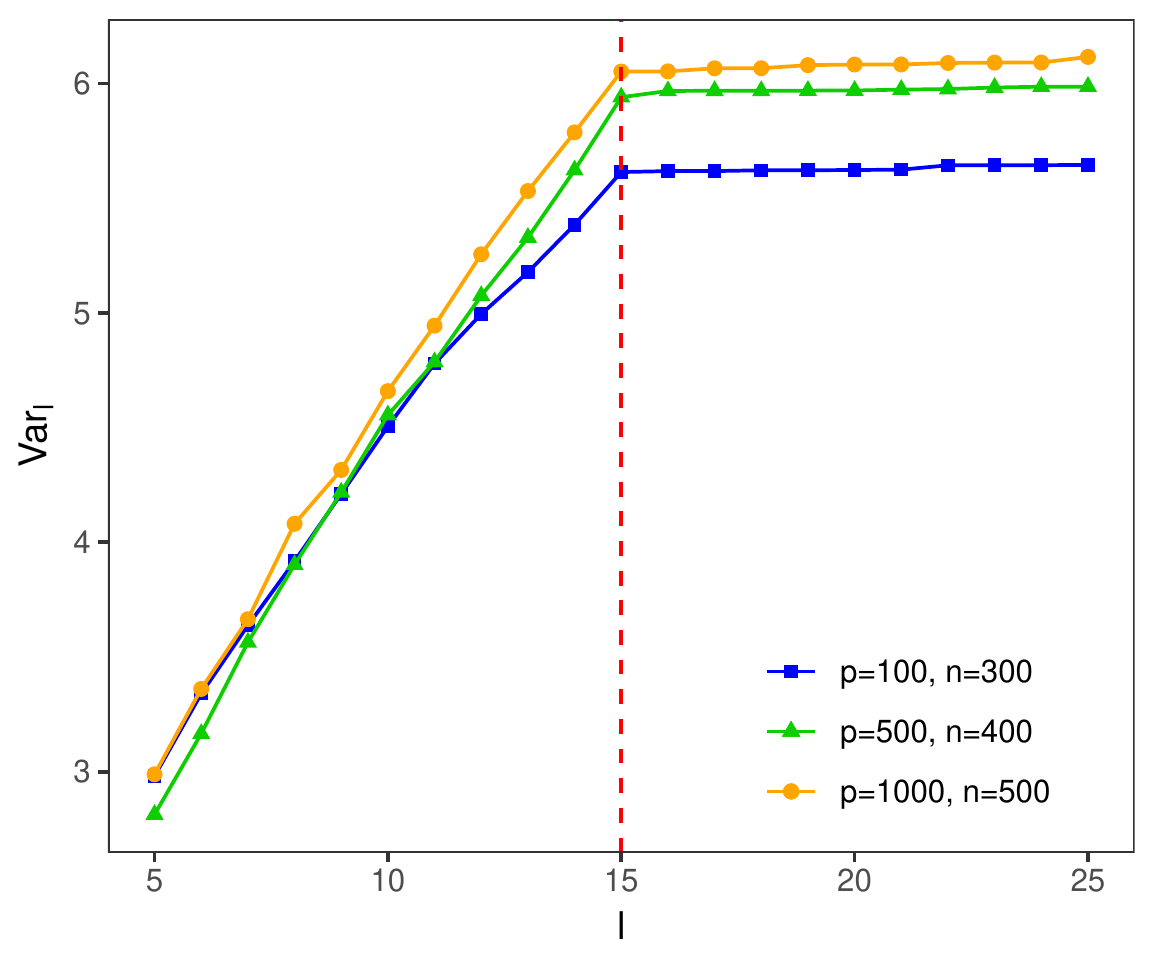} \includegraphics[scale=0.55]{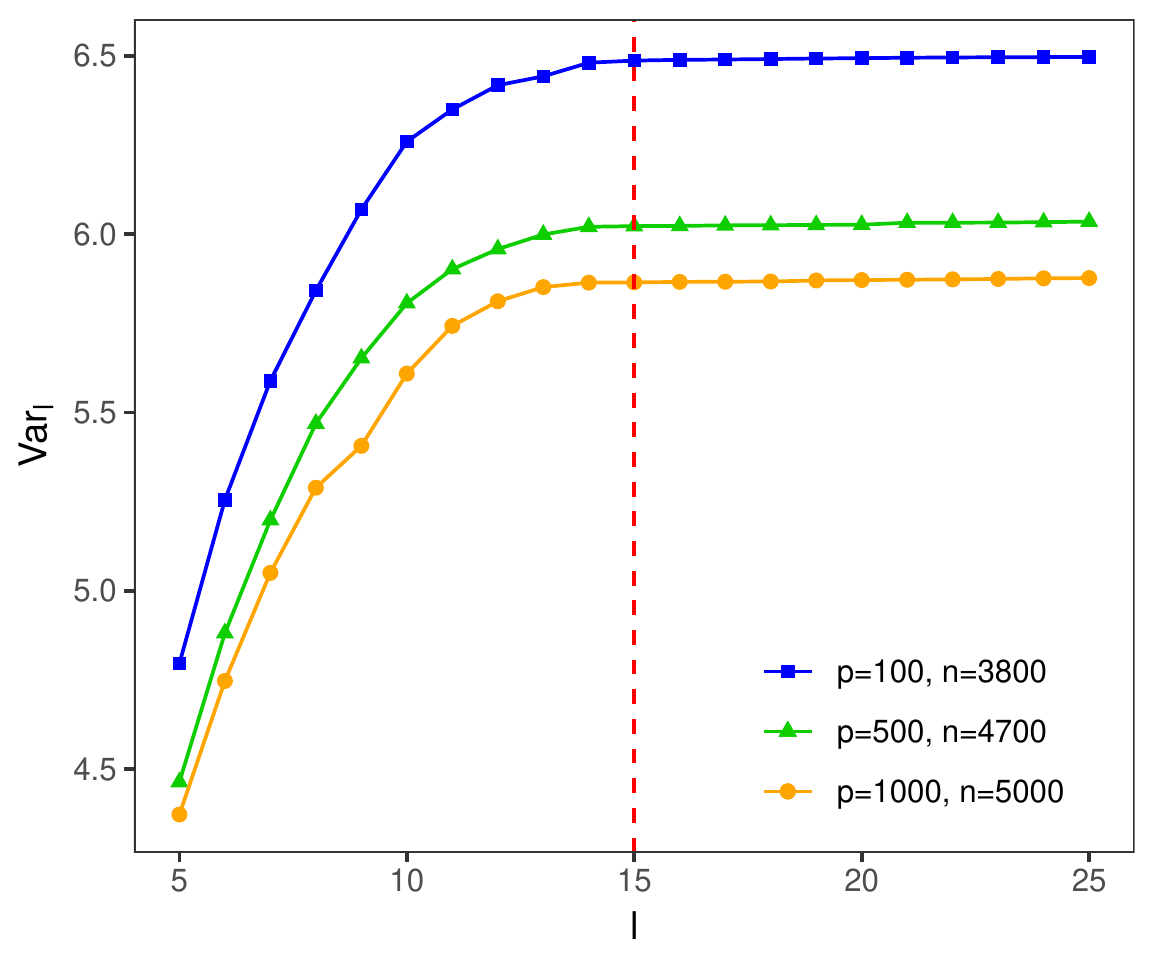} 
  \caption{\textbf{Selecting $\ell$ by inspecting the total variance $\mathrm{Var}_{\ell}$}. Observations are generated from $N_p(0,I_p + 5 v_1v_1^\top)$ with $v_1$ as written above each panel and $k=10$. SPCAvRP is used with parameters $d=10$, $A=300$, $B=100$. }
  \label{fig:selecting_k}
\end{figure}

\subsection{Comparison with existing methods}\label{ss:comparison}

In this subsection, we compare our method with several existing approaches for SPCA.  We first present examples where only the first principal component is computed, followed by examples of higher-order principal component estimation and an illustration on some genetic data.

\subsubsection{First principal component}\label{sss:comparison-single} 

In addition to the example presented in Figure \ref{fig:iterative_algs} of the introduction, we consider four further examples with 
data generated from a $N_p(0,\Sigma)$ distribution, where $\Sigma$ takes one of the two forms below:
\begin{align}\label{mod_compare}
 \Sigma_{(1)} =  \begin{pmatrix}
           2 J_{k} & & \\
                    & J_{k} & \\
                    & & \mathbf{0}
          \end{pmatrix} + I_{p}, \qquad 
 \Sigma_{(2)} = \begin{pmatrix}
           kJ_{k}  &  & \\
                  & 0.99k J_{3k} &\\
                  & & I_{(p-4k)}
          \end{pmatrix}+0.01 I_{p},
\end{align}
with different choices of $p\in\{100,200,1000,2000\}$ and $k\in\{10,30\}$.
Observe that $v_1=k^{-1/2} (\boldsymbol{1}_{k}^\top, \boldsymbol{0}_{p-k}^\top)^{\top}$ in all of these examples. The covariance matrix $\Sigma_{(1)}$ is double-spiked with $\theta_1=2$, $\theta_2=1$ and  $v_2=k^{-1/2} (\boldsymbol{0}_{k}^\top, \boldsymbol{1}_{k}^\top, \boldsymbol{0}_{p-2k}^\top)^{\top}$. We compare the empirical performance of our algorithm with methods proposed by \citet{Zou2006, Shen2008, d'Aspremont2008, Witten2009} and \citet{Ma2013}, as well as the SDP method mentioned in the introduction, by computing the average loss for each algorithm over 100 repetitions on the same set of data. We note that these are all iterative methods, whose success, with the exception of the SDP method, depends on good initialisation, so we recall their default choices.  The methods by \citet{Zou2006, Shen2008} and \citet{Witten2009} use eigendecomposition of the sample covariance matrix, i.e.~classical PCA, to compute their initial point, while \cite{d'Aspremont2008} and \cite{Ma2013} select their initialisation according to largest diagonal entries of $\hat\Sigma$.

In Figure~\ref{fig:compare}, we see that while the average losses of all algorithms decay appropriately with the sample size $n$ in the double-spiked $\Sigma_{(1)}$ setting, most of them perform very poorly in the setting of $\Sigma_{(2)}$, where the spiked structure is absent.  Indeed, only the SPCAvRP and SDP algorithms produce consistent estimators in both settings, but the empirical performance of the SPCAvRP algorithm is much better in both of the top panels; moreover, since SDP takes such a long time when $p \in \{1000,2000\}$, we do not present it in the bottom panels of Figure~\ref{fig:compare}.  

\begin{figure}[htbp]
 \centering
  \hspace{2em} {\scriptsize{$\Sigma_{(1)}$}} \hspace{15em} {\scriptsize{$\Sigma_{(2)}$}}\\
\includegraphics[scale=0.55]{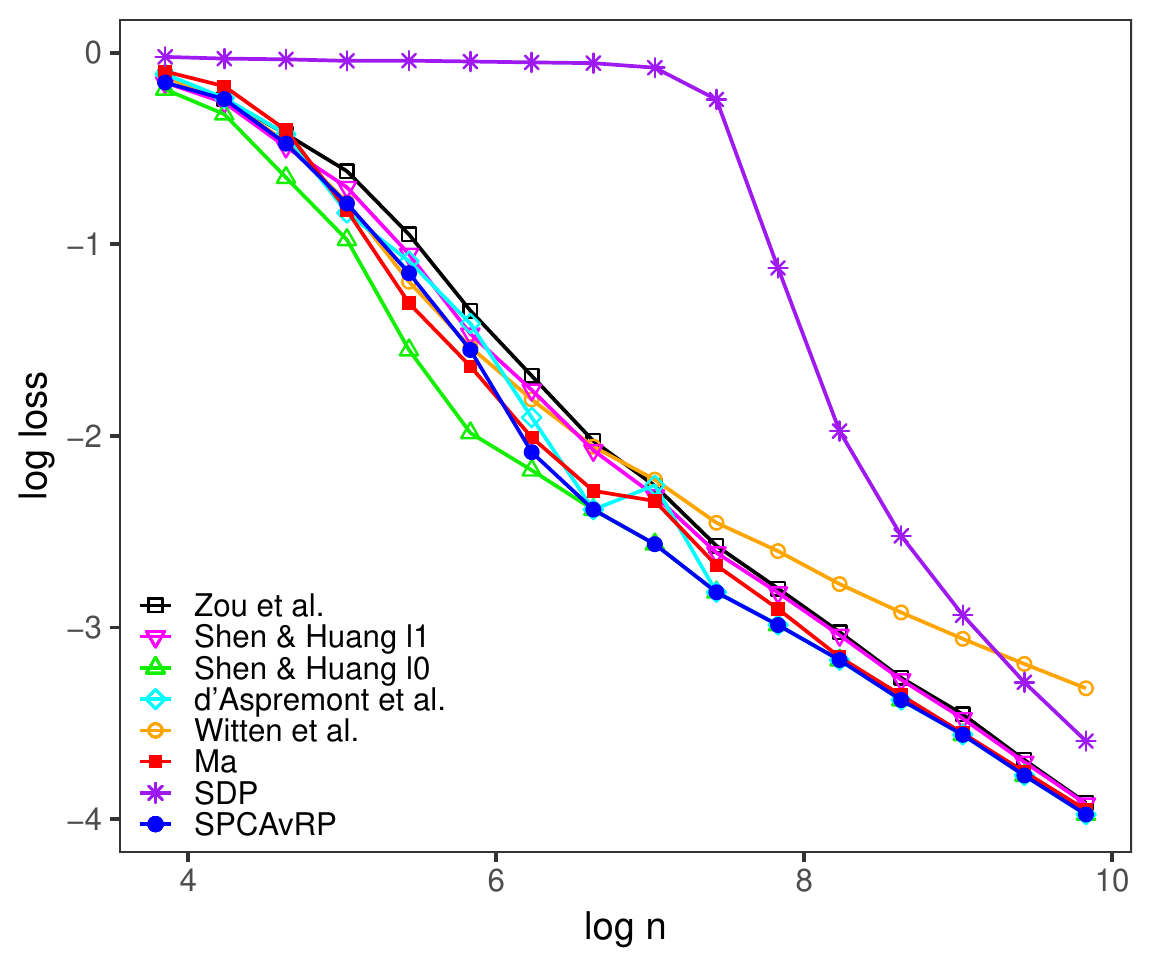} \includegraphics[scale=0.55]{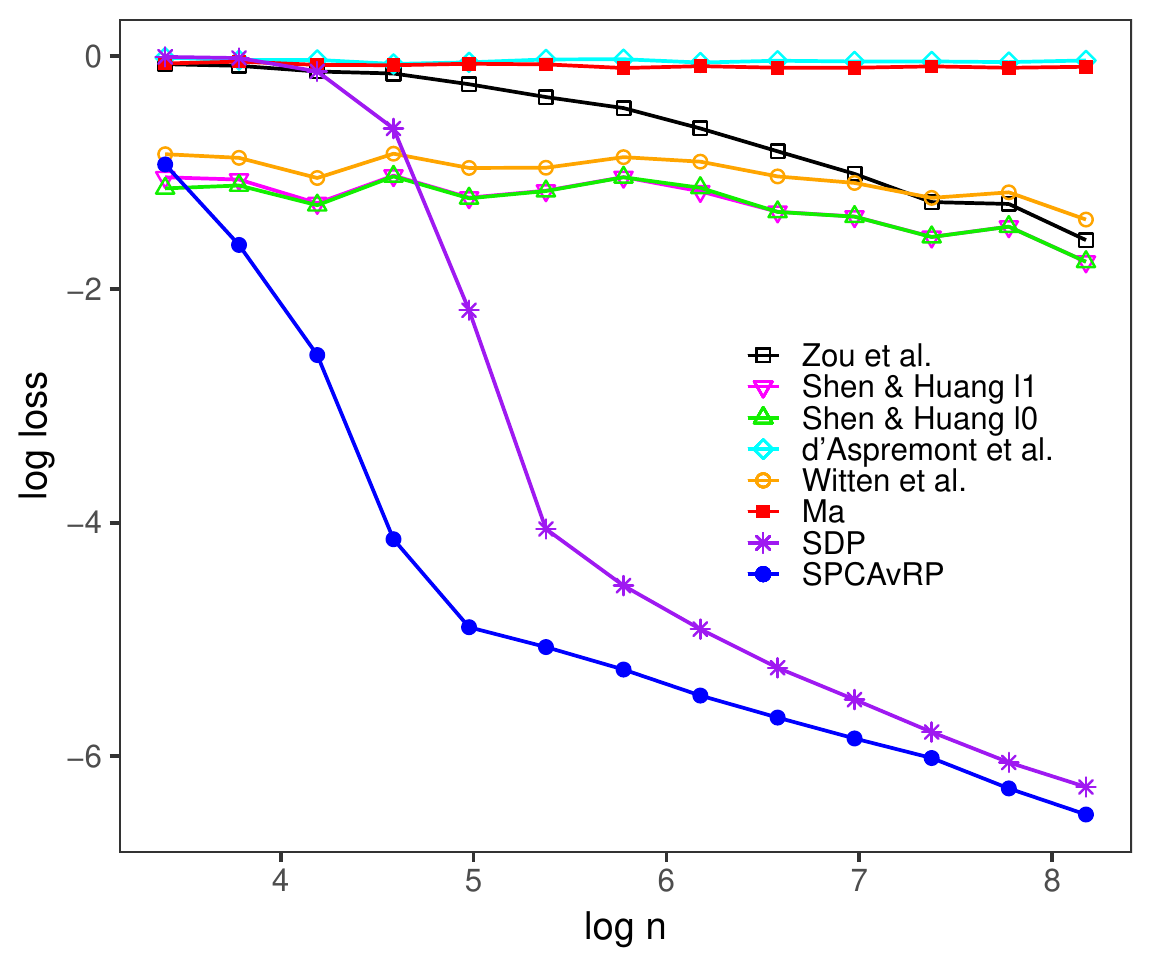}\\
\includegraphics[scale=0.55]{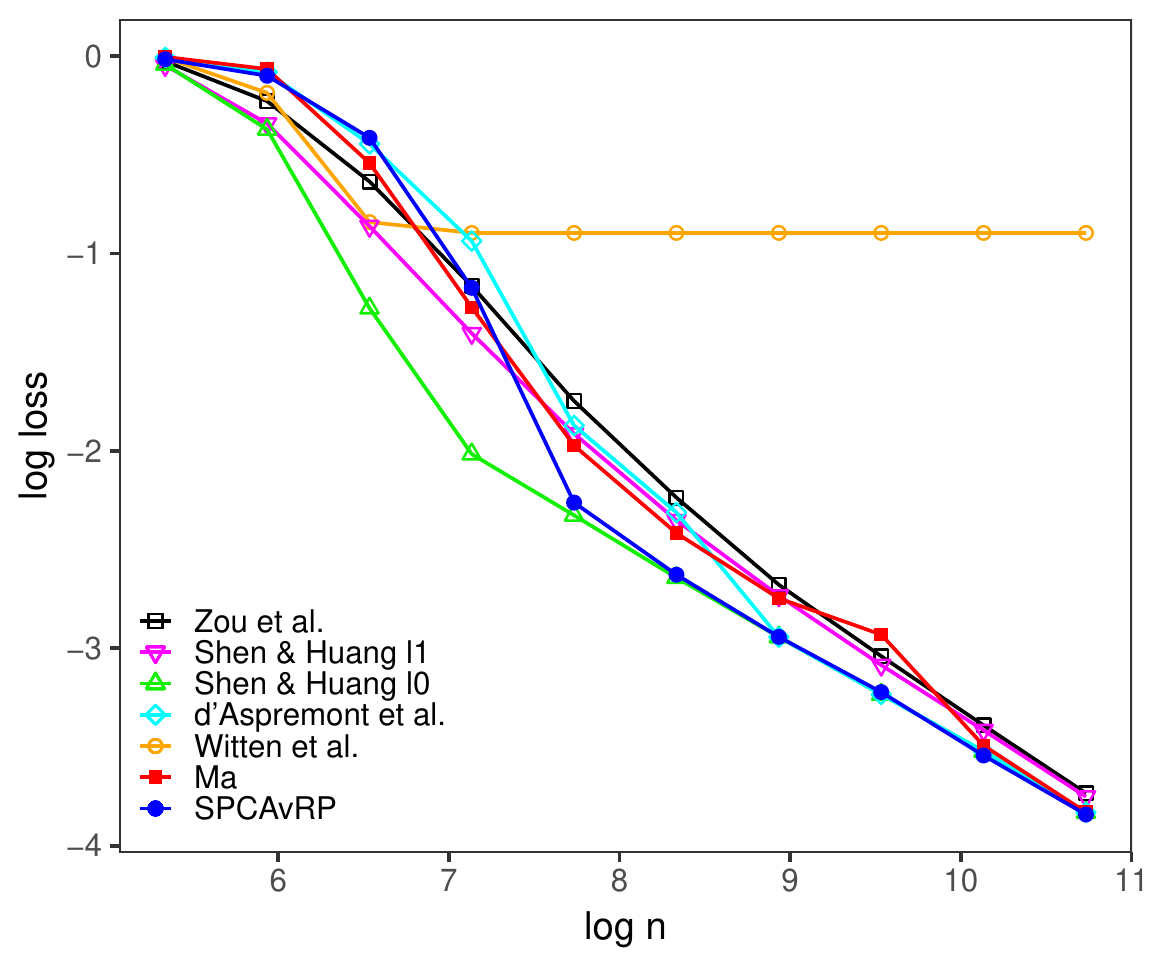} \includegraphics[scale=0.55]{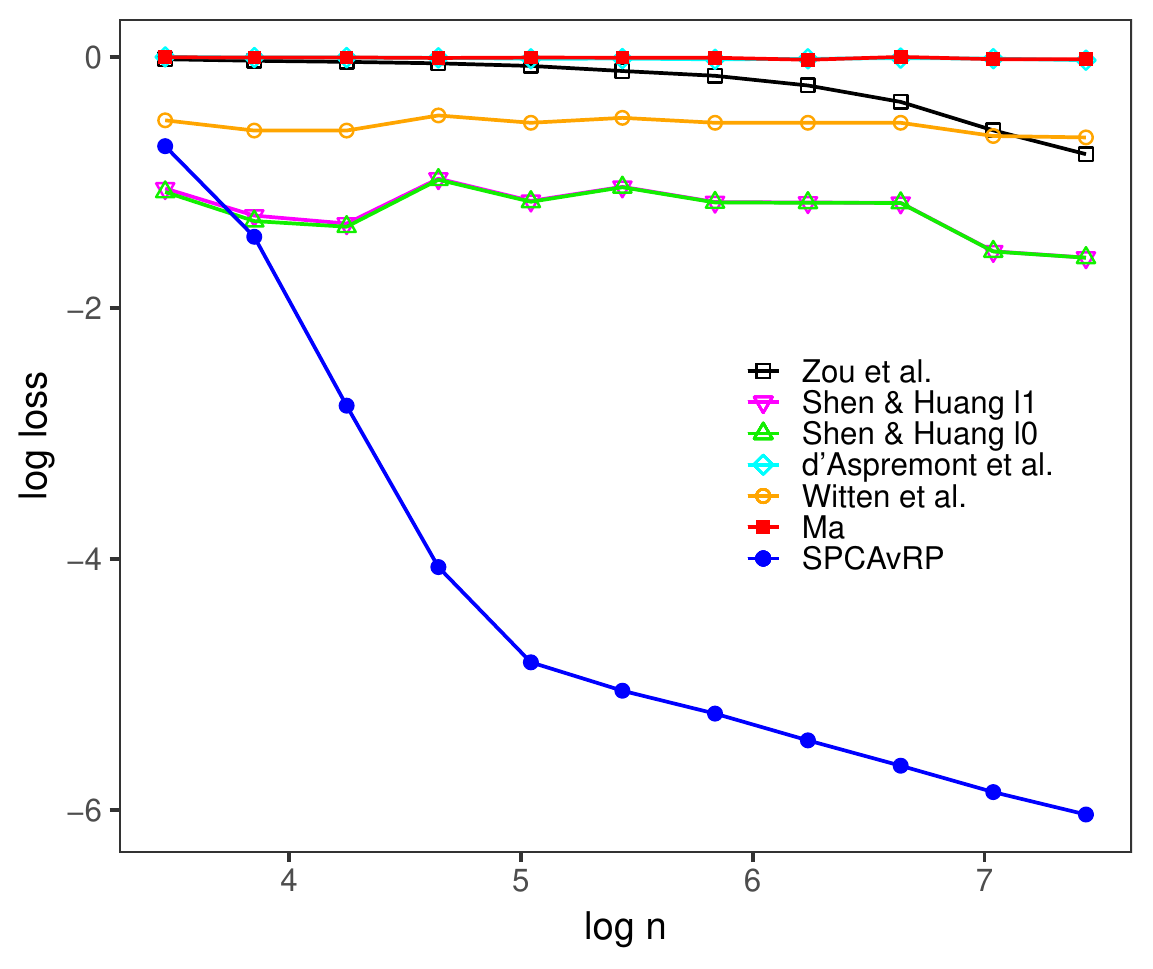}
 \caption{\textbf{Comparison of different principal component estimators.} Average loss against sample size $n$, on the log-log scale, using two different covariance structures from \R{mod_compare}: $\Sigma_{(1)}$ with $p=100$, $k=10$ (top left), $\Sigma_{(2)}$ with $p=200$, $k=10$ (top right), $\Sigma_{(1)}$ with $p=1000$, $k=30$ (bottom left), $\Sigma_{(2)}$ with $p=2000$, $k=30$ (bottom right). Blue: SPCAvRP with $A=300, B=150$ (top) or $A=800$, $B=300$ (bottom) and  $d=l=k$; black: \cite{Zou2006} with given $k$; magenta and green: \cite{Shen2008} with $\ell_1$ and $\ell_0$-thresholding, respectively, both with given $k$; cyan: \cite{d'Aspremont2008} with given $k$;   orange:  \cite{Witten2009} with parameters chosen by their default cross-validation; red: \cite{Ma2013} with the default parameters; purple: SDP.}
 \label{fig:compare}
\end{figure}

\subsubsection{Higher-order components}\label{sss:comparison-multi}

In Table~\ref{tab:sub} and Figure~\ref{fig:sub1} we compare Algorithms  \ref{Algo:SPCAvRPdefl} and \ref{Algo:SPCAvRPsub} with existing SPCA algorithms for subspace estimation, namely those proposed by \cite{Zou2006}, \cite{Witten2009} and \cite{Ma2013}. For this purpose we simulate observations from a normal distribution with a covariance matrix which is two- and three-spiked, respectively. 
From Table~\ref{tab:sub} and Figure~\ref{fig:sub1}, we observe that the SPCAvRP estimators computed by Algorithms \ref{Algo:SPCAvRPdefl} and \ref{Algo:SPCAvRPsub} perform well
when compared with the alternative approaches.  When the supports of leading eigenvectors are disjoint, namely $S_r \cap S_q = \emptyset$, $r\neq q$, $r,q\in[m]$,  where $S_r:=\{j\in[p]:v_r^{(j)}\neq0\}$, we observe that the deflation scheme proposed in Algorithm \ref{Algo:SPCAvRPdefl} may perform better than Algorithm \ref{Algo:SPCAvRPsub}, since it estimates each support $S_r$ individually. On the other hand, if their supports are overlapping,  Algorithm \ref{Algo:SPCAvRPsub} may perform better than Algorithm \ref{Algo:SPCAvRPdefl}, since it directly estimates $\cup_{r=1}^m S_r$.
From Table~\ref{tab:sub}, we also see that only SPCAvRP algorithms and the one proposed by \cite{Ma2013} compute components that are orthogonal in both cases $S_1\cap S_2 = \emptyset$ and $S_1\cap S_2 \neq \emptyset$.  

\begin{table}\caption{\label{tab:sub}\textbf{Comparison of different subspace estimators when $m=2$.} Observations are generated from $N_p(0,\Sigma)$, $\Sigma=I_p+\sum_{r=1}^2\theta_r v_rv_r^{\top}$, $\theta_1=50$, $\theta_2=30$,  $p=200$, $n=150$, where $v_1$ and $v_2$ have homogeneous signal strengths with $S_1 = \{1,\ldots,14\}$, and $S_2 = \{7,\ldots,20\}$ (top), $S_2 = \{15,\ldots,28\}$ (bottom).  The SPCAvRP estimators computed  by Algorithms~\ref{Algo:SPCAvRPdefl} and \ref{Algo:SPCAvRPsub}, with $A=300$, $B=150$, $m=2$, $d=\ell_1=\ell_2=k$, $\ell=|S_1\cup S_2|$, are compared with estimators computed by algorithms proposed by \cite{Zou2006}, \cite{Witten2009} and \cite{Ma2013}, which are used with their default parameters.}
\centering
\begin{tabular}{|l|cccc|}
\hline
 $S_1\cap S_2 \neq \emptyset$ &$L(\hat{V}_2,V_2)$ & $L(\hat{v}_1,v_1)$  & $L(\hat{v}_2,v_2)$ &$|\hat{v}^{\top}_1\hat{v}_2|$ \\ 
 \hline
 Algorithm~\ref{Algo:SPCAvRPdefl}     & $8.51\times 10^{-2}$ & $9.18\times 10^{-2}$ & $9.58\times 10^{-2}$ & $<10^{-15}$ \\
Algorithm~\ref{Algo:SPCAvRPsub}      & $6.72 \times 10^{-2}$ & $1.59\times 10^{-1}$ & $1.68\times 10^{-1}$  & $<10^{-15}$ \\
Ma           & $7.89\times 10^{-2}$ & $1.51\times 10^{-1}$ & $1.61\times 10^{-1}$ & $<10^{-15}$ \\
Witten et al. & $9.26\times 10^{-2}$ & $1.50\times 10^{-1}$ & $1.52\times 10^{-1}$ & $5.04 \times 10^{-4}$  \\
Zou et al.    & $1.80\times 10^{-1}$ & $2.06\times 10^{-1}$ & $2.23 \times 10^{-1}$  & $2.59 \times 10^{-4}$\\ 
\hline
\hline
 $S_1\cap S_2 = \emptyset$ &$L(\hat{V}_2,V_2)$ & $L(\hat{v}_1,v_1)$  & $L(\hat{v}_2,v_2)$ &$|\hat{v}^{\top}_1\hat{v}_2|$ \\  \hline
 Algorithm~\ref{Algo:SPCAvRPdefl}  & $5.42 \times 10^{-2}$ & $4.18 \times 10^{-2}$ & $5.32 \times 10^{-2}$ & $<10^{-15}$ \\
 Algorithm~\ref{Algo:SPCAvRPsub}      & $8.03 \times 10^{-2}$ & $1.64\times 10^{-1}$ & $1.75\times 10^{-1}$  & $<10^{-15}$ \\
Ma           & $8.91\times 10^{-2}$ & $1.43\times 10^{-1}$  & $1.53\times 10^{-1}$  & $<10^{-15}$ \\
Witten et al. & $8.97\times 10^{-2}$ & $1.11\times 10^{-1}$ & $1.09\times 10^{-1}$  & $1.36 \times 10^{-3}$ \\
Zou  et al.   & $9.97\times 10^{-2}$ & $7.13\times 10^{-2}$ & $9.62\times 10^{-2}$  & $<10^{-15}$ \\
\hline
\end{tabular}
\end{table}

\begin{figure}[htbp]
 \centering
\includegraphics[scale=0.55]{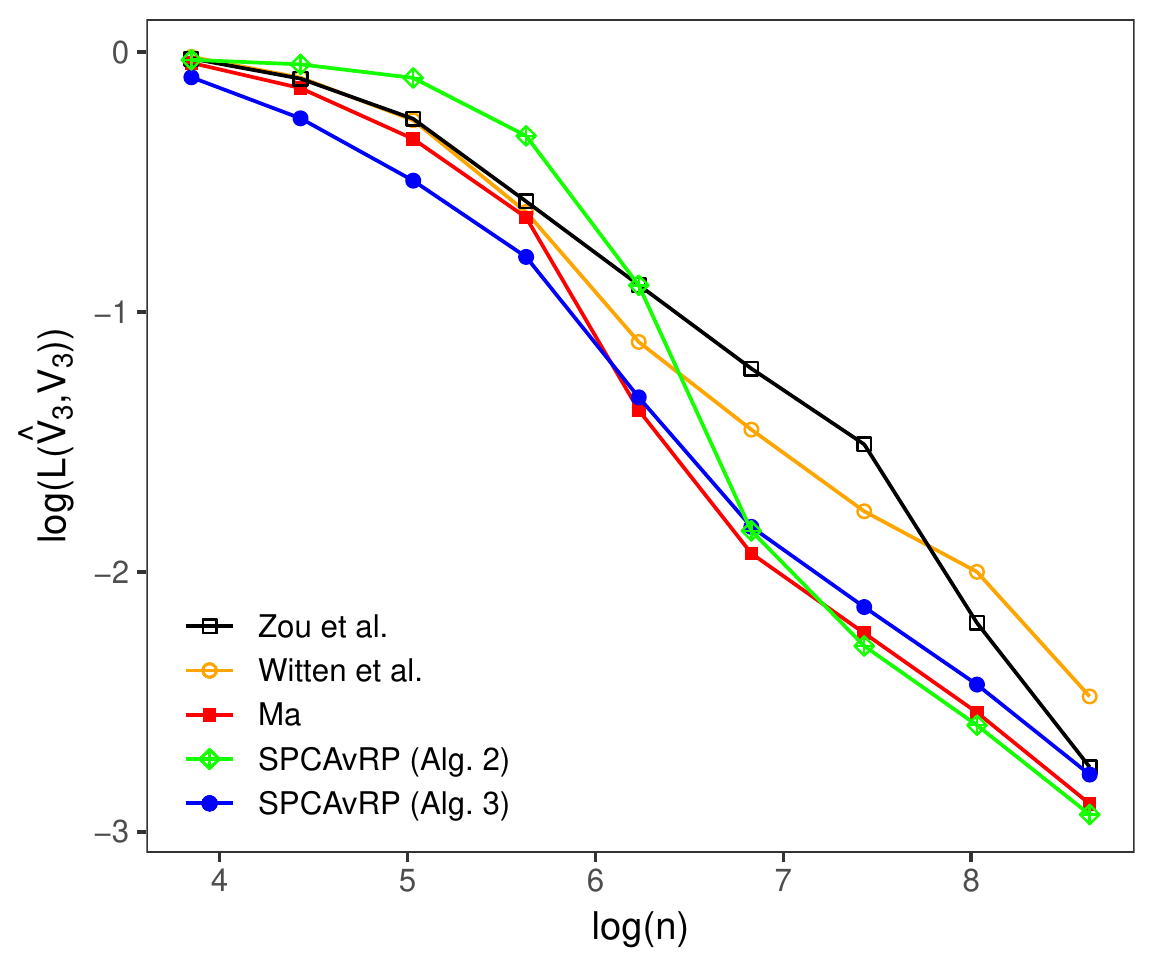} \includegraphics[scale=0.55]{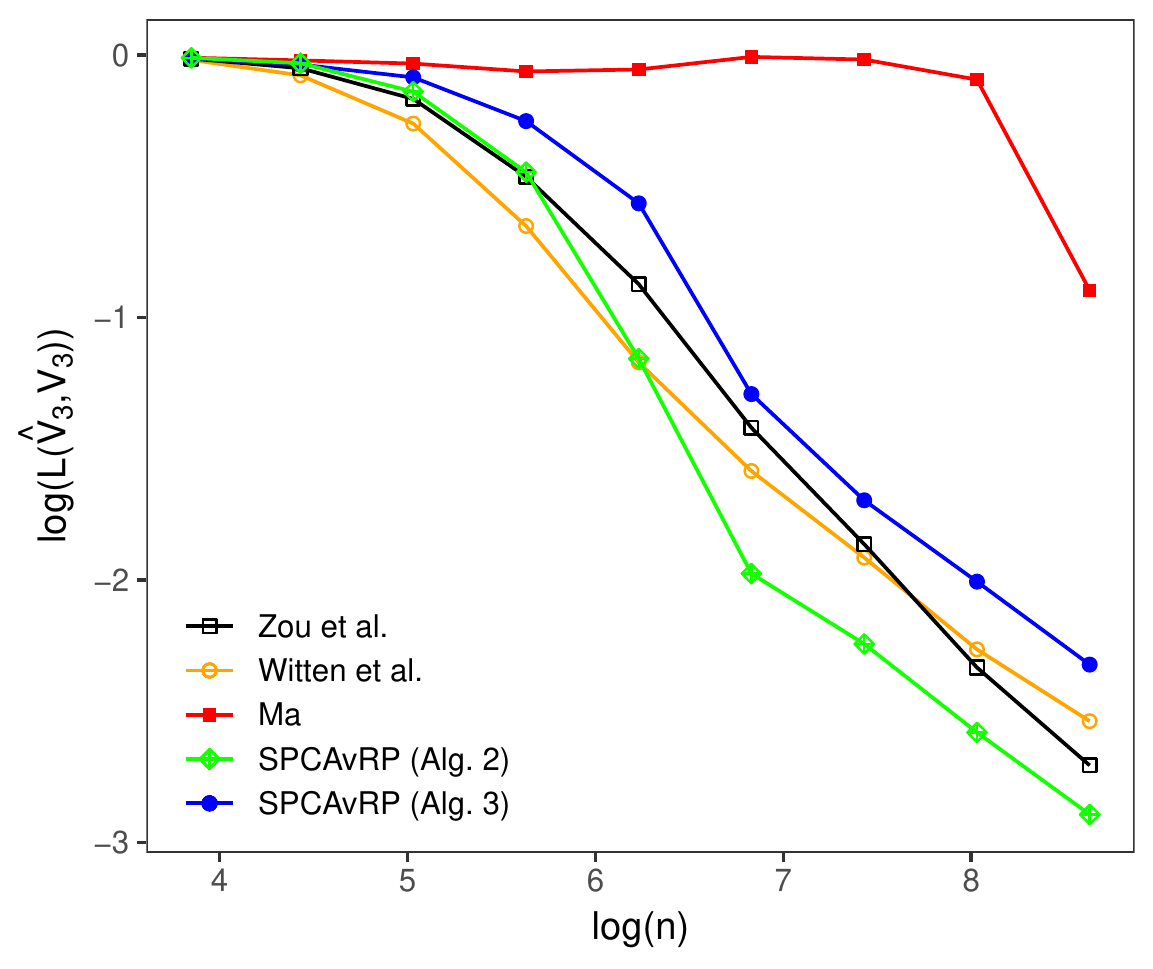} 
  \caption{\textbf{Comparison of different subspace estimators when $m=3$.} Average loss $L(\hat V_3,V_3)$ is plotted against sample size $n$, on the log-log scale.  Observations are generated from $N_p(0,\Sigma)$, $\Sigma=I_p+\sum_{r=1}^{3}\theta_r v_rv_r^{\top}$, $\theta_1=3$, $\theta_2=2$, $\theta_3=1$, $p=100$, where $v_1,v_2,v_3$ have homogeneous signals strengths with $S_1=\{1,\ldots,10\}$, $S_2=\{3,\ldots,12\}$, $S_3=\{5,\ldots,14\}$ (left) or $S_1=\{1,\ldots,10\}$, $S_2=\{11,\ldots,20\}$, $S_3=\{21,\ldots,30\}$ (right). SPCAvRP estimators computed by Algorithm \ref{Algo:SPCAvRPdefl} and Algorithm \ref{Algo:SPCAvRPsub}, with input parameters $A=400$, $B=200$, $m=3$, and $d=\ell_1=\ell_2=\ell_3=10$ (Algorithm \ref{Algo:SPCAvRPdefl}) or $d=\ell=|S_1\cup S_2\cup S_3|$ (Algorithm \ref{Algo:SPCAvRPsub}), are compared with subspace estimators computed by algorithms proposed by \cite{Zou2006}, \cite{Witten2009} and \cite{Ma2013}, with their default parameters.
}
  \label{fig:sub1}
\end{figure}


\subsubsection{Microarray data} 

We test our SPCAvRP algorithm on the \citet{Alon1999} gene expression data set, which contains 40 colon tumour and 22 normal observations.  A preprocessed data set can be downloaded from the \textsf{R} package `datamicroarray' \citep{R2016}, with a total of $p=2000$ features and $n=62$ observations. For comparison with alternative SPCA approaches, we use algorithms that accept the output sparsity $\ell$ as an input parameter, namely those proposed by \cite{Zou2006}, \cite{d'Aspremont2008} and \cite{Shen2008}.  For each $\ell$ considered, we computed the estimator $\hat{v}_{1,\ell}$ of the first principal component, and in Figure~\ref{fig:sparsitycurves} we plot the explained variance $\mathrm{Var}_{\ell}:=\hat{v}_{1,\ell}^\top\hat{\Sigma}\hat{v}_{1,\ell}$ as well as two different metrics for the separability of the two classes of observations projected along first principal component $\hat{v}_{1,\ell}$, namely the Wasserstein distance $W_\ell$ of order one and the $p$-value of Welch's $t$-test \citep{Welch1947}. Furthermore, in Figure~\ref{fig:boxplots}, we display their corresponding values for $\ell=20$ together with the box plots of the observations 
projected along $\hat{v}_{1,20}$. From Figures~\ref{fig:sparsitycurves} and \ref{fig:boxplots}, we observe that the SPCAvRP algorithm performs similarly to those proposed by \cite{d'Aspremont2008} and \cite{Shen2008}, all of which are superior in this instance to the SPCA algorithm of \cite{Zou2006}. In particular, for small values of $\ell$, we observe a steep slope of the blue Wasserstein and $p$-value curves corresponding to SPCAvRP algorithm in Figure~\ref{fig:sparsitycurves}, indicating that the two classes are well separated by projecting the observations along the estimated principal component which contains expression levels of only a few different genes. 

\begin{figure}[htbp]
 \centering
 \includegraphics[scale=0.45]{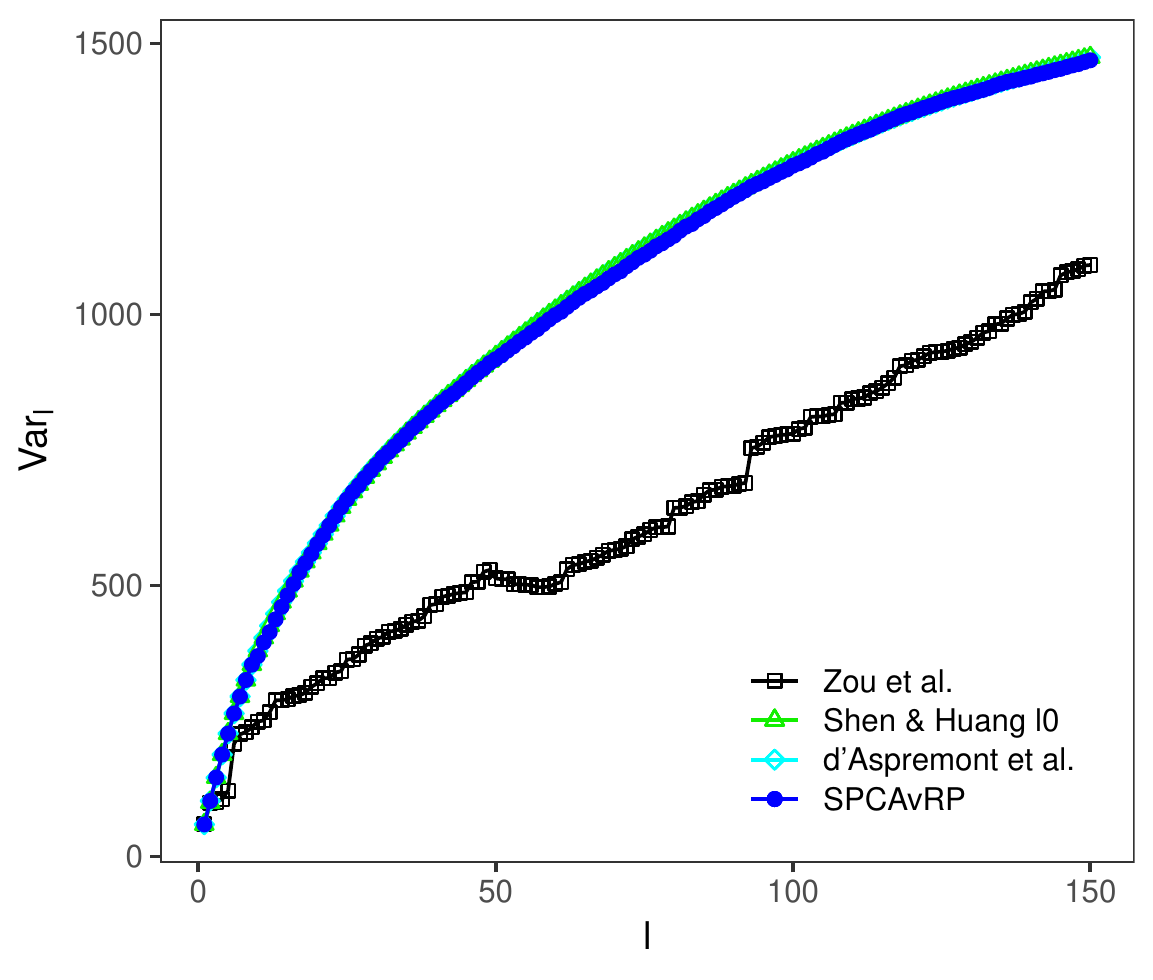}  \includegraphics[scale=0.45]{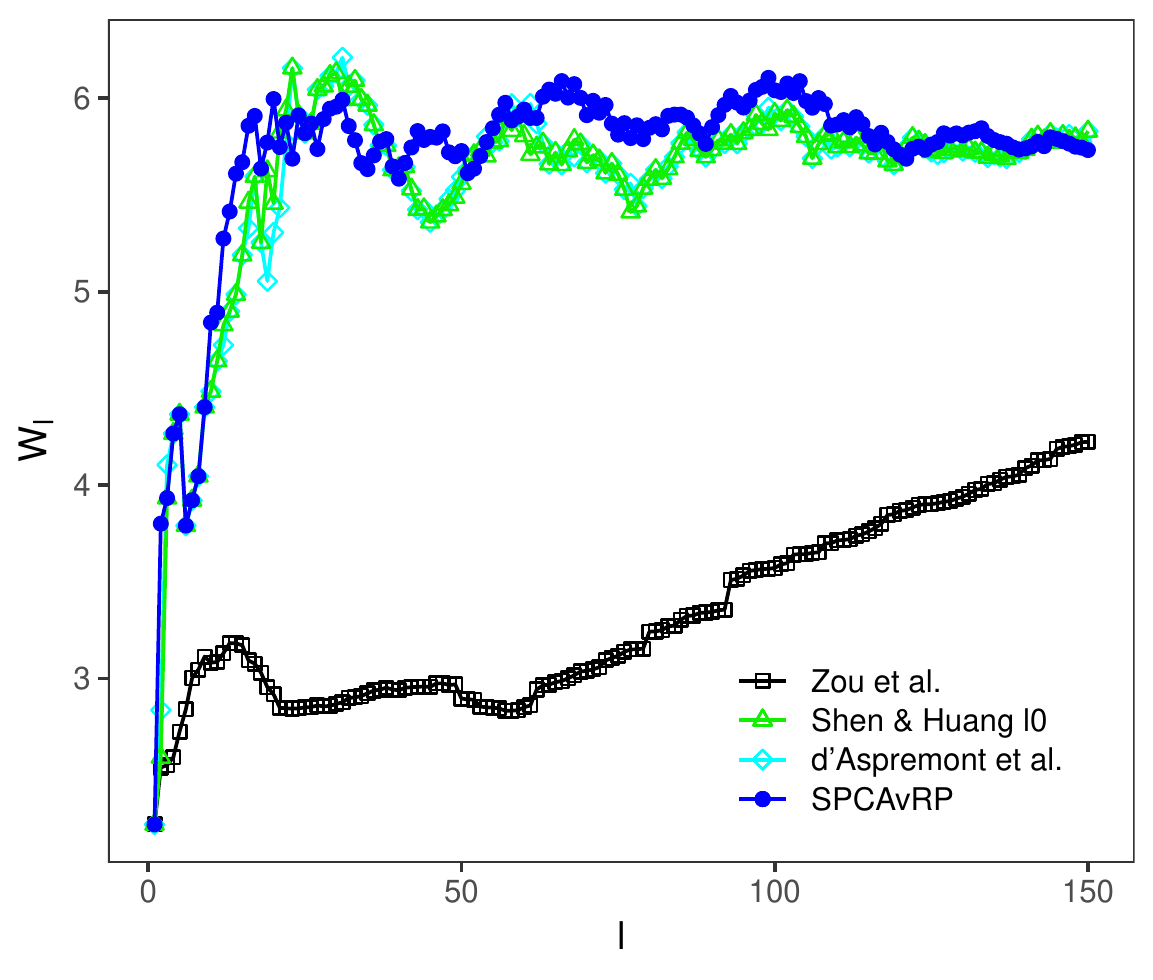} \includegraphics[scale=0.45]{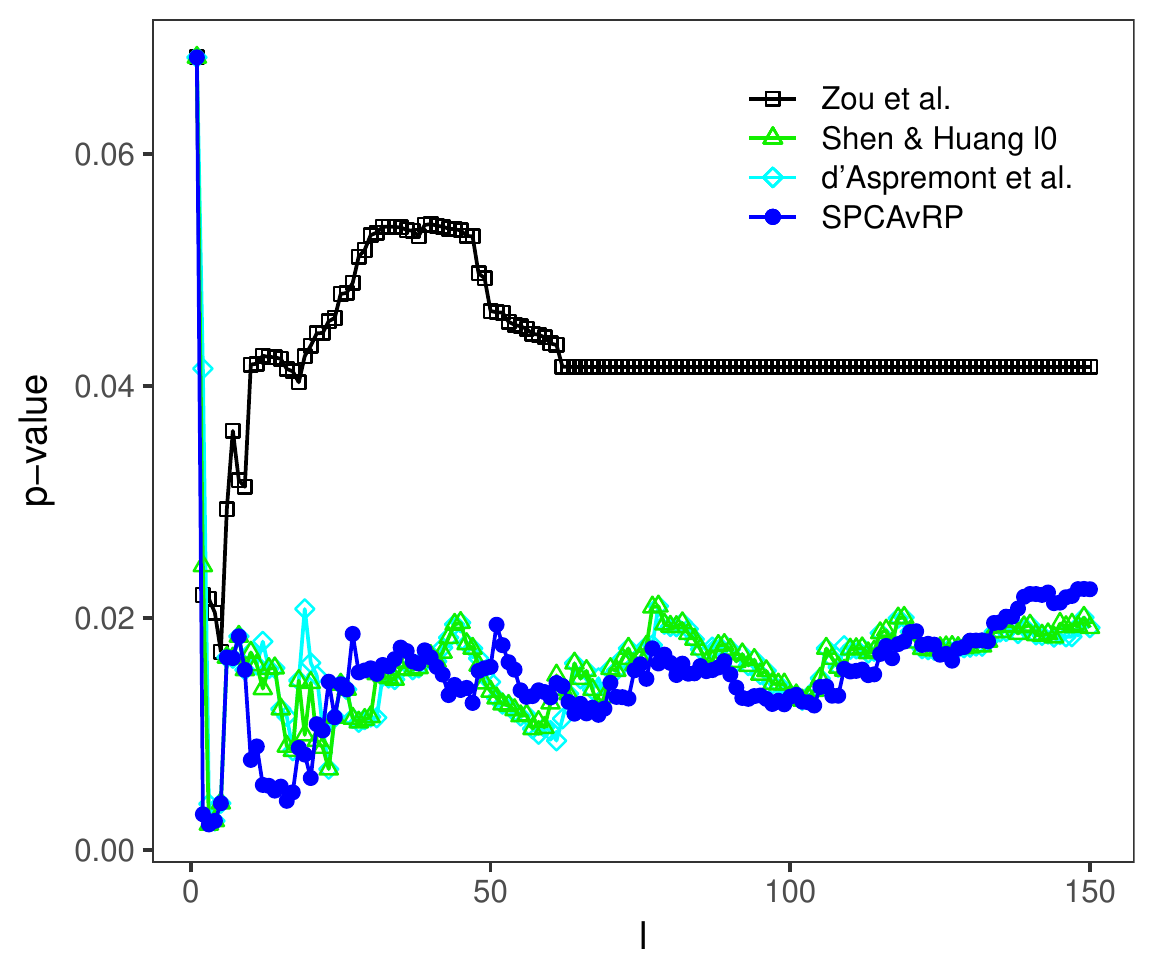}
  \caption{Left panel: $\mathrm{Var}_{\ell}$;  middle panel:  Wasserstein distance $W_\ell$ between the empirical distributions of the two classes projected along $\hat{v}_{1,\ell}$; right panel: $p$-value of Welch's t-test for the two classes projected along $\hat{v}_{1,\ell}$, where $\hat{v}_{1,\ell}$ is the estimator of $v_1$ for varied sparsity level $\ell$. For estimation we use SPCAvRP  ($d=30$, $A=1200$, $B=200$), \cite{Zou2006},  \cite{d'Aspremont2008} and \cite{Shen2008} with $\ell_0$-thresholding. 
}
\label{fig:sparsitycurves}
\end{figure}

\begin{figure}[htbp]
 \centering
\begin{tabular}{ccccccccccc}
 & & {\scriptsize PCA}  & & {\scriptsize Zou et al.}  & &  {\scriptsize d'Aspr.~et al.} & &  {\scriptsize Huang et al.}  & &  {\scriptsize SPCAvRP} \\
 & &{\scriptsize$\mathrm{Var}_{2000}= 1976.57$} & & {\scriptsize$\mathrm{Var}_{20}=319.70$} & &  {\scriptsize$\mathrm{Var}_{20}=577.81$} & &  {\scriptsize$\mathrm{Var}_{20}=577.27$} & &  {\scriptsize$\mathrm{Var}_{20}=576.52$} \\
 & &{\scriptsize$W_{2000}= 5.68$} & &  {\scriptsize$W_{20}=2.92$} & &  {\scriptsize$W_{20}=5.30$} & &   {\scriptsize$W_{20}= 5.46$} & &   {\scriptsize$W_{20}=6.00$} \\
 & & {\scriptsize $p$ = 0.0416} & & {\scriptsize $p$ = 0.0435} & &   {\scriptsize $p$ = 0.0161} & &   {\scriptsize $p$ = 0.0144}& &   {\scriptsize $p$ = 0.0062} 
\end{tabular}
\begin{tabular}{cc}
\rotatebox{90}{\hspace{4em}\scriptsize$x_i^{\top}\hat v_{1,20}$} &\includegraphics[scale=0.65]{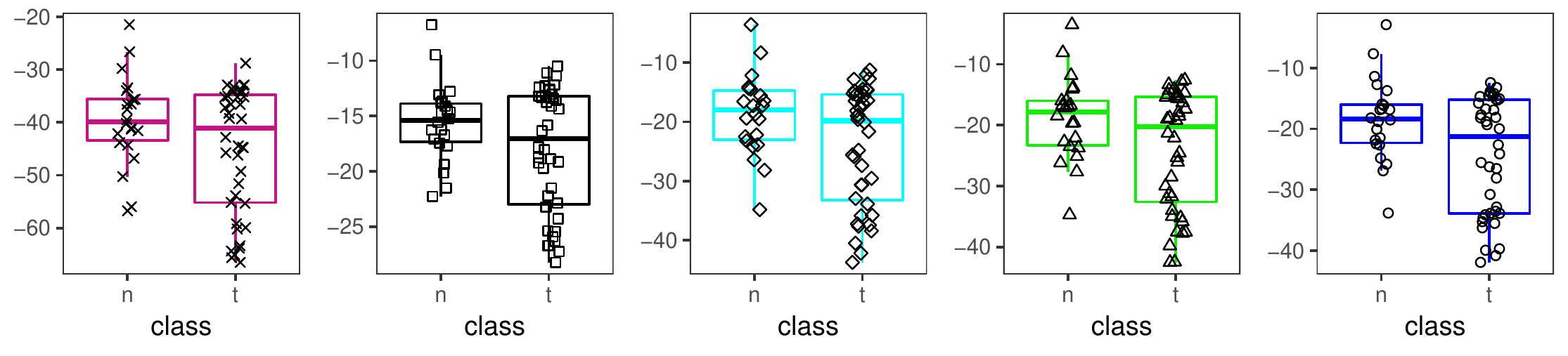}
\end{tabular}
  \caption{Variance $\mathrm{Var}_{\ell}$, Wasserstein distance $W_{\ell}$, $p$-value of the Welch's $t$-test and the corresponding box plots of the observations from the two classes projected along estimator $\hat{v}_{1,\ell}$ computed by five different approaches: classical PCA,  \cite{Zou2006},  \cite{d'Aspremont2008}, \cite{Shen2008} with $\ell_0$-thresholding, and SPCAvRP . The desired sparsity level in all SPCA algorithms is set to $\ell=20$.
}
\label{fig:boxplots}
\end{figure}


\section{Proofs of theoretical results}\label{ap:proofs}

\begin{proof}[of Lemma~\ref{Lemma:SparseOrthogonal}]
To verify that $\hat{v}_r$ is orthogonal to $\hat{v}_1,\ldots,\hat{v}_{r-1}$, observe that since the support of $\hat{v}_r$ is contained in $\tilde{S}_r$, we have
\[
  \hat{v}_r^\top \hat{V}_{r-1} = \hat{v}_r^\top P_{\tilde{S}_r}\hat{V}_{r-1} + \hat{v}_r^\top P_{\tilde{S}_r^c}\hat{V}_{r-1} = \frac{\hat{v}_r^\top H_{\tilde S_r}  P_{\tilde S_r} \hat{\Sigma}  P_{\tilde S_r} H_{\tilde S_r}P_{\tilde{S}_r}\hat{V}_{r-1}}{\lambda_1(H_{\tilde S_r}  P_{\tilde S_r} \hat{\Sigma}  P_{\tilde S_r} H_{\tilde S_r})} = 0,
\]
where the final equality follows from the fact that $H_{\tilde S_r}$ is a projection onto the orthogonal complement of the column space of $P_{\tilde{S}_r}\hat{V}_{r-1}$, so $H_{\tilde S_r}P_{\tilde{S}_r}\hat{V}_{r-1} = 0$.
\end{proof}

\begin{proof}[of Theorem~\ref{Thm:Main}]
For notational simplicity, we drop the subscript $m$ from $\hat{V}$ and $V$ in this proof, write $X := (X_1,\ldots,X_n)$ and define $\binom{[p]}{d}:=\{S\subseteq[p]:|S|=d\}$. For any $S\in\binom{[p]}{d}$, we note that $\Sigma^{(S,S)} = I_d + V^{(S,\cdot)}\Theta (V^{(S,\cdot)})^\top$ is a rank (at most) $m$ perturbation of the identity. Hence,
\begin{equation}
\label{Eq:EvalsSum}
\sum_{r=1}^m\lambda_r(\Sigma^{(S,S)}) = \mathrm{tr}(\Sigma^{(S,S)}) - (d-m) = m+\mathrm{tr}\bigl(V^{(S,\cdot)}\Theta(V^{(S,\cdot)})^\top\bigr) = m + \sum_{r=1}^m \sum_{j\in S\cap S_0} \theta_r (V^{(j,r)})^2.
\end{equation}
By the definition of $\text{RCC}_p(K)$ in~\eqref{Eq:RCC}, there is an event $\Omega_{\text{RCC}}$ with probability at least $1- 2p^{-3}$ such that on $\Omega_{\text{RCC}}$, we have 
\[
\sup_{u\in\mathcal{B}_0^{p-1}(d)} u^\top (\hat\Sigma-\Sigma)u \leq 2K\sqrt\frac{d\log p}{n} \quad \text{and} \quad \sup_{u\in\mathcal{B}_0^{p-1}(\ell)} u^\top (\hat\Sigma-\Sigma)u \leq 2K\sqrt\frac{\ell\log p}{n}.
\]
On $\Omega_{\text{RCC}}$, by~\eqref{Eq:EvalsSum}, Weyl's inequality \citep[Corollary~IV.4.9]{Weyl1912,StewartSun1990} and  \eqref{Eq:Cond2}, we have for any $S\in\binom{[p]}{d}$ that
\begin{align}
\label{Eq:EvalProjection2}
\biggl|\sum_{r=1}^m \lambda_r(\hat\Sigma^{(S,S)}) - m - \sum_{r=1}^m\sum_{j\in S\cap S_0} \theta_r (V^{(j,r)})^2\biggr| & = \biggl|\sum_{r=1}^m\bigl\{\lambda_r(\hat\Sigma^{(S,S)})-\lambda_r(\Sigma^{(S,S)})\bigr\}\biggr|\nonumber\\
& \leq 2Km\sqrt\frac{d\log p}{n} \leq \frac{m\theta_m}{16k\mu^2}.
\end{align}
By~\eqref{Eq:VMinVMax2}, we have
$\sum_{r=1}^m \theta_r (V^{(j,r)})^2\geq \theta_m\|V^{(j,\cdot)}\|_2^2\geq m\theta_m k^{-1}\mu^{-2}$ for every $j\in S_0$, which is more than twice the right-hand side of~\eqref{Eq:EvalProjection2}. Thus, an important consequence of~\eqref{Eq:EvalProjection2} is that on $\Omega_{\text{RCC}}$, for any $S,S'\in\binom{[p]}{d}$,
\begin{equation}
\label{Eq:Consequence2}
\text{if } S\cap S_0\subsetneq S'\cap S_0, \quad \text{then } \sum_{r=1}^m \lambda_r(\hat\Sigma^{(S,S)}) < \sum_{r=1}^m \lambda_r(\hat\Sigma^{(S',S')}).
\end{equation}

Fix $a\in[A]$, and for any $\tilde j \in[p]$ define $q_{\tilde j}:= \mathbb{P}(\tilde j\in S_{a,b^*(a)}\mid X)$. Now, fix some $j\in S_0$ and $j'\in [p]\setminus S_0$. We claim that
\begin{equation}
\label{Eq:CoordinateInclusionProbability2}
q_{j}\geq q_{j'} \quad \text{on $\Omega_{\text{RCC}}$}.
\end{equation} 
Before proving the claim, we first observe that, if~\eqref{Eq:CoordinateInclusionProbability2} holds, then since the same inequality would hold if we replace $j'$ by any other index in $S_0^{\text{c}}$, we would have on $\Omega_{\text{RCC}}$ that
\begin{equation}
\label{Eq:SignalCoordinateProbability2}
q_{j} \geq \frac{\sum_{\tilde j \in ([p]\setminus S_0)\cup \{j\}} q_{\tilde j}}{p-k+1} = \frac{d - \sum_{\tilde j\in S_0\setminus\{j\}}q_{\tilde j}}{p-k+1} \geq \frac{d-k+1}{p-k+1}\geq \frac{1}{p}.
\end{equation}
To verify the claim, define for $\tilde j\in\{j,j'\}$ and $b\in[B]$ the following sets:
\begin{align*}
\mathcal{S}_{b,\tilde j} &:= \bigl\{(S_{a,1},\ldots,S_{a,B}): b^*(a)=b,\ \tilde j\in S_{a,b}\bigr\} \quad\text{and}\quad \mathcal{S}_{b} := \bigl\{(S_{a,1},\ldots,S_{a,B}): b^*(a)=b \bigr\}.
\end{align*}
Let $\psi: \binom{[p]}{d}\to \binom{[p]}{d}$ be defined such that $\psi(S) := (S\setminus \{j'\}) \cup \{j\}$ if $j'\in S$ and $j\notin S$ and $\psi(S):=S$ otherwise.
Since, for every $S\in\binom{[p]}{d}$, we have either $\psi(S) = S$ or $S\cap S_0\subsetneq \psi(S)\cap S_0$, by~\eqref{Eq:Consequence2} we have on $\Omega_{\text{RCC}}$ that $\sum_{r=1}^m\lambda_r(\hat\Sigma^{(S,S)}) \leq \sum_{r=1}^m\lambda_r(\hat\Sigma^{(\psi(S),\psi(S))})$. Thus, for any $b\in[B]$ and any fixed $\hat\Sigma$ satisfying $\Omega_{\text{RCC}}$, the map $\psi$ induces an injection $\Psi: \mathcal{S}_{b,j'}\to\mathcal{S}_{b,j}$, given by
\begin{align*}
\Psi(S_{a,1},\ldots,S_{a,B}) := (S_{a,1},\ldots,S_{a,b-1},\psi(S_{a,b}),S_{a,b+1},\ldots,S_{a,B}),
\end{align*}
which in particular means that $|\mathcal{S}_{b,j'}|\leq |\mathcal{S}_{b,j}|$. Therefore, on $\Omega_{\text{RCC}}$, we have for all $b\in[B]$ that
\begin{align*}
\mathbb{P}(j\in S_{a,b^*(a)}&\mid X, b^*(a) = b) = \frac{\mathbb{P}(j\in S_{a,b^*(a)}, b^*(a)=b\mid X)}{\mathbb{P}(b^*(a)=b\mid X)} = \frac{|\mathcal{S}_{b,j}|}{|\mathcal{S}_{b}|}\\
&\geq  \frac{|\mathcal{S}_{b,j'}|}{|\mathcal{S}_{b}|}  = \frac{\mathbb{P}(j'\in S_{a,b^*(a)}, b^*(a)=b\mid X)}{\mathbb{P}(b^*(a)=b\mid X)}  = \mathbb{P}(j'\in S_{a,b^*(a)}\mid X, b^*(a) = b),
\end{align*}
and consequently $q_{j}\geq q_{j'}$ as claimed in~\eqref{Eq:CoordinateInclusionProbability2}. 

For $b\in[B]$ and $r\in[d]$, define $v_{a,b;r} := v_r(P_{a,b}\Sigma P_{a,b})$ and $\lambda_{a,b;r} := \lambda_r(P_{a,b}\Sigma P_{a,b})$. Note that $\lambda_{a,b;m+1}=\cdots=\lambda_{a,b;d} = 1$. Write $V_{a,b} := (v_{a,b;1},\ldots,v_{a,b;m})$, $\hat V_{a,b}:= (\hat v_{a,b;1},\ldots,\hat v_{a,b;m})$, $\Theta_{a,b} := \mathrm{diag}(\lambda_{a,b;1}-\lambda_{a,b;m+1}, \ldots,\lambda_{a,b;m}-\lambda_{a,b;m+1})$ and $\hat\Theta_{a,b} := \mathrm{diag}(\hat\lambda_{a,b;1}-\hat\lambda_{a,b;m+1}, \ldots,\hat\lambda_{a,b;m}-\hat\lambda_{a,b;m+1})$. By Lemma~\ref{Lem:Perturbation}, on $\Omega_{\text{RCC}}$, we have for all $\tilde j\in\{j,j'\}$ that
\begin{align}
\label{Eq:WeightDifference}
\bigl|(\hat V_{a,b^*(a)} \hat\Theta_{a,b^*(a)} &\hat V_{a,b^*(a)}^\top)^{(\tilde j,\tilde j)} - (V_{a,b^*(a)} \Theta_{a,b^*(a)} V_{a,b^*(a)}^\top)^{(\tilde j,\tilde j)}\bigr|\nonumber\\
& \leq 4m\|P_{a,b^*(a)} (\hat\Sigma - \Sigma) P_{a,b^*(a)}\|_{\mathrm{op}} \leq 8Km\sqrt\frac{d\log p}{n}\leq \frac{m\theta_m}{4k\mu^2},
\end{align}
where we used~\eqref{Eq:Cond2} in the last inequality.  Observe that
\begin{equation}
\label{Eq:WeightPopulation}
V_{a,b^*(a)} \Theta_{a,b^*(a)} V_{a,b^*(a)}^\top = \sum_{r=1}^d (\lambda_{a,b^*(a);r}-1) v_{a,b^*(a);r} v_{a,b^*(a);r}^\top = P_{a,b^*(a)}(\Sigma-I_p) P_{a,b^*(a)}.
\end{equation}
Also, we have
\begin{equation}
\label{Eq:WeightSample}
(\hat V_{a,b^*(a)} \hat\Theta_{a,b^*(a)} \hat V_{a,b^*(a)}^\top)^{(\tilde j,\tilde j)} = \sum_{r=1}^m(\hat\lambda_{a,b^*(a);r}-\hat\lambda_{a,b^*(a);m+1})(\hat v_{a,b^*(a);r}^{(\tilde j)})^2 =: \hat w_a^{(\tilde j)}.
\end{equation}
By~\eqref{Eq:VMinVMax2}, \eqref{Eq:WeightDifference}, \eqref{Eq:WeightPopulation} and~\eqref{Eq:WeightSample}, we have on $\Omega_{\text{RCC}}\cap \{j\in S_{a,b^*(a)}\}$ that
\begin{align}
\label{Eq:Rangej}
\frac{3m\theta_m}{4k\mu^2} \leq \theta_m\|V^{(j,\cdot)}\|_2^2 - \frac{m\theta_m}{4k\mu^2}&\leq \Sigma^{(j,j)}-1-\frac{m\theta_m}{4k\mu^2} \leq \hat w_a^{(j)}\nonumber\\
& \leq \Sigma^{(j,j)}-1+\frac{m\theta_m}{4k\mu^2}\leq \theta_1\|V^{(j,\cdot)}\|_2^2+\frac{m\theta_m}{4k\mu^2}\leq \frac{5m\theta_1\mu^2}{4k}.
\end{align}
Moreover, on $\Omega_{\text{RCC}}\cap \{j'\in S_{a,b^*(a)}\}$, we have
\begin{equation}
\label{Eq:Rangejprime}
-\frac{m\theta_m}{4k\mu^2}\leq \hat w_a^{(j')} \leq \frac{m\theta_m}{4k\mu^2}
\end{equation}
Recall that for all $j\in[p]$, if $j\notin S_{a,b^*(a)}$, then $\hat w_a^{(j)} = 0$.  Combining the lower bound on $\hat w_a^{(j)}$ in~\eqref{Eq:Rangej} and the upper bound on $\hat w_a^{(j')}$ in~\eqref{Eq:Rangejprime}, we have by~\eqref{Eq:CoordinateInclusionProbability2} and~\eqref{Eq:SignalCoordinateProbability2} that on $\Omega_{\text{RCC}}$,
\begin{equation}
\label{Eq:Gap}
\mathbb{E}\bigl(\hat w_a^{(j)} - \hat w_a^{(j')}\bigm | X\bigr) = \mathbb{E}\bigl(\hat w_a^{(j)}\mathds{1}_{\{j\in S_{a,b^*(a)}\}} - \hat w_a^{(j')}\mathds{1}_{\{j'\in S_{a,b^*(a)}\}}\bigm| X\bigr)\geq \frac{q_j m\theta_m}{2k\mu^2} \geq \frac{m\theta_m}{2pk\mu^2}.
\end{equation}

Now, let $a, j,j'$ be freely varying again, and define $\Omega:=\{\min_{j\in S_0} \hat{w}^{(j)} > \max_{j\notin S_0}\hat{w}^{(j)}\}$. Since~\eqref{Eq:Gap} holds for arbitrary $j\in S_0$ and $j'\notin S_0$, and since $\hat w^{(j)} = A^{-1}\sum_{a=1}^A \hat w_a^{(j)}$, we have 
\[
\Omega^{\text{c}} \subseteq \bigcup_{j\in S_0} \biggl\{\hat w^{(j)} - \mathbb{E}(\hat w^{(j)}\mid X) \leq -\frac{m\theta_m}{4pk\mu^2} \biggr\}\cup \bigcup_{j'\notin S_0} \biggl\{\hat w^{(j')} - \mathbb{E}(\hat w^{(j')}\mid X) \geq \frac{m\theta_m}{4pk\mu^2} \biggr\}.
\]
Observe that $(\hat w_{a}^{(j)}:a\in[A])$ are independent and identically distributed conditional on $X$.  By~\eqref{Eq:Rangej} and~\eqref{Eq:Rangejprime}, $\hat w_a^{(j)}$ is bounded on $\Omega_{\text{RCC}}$ for all $j\in [p]$. Thus, we can use a union bound and apply Hoeffding's inequality conditional on $X$ to obtain that on $\Omega_{\text{RCC}}$,
\begin{align*}
\mathbb{P}(\Omega^{\text{c}}\mid X)&\leq p \exp\biggl\{\frac{-A}{2}\biggl(\frac{m\theta_m}{4pk\mu^2}\biggr)^2\biggm/ \biggl(\frac{5m\theta_1\mu^2}{4k}\biggr)^2\biggr\} \leq pe^{-A\theta_m^2/(50p^2\mu^8\theta_1^2)}. 
\end{align*}
Since $\ell\geq k$, on $\Omega$, we have $\hat S\supseteq S_0$. Therefore, by \citet[Theorem~2]{Yu2015}, on $\Omega_{\text{RCC}}\cap\Omega$, 
\begin{align*}
L(\hat V, V) \leq \frac{2m^{1/2}\|P_{\hat S} (\hat \Sigma - \Sigma)P_{\hat S}\|_{\text{op}}}{\theta_m} \leq 4K\sqrt\frac{m\ell\log p}{n\theta_m^2}.
\end{align*}
The desired result follows from the fact that
\[
\mathbb{P}(\Omega_{\text{RCC}}\cap\Omega) \geq 1 - \mathbb{P}(\Omega_{\text{RCC}}^{\text{c}}) - \mathbb{E}\{\mathbb{P}(\Omega^{\text{c}}\mid X)\mathds{1}_{\Omega_{\text{RCC}}}\} \geq 1 - 2p^{-3} - pe^{-A\theta_m^2/(50p^2\mu^8\theta_1^2)}.
\] 
\end{proof}

\begin{proof}[of Proposition~\ref{Prop:LowerBound}]
Let $\mathbb{O}_{p,m,k}:=\{V\in\mathbb{O}_{p,m}: \mathrm{nnzr}(V)\leq k\}$. Writing $k = qm+h$ for $q\in\mathbb{N}$ and $h\in\{0,\ldots,m-1\}$, for $r\in[m]$, we define 
\[
u_r := \begin{cases}(q+1)^{-1/2}(\mathbf{0}_{(r-1)(q+1)}^\top, \mathbf{1}_{q+1}^\top, \mathbf{0}_{p-r(q+1)}^\top)^\top & \text{if $1\leq r\leq h$,}\\
q^{-1/2}(\mathbf{0}_{h(q+1)+(r-h-1)q}^\top, \mathbf{1}_{q}^\top, \mathbf{0}_{p-h(q+1)-(r-h)q}^\top)^\top  & \text{if $h+1\leq r\leq m$.}\end{cases}
\]
and write $U:=(u_1,\ldots,u_m)\in\mathbb{R}^{p\times m}$. By construction, $U^\top U = I_m$, so there exists $\tilde U\in\mathbb{O}_p$ whose first $m$ columns are $U$.  Moreover, for $j\in[k]$, we have 
\begin{equation}
\label{Eq:RowNorm}
\frac{4m}{5k} \leq \frac{m}{k+m} \leq \frac{1}{q+1} \leq \|U^{(j,\cdot)}\|_2^2 \leq \frac{1}{q}\leq \frac{m}{k-m}\leq \frac{4m}{3k}.
\end{equation}
Now, fix some $\epsilon \in \bigl(0, \sqrt{m/(16k)}\bigr]$ to be specified later. For any $J\in\mathbb{O}_{p-m, m, k-m}$, define
\[
V_J:= \tilde U \begin{pmatrix} \sqrt{1-\epsilon^2}I_m\\ \epsilon J\end{pmatrix} = U + \tilde U\begin{pmatrix} \bigl(\sqrt{1-\epsilon^2}-1\bigr)I_m\\ \epsilon J\end{pmatrix} =: U + \tilde U\Delta_J.
\]
For any $M\in\mathbb{R}^{p\times m}$, we define its \emph{two-to-infinity norm} as $\|M\|_{2\to\infty}:=\sup_{v\in\mathcal{S}^{m-1}}\|Mv\|_{\infty} = \max_{j\in[p]}\|M^{(j,\cdot)}\|_2$.  Then for $J\in\mathbb{O}_{p-m, m, k-m}$, we have 
\begin{equation}
\label{Eq:RowNormPerturbation}
\|V_J - U\|_{2\to\infty} \leq \|\tilde U\|_{2\to\infty}\|\Delta_J\|_{\mathrm{op}} = \|\Delta_J^\top\Delta_J\|_{\mathrm{op}}^{1/2} \leq \sqrt{2}\epsilon.
\end{equation}
Combining equations~\eqref{Eq:RowNorm} and~\eqref{Eq:RowNormPerturbation}, and since $\epsilon\leq \sqrt{m/(16k)}$, we have that $\|V_J^{(j,\cdot)}\|_2\in[0.54(m/k)^{1/2}, 1.51(m/k)^{1/2}]$ for all $j\in[k]$, which implies that $V_J \in \mathbb{O}_{p,m,k}(3)$. 

Using the definition of $V_J$ and the triangle inequality, we have that for any $J,J'\in\mathbb{O}_{p-m,m,k-m}$,
\begin{equation}
\label{Eq:VtopV}
\|V_J^\top V_{J'}\|_{\mathrm{F}} = \|(1-\epsilon^2)I_m + \epsilon^2 J^\top J'\|_{\mathrm{F}} \geq (1-\epsilon^2)\|I_m\|_{\mathrm{F}} - \epsilon^2 \sqrt{m} \|J^\top J'\|_{\mathrm{op}} = (1 - 2\epsilon^2)\sqrt{m}.
\end{equation}
Writing $D_{\mathrm{KL}}(P \,\|\, Q)$ for the Kullback--Leibler divergence from a distribution $P$ to a distribution $Q$ and $\Sigma_J := I_p + \theta V_JV_J^\top$, we have for any $J,J'\in\mathbb{O}_{p-m,m,k-m}$ that
\begin{align}
\label{Eq:KLRadius}
D_{\mathrm{KL}}&\bigl(N_p(0,\Sigma_J) \bigm\| N_p(0,\Sigma_{J'})\bigr) = \frac{1}{2}\mathrm{tr}\bigl(\Sigma_{J'}^{-1}\Sigma_J-I_p\bigr) = \frac{\theta}{2}\mathrm{tr}\bigl\{(I_p+\theta V_{J'}V_{J'}^\top)^{-1}(V_JV_J^\top - V_{J'}V_{J'}^\top)\bigr\}\nonumber\\
&= \frac{\theta}{2}\mathrm{tr}\biggl\{\biggl(I_p - \frac{\theta}{1+\theta}V_{J'} V_{J'}^\top\biggr)(V_JV_J^\top - V_{J'}V_{J'}^\top)\biggr\} =\frac{\theta^2}{2(1+\theta)}\{m - \mathrm{tr}(V_{J'}V_{J'}^\top V_JV_J^\top)\} \nonumber\\
&= \frac{\theta^2}{2(1+\theta)}\bigl(m - \|V_J^\top V_{J'}\|_{\mathrm{F}}^2\bigr)\leq \frac{2m\epsilon^2\theta^2}{1+\theta},
\end{align}
where we used~\eqref{Eq:VtopV} in the final inequality. On the other hand, we also have
\begin{align}
\label{Eq:PackingDistance}
L(V_J, V_{J'}) &= \frac{1}{\sqrt{2}}\|V_JV_J^\top - V_{J'}V_{J'}^\top\|_{\mathrm{F}} = \bigl\{\epsilon^4 L^2(J,J') + \epsilon^2(1-\epsilon^2) \|J - J'\|_{\mathrm{F}}^2\bigr\}^{1/2}\geq \epsilon L(J,J'),
\end{align}
where we used~\citet[Proposition~2.2]{Vu2013} in the last inequality. Thus, if we can find some finite subset $\mathcal{J}\subseteq \mathbb{O}_{p-m,m,k-m}$ such that $3\leq |\mathcal{J}|\leq e^{nm^2\theta^2/k}$ and $\min_{J,J'\in\mathcal{J}: J\neq J'}L(J, J')\geq cm^{1/2}$ for some universal constant $c>0$, then by~\eqref{Eq:KLRadius}, \eqref{Eq:PackingDistance} and Fano's lemma \citep[see, e.g.][Lemma~3]{Yu1997}, we have
\begin{align*}
\inf_{\tilde{V}}\sup_{V \in \mathbb{O}_{p,m,k}(3)}&\mathbb{E}_{P_{V,\theta}} L(\tilde V, V) \geq \inf_{\tilde{V}}\max_{J \in \mathcal{J}}\mathbb{E}_{P_{V_J, \theta}} L(\tilde V, V_J) \\
&\geq \frac{cm^{1/2}\epsilon}{2}\biggl(1-\frac{2nm\epsilon^2\theta^2/(1+\theta)+\log 2}{\log|\mathcal{J}|}\biggr) \geq \frac{cm^{1/2}\epsilon}{2}\biggl(\frac{1}{3}-\frac{2nm\epsilon^2\theta^2}{\log|\mathcal{J}|}\biggr),
\end{align*}
where we used the fact that $|\mathcal{J}|\geq 3$ in the final inequality. 
Choosing $\epsilon = \sqrt\frac{\log |\mathcal{J}|}{16nm\theta^2}$ (noting that the condition $\log |\mathcal{J}|\leq nm^2\theta^2/k$ ensures that $\epsilon\leq \sqrt{m/(16k)}$), we obtain
\begin{equation}
\label{Eq:Fano}
\inf_{\tilde{V}}\sup_{V \in \mathbb{O}_{p,m,k}(3)}\mathbb{E}_{P_{V,\theta}} L(\tilde V, V) \geq \frac{cm^{1/2}\epsilon}{10} \gtrsim \sqrt\frac{\log|\mathcal{J}|}{n\theta^2}.
\end{equation}

It remains to construct a suitable $\mathcal{J}$. By \citet{Szarek1982} \citep[see also][Proposition~8]{Pajor1998}, there exists a finite subset $\tilde{\mathcal{J}}\subseteq \mathbb{O}_{k-m,m}$ such that $|\tilde {\mathcal{J}}| =  \lfloor e^{m(k-2m)}\rfloor$ and $L(\tilde J, \tilde J')\geq cm^{1/2}$ for all distinct $\tilde J, \tilde J'\in\tilde{\mathcal{J}}$. Define $\mathcal{J} := \{(\tilde J^\top, \mathbf{0}_{(p-k)\times m}^\top)^\top: \tilde J \in\tilde{\mathcal{J}}\}$. We have $\min_{J,J'\in\mathcal{J}: J\neq J'}L(J, J') = \min_{\tilde J,\tilde J'\in\tilde{\mathcal{J}}: \tilde J\neq \tilde J'}L(\tilde J, \tilde J') \geq cm^{1/2}$ and $|\mathcal{J}| = |\tilde{\mathcal{J}}|$.  Since $k\geq 4m$ and $nm\theta^2\geq k^2$, we have $3\leq |\mathcal{J}|\leq e^{nm^2\theta^2/k}$ as desired. Hence, by~\eqref{Eq:Fano}, 
\begin{equation}
\label{Eq:LowerBound1}
\inf_{\tilde{V}}\sup_{V \in \mathbb{O}_{p,m,k}(3)}\mathbb{E}_{P_{V,\theta}} L(\tilde V, V)\gtrsim \sqrt\frac{mk}{n\theta^2}.
\end{equation}
Alternatively, we can also construct $\mathcal{J}$ as follows. Recall the definition of $\binom{[p-m]}{k}$ from the proof of Theorem~\ref{Thm:Main}. For any $S\in\binom{[p-m]}{k}$, define $J_S\in\mathbb{R}^{(p-m)\times m}$ such that $J_S^{(S,\cdot)} = U^{([k],\cdot)}$ and $J_S^{(S^{\mathrm{c}},\cdot)} = \mathbf{0}$. By the Gilbert--Varshamov Lemma \citep[see, e.g.][Lemma~4.10]{Massart2007}, and since $p \geq 5k$, there exists $\mathcal{S}\subseteq \binom{[p-m]}{k}$ such that $|\mathcal{S}| = \lfloor e^{\frac{1}{15}k\log((p-m)/k)}\rfloor$ and for any distinct $S, S'\in\mathcal{S}$, $|S\cap S'|\leq k/2$. Let $\mathcal{J} := \{J_S: S\in \mathcal{S}\}$.  Then $|\mathcal{J}| = |\mathcal{S}|$ and 
\[
\min_{J,J'\in\mathcal{J}: J\neq J'}L(J,J')=\min_{J,J'\in\mathcal{J}: J\neq J'}(m-\|J^\top J'\|_{\mathrm{F}}^2)^{1/2} \geq \biggl(m - \frac{k}{2q}\biggr)^{1/2} \geq \sqrt\frac{m}{3},
\]
where the final inequality uses~\eqref{Eq:RowNorm}. Since $k\log ((p-m)/k)\geq 17$ and $nm^2\theta^2\geq k^2\log(p/k)$, we have $3\leq |\mathcal{J}|\leq e^{nm^2\theta^2/k}$ as desired. Hence, by~\eqref{Eq:Fano},
\begin{equation}
\label{Eq:LowerBound2}
\inf_{\tilde{V}}\sup_{V \in \mathbb{O}_{p,m,k}(3)}\mathbb{E}_{P_{V,\theta}} L(\tilde V, V)\gtrsim \sqrt\frac{k\log(p/k)}{n\theta^2}.
\end{equation}
We complete the proof by combining~\eqref{Eq:LowerBound1} and~\eqref{Eq:LowerBound2}.
\end{proof}

\begin{proof}[of Corollary~\ref{Cor:Homogeneous}]
The proof of Theorem~\ref{Thm:Main} remains valid for the setting of this corollary.  Fix a specific $a\in[A]$. Since $V_m\in\mathbb{O}_{p,m,k}(1)$ and $\theta_1 = \cdots = \theta_m$, we have by~\eqref{Eq:EvalProjection2} that on $\Omega_{\text{RCC}}$, for any $S, S'\in\binom{[p]}{d}$, if $|S\cap S_0| < |S'\cap S_0|$, then $\sum_{r=1}^m\lambda_r(\hat\Sigma^{(S,S)}) < \sum_{r=1}^m\lambda_r(\hat\Sigma^{(S',S')})$. Thus, in particular, $|S_{a,b^*(a)}\cap S_0| = \max_{b\in[B]}|S_{a,b}\cap S_0|$ on $\Omega_{\text{RCC}}$. 

Observe that $|S_{a,b}\cap S_0|\stackrel{\mathrm{iid}}{\sim} \mathrm{HyperGeom}(d,k,p)$. Let $M:=\max_{b\in[B]}|S_{a,b}\cap S_0|$ and $R:=|\{b\in[B]: |S_{a,b}\cap S_0|=M\}|$. Conditional on $R=1$ and $X$ such that $\Omega_{\text{RCC}}$ holds, each signal coordinate $j\in S_0$ has the same probability of being included in $S_{a,b^*(a)}$, which is the unique subset of maximal intersection with $S_0$.  Thus, we have on $\Omega_{\text{RCC}}$ that
\begin{equation}
\label{Eq:EqualProb}
\mathbb{P}(\{j\in S_{a,b^*(a)}\}\cap\{R=1\}\mid X) = \mathbb{P}(\{j'\in S_{a,b^*(a)}\}\cap\{R=1\}\mid X)
\end{equation}
for $j,j'\in S_0$. 
Recall the definition of $q_j$ from the proof of Theorem~\ref{Thm:Main}. By~\eqref{Eq:EqualProb}, for any $j\in S_0$, we have on $\Omega_{\text{RCC}}$ that
\begin{align}
\label{Eq:SignalProbAltern}
q_j &\geq \mathbb{P}\bigl(\{j\in S_{a,b^*(a)}\}\cap\{R=1\}\bigm| X\bigr) = \frac{1}{k}\sum_{\tilde j \in S_0} \mathbb{E}\bigl(\mathds{1}_{\{\tilde j \in S_{a,b^*(a)}\}}\mathds{1}_{\{R = 1\}}\bigm | X\bigr) \nonumber\\
&= \frac{1}{k}\mathbb{E}\bigl(|S_{a,b^*(a)}\cap S_0|\mathds{1}_{\{R=1\}} \bigm | X\bigr)\geq \frac{t}{k}\mathbb{P}(M\geq t, R = 1) \geq \frac{t}{4k},
\end{align}
where the penultimate inequality uses Markov's inequality and the fact that the pair $(M,R)$ is independent of $X$, and the final bound follows from Lemma~\ref{Lem:UniqueMax}. Now, using~\eqref{Eq:SignalProbAltern} in place of~\eqref{Eq:SignalCoordinateProbability2}, we find that
$
\mathbb{E}(\hat w_a^{(j)}-\hat w_a^{(j')}\mid X)\geq \frac{tm\theta_m}{8k^2}
$
instead of~\eqref{Eq:Gap}. Thus, 
$
\mathbb{P}(\Omega^{\mathrm{c}}\mid X)\leq pe^{-At^2/(800k^2)}.
$
The desired result is then concluded in a similar fashion as in Theorem~\ref{Thm:Main}.
\end{proof}

\begin{lemma}
\label{Lem:Perturbation}
Suppose $\Sigma,\hat\Sigma$ are symmetric $d\times d$ matrices. For $r\in[d]$, let $\lambda_r := \lambda_r(\Sigma)$ and $v_r:=v_r(\Sigma)$ be the eigenvalues and corresponding eigenvectors of $\Sigma$, and let $\hat\lambda_r := \lambda_r(\hat\Sigma)$ and $\hat v_r:=v_r(\hat\Sigma)$ be the eigenvalues and corresponding eigenvectors of $\hat\Sigma$. Also, for $r\in[d]$, define $V_r := (v_1,\ldots,v_r)$, $\hat V_r := (\hat v_1,\ldots,\hat v_r)$, $\Theta_r := \mathrm{diag}(\lambda_1-\lambda_{r+1},\ldots,\lambda_r-\lambda_{r+1})$ and $\hat\Theta_r := \mathrm{diag}(\hat\lambda_1-\hat\lambda_{r+1},\ldots,\hat\lambda_r-\hat\lambda_{r+1})$ (with the convention that $\lambda_{d+1}= \hat\lambda_{d+1}:=0$). Then for any $m\in[d]$,
\[
\bigl\| \hat V_m \hat\Theta_m \hat V_m^\top - V_m \Theta_m V_m^\top\bigr\|_{\mathrm{op}} \leq 4m\|\hat\Sigma - \Sigma\|_{\mathrm{op}}.
\] 
\end{lemma}
\begin{proof}
By the Davis--Kahan theorem \citep[see, e.g.][Theorem~V.3.6]{StewartSun1990} and Weyl's inequality, we have for any $r\in[d]$ that
\[
\bigl(\lambda_r-\lambda_{r+1}-\|\hat\Sigma-\Sigma\|_{\mathrm{op}}\bigr) \|\sin\Theta(\hat V_r, V_r)\|_{\mathrm{op}} \leq \|\hat\Sigma-\Sigma\|_{\mathrm{op}}.
\]
After rearranging, while noting $\|\sin\Theta(\hat V_r,V_r)\|_{\mathrm{op}}\leq 1$, we obtain that
\begin{equation}
\label{Eq:DKop}
\bigl(\lambda_r-\lambda_{r+1}\bigr) \|\sin\Theta(\hat V_r, V_r)\|_{\mathrm{op}} \leq 2\|\hat\Sigma-\Sigma\|_{\mathrm{op}}.
\end{equation}

Now, we can rewrite 
\[
V_m \Theta_m V_m^\top = \sum_{r=1}^m (\lambda_r-\lambda_{m+1})v_rv_r^\top = \sum_{r=1}^m (\lambda_r - \lambda_{r+1})V_r V_r^\top,
\]
and similarly, $\hat V_m \hat\Theta_m \hat V_m^\top = \sum_{r=1}^m (\hat\lambda_r-\hat\lambda_{r+1})\hat V_r\hat V_r^\top$. Thus,
\begin{align*}
&\bigl\| \hat V_m \hat\Theta_m \hat V_m^\top - V_m \Theta_m V_m^\top\bigr\|_{\mathrm{op}} \leq \sum_{r=1}^m \bigl\|(\hat\lambda_r-\hat\lambda_{r+1})\hat V_r\hat V_r^\top - (\lambda_r - \lambda_{r+1})V_r V_r^\top\bigr\|_{\mathrm{op}}\\
&\leq \sum_{r=1}^m \Bigl\{\bigl|\hat\lambda_r - \lambda_r -(\hat\lambda_{r+1}-\lambda_{r+1})\bigr|\|\hat V_r\hat V_r^\top\|_{\mathrm{op}} + (\lambda_r-\lambda_{r+1})\|\hat V_r\hat V_r^\top - V_rV_r^\top\|_{\mathrm{op}}\Bigr\}\\
&\leq \sum_{r=1}^m \bigl\{|\hat\lambda_r - \lambda_r| + |\hat\lambda_{r+1}-\lambda_{r+1}| + (\lambda_r-\lambda_{r+1})\|\sin\Theta(\hat V_r, V_r)\|_{\mathrm{op}}\bigr\}\leq 4m\|\hat\Sigma-\Sigma\|_{\mathrm{op}},
\end{align*}
where we used Lemma~\ref{Lem:PMEval} in the penultimate inequality, and Weyl's inequality and~\eqref{Eq:DKop} in the final one.
\end{proof}
\begin{lemma}
\label{Lem:PMEval}
For $U, V\in\mathbb{O}_{d,r}$ with $r\leq d$, let $\lambda_1,\ldots,\lambda_s$ (where $s \leq r$) denote the non-zero eigenvalues of $\sin\Theta(U,V)$.  Then the non-zero eigenvalues of $UU^\top - VV^\top$ are given by $\lambda_1,\ldots,\lambda_s,-\lambda_1,\ldots,-\lambda_s$. In particular, $\|UU^\top - VV^\top\|_{\mathrm{op}} = \|\sin\Theta(U,V)\|_{\mathrm{op}}$ and $\|UU^\top - VV^\top\|_{\mathrm{F}}^2 = 2\|\sin\Theta(U,V)\|_{\mathrm{F}}^2$. 
\end{lemma}
\begin{proof}
We only need to prove the first statement. First assume $2r\leq d$. By the first part of \citet[Theorem~I.5.2]{StewartSun1990}, there exist $Q\in\mathbb{O}_d$ and $G,H\in\mathbb{O}_r$ such that
\[
U = Q\begin{pmatrix}I_r\\\mathbf{0}_{r\times r}\\\mathbf{0}_{(d-2r)\times r}\end{pmatrix} G, \qquad V = Q\begin{pmatrix}\Gamma\\\Sigma\\\mathbf{0}_{(d-2r)\times r}\end{pmatrix}H,
\]
where $\Gamma = \mathrm{diag}(\gamma_1,\ldots,\gamma_r)$, $\Sigma = \mathrm{diag}(\sigma_1,\ldots,\sigma_r)$, $0\leq \gamma_1\cdots\leq \gamma_r$, $\sigma_1\geq\cdots\geq \sigma_r\geq 0$, and $\Gamma^2 + \Sigma^2 = I_r$. Hence,
$
U^\top V = G^\top \Gamma H
$
has singular values $\gamma_1,\ldots,\gamma_r$ and $\sin\Theta(U,V) = \mathrm{diag}(\sqrt{1-\gamma_1^2},\ldots,\sqrt{1-\gamma_r^2})$ has eigenvalues $\sigma_1,\ldots,\sigma_r$.  On the other hand, we compute that
\[
Q^\top(UU^\top - VV^\top)Q = \begin{pmatrix} \Sigma^2 & -\Gamma\Sigma & \mathbf{0}\\-\Sigma\Gamma & -\Sigma^2 & \mathbf{0}\\\mathbf{0} & \mathbf{0} & \mathbf{0}\end{pmatrix},
\]
which after permuting rows and columns is a block diagonal matrix with diagonal blocks 
\[
\begin{pmatrix} \sigma_j^2 & -\sigma_j\gamma_j \\ -\sigma_j\gamma_j & -\sigma_j^2\end{pmatrix}
\]
for $j\in[r]$. Each of these diagonal blocks has eigenvalues $\pm\sigma_j$. Thus, the eigenvalues of $UU^\top - VV^\top$ are $\pm \sigma_1,\ldots,\pm\sigma_r,0,\ldots,0$.

Now, assume that $2r>d$ instead. Then by the second part of \citet[Theorem~I.5.2]{StewartSun1990}, there exist $Q\in\mathbb{O}_d$ and $G,H\in\mathbb{O}_r$ such that
\[
U = Q\begin{pmatrix}I_{d-r} & \mathbf{0}_{(d-r)\times (2r-d)}\\ \mathbf{0}_{(d-r)\times (d-r)} & \mathbf{0}_{(d-r)\times (2r-d)} \\\mathbf{0}_{(2r-d)\times (d-r)} & I_{2r-d}\end{pmatrix} G, \qquad V = Q\begin{pmatrix}\Gamma & \mathbf{0}_{(d-r)\times(2r-d)}\\\Sigma & \mathbf{0}_{(d-r)\times (2r-d)}\\\mathbf{0}_{(2r-d)\times(d-r)} & I_{2r-d}\end{pmatrix}H,
\]
where $\Gamma = \mathrm{diag}(\gamma_1,\ldots,\gamma_{d-r})$, $\Sigma = \mathrm{diag}(\sigma_1,\ldots,\sigma_{d-r})$ and $\Gamma^2 + \Sigma^2 = I_{d-r}$. We may assume $\sigma_1\geq\cdots\geq \sigma_{d-r}$. Hence, 
\[
U^\top V = G^\top \begin{pmatrix}\Gamma & \mathbf{0}\\\mathbf{0} & I_{2r-d}\end{pmatrix} H
\]
has singular values $\gamma_1,\ldots,\gamma_{d-r}, 1,\ldots,1$ and $\sin\Theta(U,V)$ has eigenvalues $\sigma_1,\ldots,\sigma_{d-r},0,\ldots,0$. On the other hand, we again have
\[
Q^\top(UU^\top - VV^\top)Q = \begin{pmatrix} \Sigma^2 & -\Gamma\Sigma & \mathbf{0}\\-\Sigma\Gamma & -\Sigma^2 & \mathbf{0}\\\mathbf{0} & \mathbf{0} & \mathbf{0}\end{pmatrix}.
\]
Thus, $UU^\top - VV^\top$ has eigenvalues $\pm\sigma_1,\ldots,\pm\sigma_{d-r}, 0,\ldots,0$ as desired.
\end{proof}

\begin{lemma}
\label{Lem:UniqueMax}
Let $Y_1,\ldots, Y_B$ be independent and identically distributed on $\mathbb{N}\cup\{0\}$ with distribution function $F$. Define $M:=\max_{b\in[B]} Y_b$ and $R:=|\{b: Y_b = M\}|$. Then for $B = \lceil 2^{-1}(1-F(t-1))^{-1}\rceil$, we have 
$
\mathbb{P}(M\geq t,\ R = 1)\geq 1/4.
$
\end{lemma}
\begin{proof}
For $m \in\mathbb{N}\cup\{0\}$, define $p_m:= \mathbb{P}(Y_1 = m)$ and $q_m := \mathbb{P}(Y_1\geq m)$. By the definition of $B$, we have $(B-1)q_t\leq 1/2\leq Bq_t$. Also, observe that
\[
\mathbb{P}(M = m,\ R=1) = B\mathbb{P}(X_1 = m)\prod_{b=2}^B \mathbb{P}(X_b < m) = Bp_m(1-q_m)^{B-1}.
\]
Therefore,
\begin{align*}
\mathbb{P}(M\geq t,\ R=1) &= \sum_{m=t}^\infty \mathbb{P}(M=m,\ R=1) = \sum_{m=t}^\infty Bp_m(1-q_m)^{B-1} \\ 
&\geq Bq_t(1-(B-1)q_t)\geq 1/4
\end{align*}
as desired.
\end{proof}

\textbf{Acknowledgements}: The research of the first and third authors was supported by an Engineering and Physical Sciences Research Council (EPSRC) grant EP/N014588/1 for the centre for Mathematical and Statistical Analysis of Multimodal Clinical Imaging.  The second and third authors were supported by EPSRC Fellowship EP/J017213/1 and EP/P031447/1, and grant RG81761 from the Leverhulme Trust.  The authors would also like to thank the Isaac Newton Institute for Mathematical Sciences for its hospitality during the programme Statistical Scalability, which was supported by EPSRC Grant Numbers LNAG/036, RG91310. We thank the anonymous reviewers for their helpful and constructive comments.


\begin{thebibliography}{99}
\bibitem[{Alon et al.(1999)}]{Alon1999}Alon, U., Barkai, N., Notterman, D.A., Gish, K., Ybarra, S., Mack, D. and Levine, A.J. (1999) Broad patterns of gene expression revealed by clustering analysis of tumor and normal colon tissues probed by oligonucleotide arrays.  \emph{PNAS}, \textbf{96}, 6745--6750.
\bibitem[{Amini and Wainwright(2009)}]{Amini2009}Amini, A. A. and Wainwright, M. J. (2009) High-dimensional analysis of semidefinite relaxations for sparse principal components.  \emph{Ann. Statist.}, \textbf{37}, 2877--2921.
\bibitem[{Cai, Ma and Wu(2013)}]{Cai2013}Cai, T. T., Ma, Z. and Wu, Y. (2013) Sparse PCA: Optimal rates and adaptive estimation.  \emph{Ann. Statist.}, \textbf{41}, 3074--3110.
\bibitem[{Cannings and Samworth(2017)}]{Cannings2017}Cannings, T. I. and Samworth, R. J. (2017) Random-projection ensemble classification. \emph{J. Roy. Statist. Soc., Ser. B (with discussion)}, \textbf{79}, 959--1035.
\bibitem[{d'Aspremont, Bach and {El Ghaoui}(2008)}]{d'Aspremont2008}d'Aspremont, A., Bach, F. and {El Ghaoui}, L. (2008) Optimal solutions for sparse principal component analysis. \emph{J. Mach. Learn. Res.}, \textbf{9}, 1269--1294.
\bibitem[{d'Aspremont et al.(2007)}]{d'Aspremont2007}d'Aspremont, A., {El Ghaoui}, L., Jordan, M., I., Lanckriet, G., R., G. (2007) A direct formulation for sparse PCA using semidefinite programming.  \emph{Adv. Neural. Inf. Process. Syst.}. \textbf{16}, 41--48.
\bibitem[{Fowler(2009)}]{Fowler2009}Fowler, J. E. (2009) Compressive-projection principal component analysis. \emph{IEEE Trans. Image Process.}, \textbf{18}, 2230--2242.
\bibitem[{Gataric, Wang and Samworth(2018)}]{GWS2017} Gataric, M., Wang, T. and Samworth, R. J. (2018) SPCAvRP: Sparse Principal Component Analysis via Random Projections.  \textsf{R} package (version 0.4), available at \texttt{https://cran.r-project.org/web/packages/SPCAvRP/index.html}. 
\bibitem[{Johnstone and Lu(2009)}]{Johnstone2009}Johnstone, I.~M. and Lu, A.~Y. (2009) On consistency and sparsity for principal components analysis in high dimensions. \emph{J. Amer. Statist. Assoc.}, \textbf{104}, 682--693.
\bibitem[{Jolliffe, Trendafilov and Uddin(2003)}]{Jolliffe2003}Jolliffe, I. T., Trendafilov, N. T. and Uddin, M. (2003) A modified principal component technique based on the LASSO.  \emph{J. Comput. Graph. Statist.}, \textbf{12}, 531--547.
\bibitem[{Ma(2013)}]{Ma2013}Ma, Z. (2013) Sparse principal component analysis and iterative thresholding. \emph{Ann. Statist.}, \textbf{41}, 772--801. 
\bibitem[{Mackey(2009)}]{Mackey2009}Mackey, L. W. (2009) Deflation Methods for Sparse PCA.  \emph{Adv. Neural. Inf. Process. Syst.}, \textbf{21}, 1017--1024.
\bibitem[{Massart(2007)}]{Massart2007}Massart, P. (2007) \emph{Concentration Inequalities and Model Selection}, Springer, Berlin.
\bibitem[{Marzetta, Tucci and Simon(2011)}]{Marzetta2011}Marzetta, T. L., Tucci, G. H. and Simon, S. H. (2011) A random matrix-theoretic approach to handling singular covariance estimates. \emph{IEEE Trans. Inf. Theory}, \textbf{57}, 6256--6271.
\bibitem[{Moghaddam, Weiss and Avidan(2006)}]{Moghaddam2006} Moghaddam, B., Weiss, Y., and Avidan, S. (2006) Spectral bounds for sparse PCA: Exact and greedy algorithms. \emph{Adv. Neural. Inf. Process. Syst.}, \textbf{18}, 915--922.
\bibitem[{Pajor(1998)}]{Pajor1998}Pajor, A. (1998) Metric entropy of the Grassmanian manifold. In \emph{Convex Geometric Analysis}. \emph{MSRI Publications.} \textbf{34}, 181--188.
\bibitem[{Paul(2007)}]{Paul2007}Paul, D. (2007) Asymptotics of sample eigenstructure for a large dimensional spiked covariance model.  \emph{Statist. Sinica}. \textbf{17}, 1617--1642.
\bibitem[{Paul and Johnstone(2012)}]{Paul2012}Paul, D. and Johnstone, I. M. (2012) Augmented sparse principal component analysis for high dimensional data. \emph{arXiv preprint}, arxiv:1202.1242v1.
\bibitem[{Pourkamali-Anaraki and Hughes(2014)}]{Anaraki2014} Pourkamali-Anaraki, F. and Hughes, S. (2014) Memory and computation efficient PCA via very sparse random projections. \emph{International Conference on Machine Learning}, \textbf{31}, 1341--1349.
\bibitem[{Qi and Hughes(2012)}]{Qi2012}Qi, H. and Hughes, S. (2012) Invariance of principal components under low-dimensional random projection of the data. \emph{IEEE International Conference on Image Processing}, \textbf{19}, 937--940.
\bibitem[{Ramey(2016)}]{R2016} Ramey, J. A. (2016) Collection of Data Sets for Classification.  \textsf{R} package, available at \texttt{https://github.com/ramhiser/datamicroarray}. 
\bibitem[{Shen and Huang(2008)}]{Shen2008}Shen, H. and Huang, J. Z. (2008) Sparse principal component analysis via regularized low rank matrix approximation. \emph{J. Multivariate Anal.}, \textbf{99}, 1015--1034.
\bibitem[{Stewart and Sun(1990)}]{StewartSun1990}Stewart, G. W. and Sun, J.-G. (1990) \emph{Matrix Perturbation Theory}. Academic Press, Inc., San Diego, California.
\bibitem[{Szarek(1982)}]{Szarek1982}Szarek, S. (1982) Nets of Grassmann manifold and orthogonal groups. \emph{Proceedings of
Research Workshop on Banach Space Theory (Iowa City, Iowa, 1981)}, 169--185.
\bibitem[{Tillmann and Pfetsch(2014)}]{TillmanPfetsch2014}Tillman, A. N. and Pfetsch, M. E. (2014) The computational complexity of the restricted isometry property, the nullspace property, and related concepts in compressed sensing.  \emph{IEEE Trans. Inform. Theory}, \textbf{60}, 1248--1259.
\bibitem[{Vu et al.(2013)}]{Vu2013a}Vu, V. Q., Cho, J., Lei, J. and Rohe, K. (2013) Fantope Projection and Selection: A near-optimal convex relaxation of sparse PCA. \emph{Adv. Neural. Inf. Process. Syst.}, \textbf{26}, 2670--2678.
\bibitem[{Vu and Lei(2013)}]{Vu2013}Vu, V. Q. and Lei, J. (2013) Minimax sparse principal subspace estimation in high dimensions.   \emph{Ann. Statist.}, \textbf{41}, 2905--2947.
\bibitem[{Wang, Berthet and Samworth(2016)}]{Wang2016}Wang, T., Berthet, Q. and Samworth, R. J. (2016) Statistical and computational trade-offs in estimation of sparse principal components. \emph{Ann. Statist.}, \textbf{44}, 1896--1930.
\bibitem[{Wang, Lu and Liu(2014)}]{Wang2014}Wang, Z., Lu, H. and Liu, H. (2014) Tighten after relax: minimax-optimal sparse PCA in polynomial time.  \emph{Adv. Neural. Inf. Process. Syst.}, \textbf{27}, 3383--3391.
\bibitem[{Welch(1947)}]{Welch1947}Welch, B. L. (1947) The generalization of ``Student's'' problem when several different population variances are involved. \emph{Biometrika}, \textbf{34}, 28--35. 
\bibitem[{Weyl(1912)}]{Weyl1912} Weyl, H. (1912) Das asymptotische Verteilungsgesetz der Eigenwerte linearer partieller Differentialgleichungen (mit einer Anwendung auf der Theorie der Hohlraumstrahlung). \emph{Math. Ann.} \textbf{71}, 441--479. 
\bibitem[{Witten, Tibshirani and Hastie(2009)}]{Witten2009}Witten, D. M., Tibshirani, R. and Hastie, T. (2009) A penalized matrix decomposition, with applications to sparse principal components and canonical correlation analysis. \emph{Biostatistics}, \textbf{10}, 515--534.
\bibitem[{Yu(1997)}]{Yu1997}Yu, B. (1997) Assouad, Fano and Le Cam.  In Pollard, D., Torgersen, E. and Yang G. L. (Eds.) \emph{Festschrift for Lucien Le Cam: Research Papers in Probability and Statistics}, 423--435. Springer, New York.
\bibitem[{Yu, Wang and Samworth(2015)}]{Yu2015}Yu, Y., Wang, T. and Samworth, R. J. (2015) A useful variant of the Davis--Kahan theorem for statisticians.  \emph{Biometrika}, \textbf{102}, 315--323.
\bibitem[{Zou, Hastie and Tibshirani(2006)}]{Zou2006}Zou, H., Hastie, T. and Tibshirani, R. (2006) Sparse Principal Components Analysis.  \emph{J. Comput. Graph. Statist.}, \textbf{15}, 265--286.
\end{thebibliography}
\end{document}